\documentclass{article}
\usepackage{arxiv}

\usepackage[T1]{fontenc}    
\usepackage{hyperref}       
\usepackage{url}            
\usepackage{booktabs}       

\usepackage[utf8]{inputenc} 
\usepackage[round]{natbib}
\usepackage{alltt} 
\usepackage{fancyvrb}

\usepackage{amsmath}
\usepackage{amsfonts}
\usepackage{amssymb}
\usepackage{amsthm}

\usepackage{listings} 

\usepackage{algorithmicx}
\usepackage{algorithm}
\usepackage{algpseudocode}

\usepackage{enumerate} 
\usepackage{mdwlist} 

\usepackage{tabularx}
\usepackage{multirow}
\usepackage{float} 
\usepackage[table]{xcolor}

\usepackage{subcaption}
\usepackage{graphicx}

\usepackage{multirow,array}
\usepackage[normalem]{ulem}

\DeclareMathOperator*{\E}{\mathbb{E}}
\newtheorem{theorem}{Theorem}[section]
\newtheorem{definition}{Definition}[section]

\title{A utility-based analysis of equilibria\\ in multi-objective normal form games \thanks{{This article extends an earlier unpublished paper \citep{Radulescu2019Equilibria} that was originally presented at the Adaptive and Learning Agents Workshop 2019 (\url{https://ala2019.vub.ac.be/papers/ALA2019_paper_29.pdf})}}}

\author{
  Roxana R\u{a}dulescu\\
  Artificial Intelligence Lab\\
  Vrije Universiteit Brussel\\
  Belgium \\
  \texttt{roxana.radulescu@vub.be} \\
   \And
 Patrick Mannion \\
  School of Computer Science\\
  National University of Ireland Galway\\
  Ireland \\
  \texttt{patrick.mannion@nuigalway.ie} \\
  \And
   Yijie Zhang \\
  Universiteit van Amsterdam\\
  The Netherlands\\
  \texttt{yijie.zhang@student.uva.nl} \\
  \And
   Diederik M. Roijers \\
   Institute of ICT \\
   HU University of Applied Sciences Utrecht \\
   The Netherlands \\
  \texttt{diederik.yamamoto-roijers@hu.nl} \\
    \And
   Ann Now\'{e}\\
  Artificial Intelligence Lab\\
  Vrije Universiteit Brussel\\
  Belgium \\
  \texttt{ann.nowe@vub.be} \\
}

\begin{document}
\maketitle

\begin{abstract}
In multi-objective multi-agent systems (MOMAS), agents explicitly consider the possible tradeoffs between conflicting objective functions. We argue that compromises between competing objectives in MOMAS should be analysed on the basis of the utility that these compromises have for the users of a system, where an agent's utility function maps their payoff vectors to scalar utility values. This utility-based approach naturally leads to two different optimisation criteria for agents in a MOMAS: expected scalarised returns (ESR) and scalarised expected returns (SER). In this article, we explore the differences between these two criteria using the framework of multi-objective normal form games (MONFGs). We demonstrate that the choice of optimisation criterion (ESR or SER) can radically alter the set of equilibria in a MONFG when non-linear utility functions are used.
\end{abstract}


\section{Introduction}

\label{sec:introduction}
Multi-agent systems (MAS) are ideally suited to model a wide range of real-world problems where autonomous actors participate in distributed decision making. Example application domains include urban and air traffic control \citep{Mannion2016Experimental,yliniemi2015simulation}, autonomous vehicles \citep{radulescudeep,Talpert2019Exploring} and energy systems \citep{Walraven16ecai,Mannion2016Multi,reymond2018reinforcement}. Although many such problems feature multiple conflicting objectives to optimise, most MAS research focuses on agents maximising their return w.r.t.\ a single objective. By contrast, in multi-objective multi-agent systems (MOMAS), agents explicitly consider the possible trade-offs between conflicting objective functions. Agents in a MOMAS receive vector-valued payoffs for their actions, where each component of a payoff vector represents the performance on a different objective. Following the utility-based approach \citep{roijers2013survey}, we assume that each agent has a utility function which maps vector-valued payoffs to scalar utility values. Compromises between competing objectives are then considered on the the basis of the utility that these trade-offs have for the users of a MOMAS.

The utility-based approach naturally leads to two different optimisation criteria for agents in a MOMAS: expected scalarised returns (ESR) and scalarised expected returns (SER). To date, the differences between the SER and ESR approaches have received little attention in multi-agent settings, despite having received some attention in single-agent settings (see e.g. \cite{roijers2013survey,roijers2018multi}). Consequently, the implications of choosing either ESR or SER as the optimisation criterion for a MOMAS are currently not well-understood. In this work, we use the framework of multi-objective normal form games (MONFGs) to explore the differences between ESR and SER in multi-agent settings.

In multi-agent systems, solution concepts such as Nash equilibria \citep{Nash1950Equilibrium,Nash1951Non} and correlated equilibria \citep{aumann1974subjectivity,aumann1987correlated} specify conditions under which each agent cannot increase its expected payoff by deviating unilaterally from an equilibrium strategy. Such solution concepts are well-studied in single objective settings, to capture stable multi-agent behaviour. However, in utility-based MOMAS the notion of an equilibrium must be redefined, as incentives to deviate from equilibrium strategies are now computed based on the relative utilities of vector-valued payoffs, rather than the relative values of scalar payoffs. Furthermore, the choice of optimisation criterion (ESR or SER) influences how equilibria are computed, as agents' incentives to deviate from an equilibrium strategy may be measured in terms of either differences in ESR or differences in SER.

The contributions of this work are:
\begin{enumerate}
    \item We provide the first comprehensive analysis of the differences between the ESR and SER optimisation criteria in multi-agent settings.
    \item We provide formal definitions of the criteria for Nash equilibria and correlated equilibria under ESR and SER. 
    \item We prove that the ESR and SER criteria are equivalent in cases where linear utility functions are used.
    \item We demonstrate that the choice of optimisation criterion radically alters the set of equilibria in an MONFG.
    \item We propose two versions of correlated equilibria for MONFGs -- single-signal and multi-signal -- corresponding to different use-cases.
    \item We prove that in MONFGs under SER {with non-linear utility functions}, Nash equilibria and multi-signal correlated equilibria need not exist, whereas single-signal correlated equilibria can exist. {We find that whether these equilibria exist in a specific MONFG depends on the multi-objective payoff structure and the utility functions used}. These examples are supported by empirical results.
    \item{{We demonstrate empirically that the well-known previous findings that CE can provide better payoffs than NE in single objective games (see e.g. \cite{aumann1974subjectivity}) can also apply in the more general class of multi-objective games, i.e., that in a MOMAS where a coordination signal can be established CE can potentially lead to higher utility for both agents than NE.}}
\end{enumerate}
The next section of this paper introduces and discusses normal form games, relevant solution concepts and optimisation criteria for multi-objective decision making. Section \ref{sec:relatedwork} provides an overview of prior work on multi-objective games. Section \ref{sec:computing} formally defines Nash and correlated equilibria in MONFGs under ESR and SER and discusses some important theoretical considerations arising from these definitions. Section \ref{sec:experiments} presents empirical results in support of the conclusions reached in Section \ref{sec:computing}.
Finally, Section \ref{sec:conclusion} concludes with a summary of our findings, a discussion of important open questions and promising directions for future work.
 
\section{Background}
\label{sec:background}

\subsection{Normal-form Games and Equilibria}

Normal-form (strategic) games (NFG) constitute a fundamental representation of interactions between players in game theory. Players are seen as rational decision-makers seeking to maximise their payoff. When multiple players are interacting, their strategies are interrelated, each decision depending on the choices of the others. For this reason, we usually try to determine interesting groups of outcomes, called solution concepts. Below we offer a formal definition for NFG and discuss two well-known solution concepts considered in this work: Nash equilibria and correlated equilibria.

\begin{definition}[Normal-form game]
\label{def:nfg}
An $n$-person finite normal-form game $G$ is a tuple $(N, \mathcal{A}, \mathbf{p})$, with $n\ge2$, where:
\begin{itemize}
    \item $N=\{1, \ldots, n\}$ is a finite set of players. 
    \item $\mathcal{A} = A_1 \times \dots \times A_n$, where $A_i$ is the finite action set of player~$i$ (i.e., the pure strategies of $i$). An \textit{action (pure strategy) profile} is a vector $\mathbf{a}=(a_1, \ldots, a_n) \in \mathcal{A}$. 
    \item $\mathbf{p}=(p_1, \ldots, p_n)$, where $p_i \colon \mathcal{A} \to \mathbb{R}$ is the real-valued payoff of player $i$, given an action profile. 
\end{itemize}
\end{definition}

\subsubsection*{Mixed-strategy profile}
Let us denote by $P(X)$ the set of all probability distributions over $X$. We can then define the set of mixed strategies of player $i$ as $\Pi_i = P(A_i)$.
The set of \textit{mixed-strategy profiles} is then the Cartesian product of all the individual mixed-strategy sets $\Pi=\Pi_1\times \ldots \times \Pi_n$.

We define $\pi_{-i} = (\pi_1,\ldots,\pi_{i-1},\pi_{i+1},\ldots,\pi_n)$ to be a strategy profile without player's $i$ strategy. We can thus write $\pi=(\pi_{i}, \pi_{-i})$.

A Nash equilibrium (NE) \citep{Nash1951Non} can be defined based on a pure or mixed-strategy profile, such that each player has selected her best response to the other players' strategies. We offer a more formal definition below.

\begin{definition}[Nash Equilibrium]
\label{def:ne}
A mixed strategy profile $\pi^{NE}$ of a game $G$ is a Nash equilibrium if for each player $i \in \{1,...,N\}$ and for any alternative strategy $\pi_i \in \Pi_i$:
\begin{equation}
\begin{split}
\E p_i(\pi_{i}^{NE}, \pi_{-i}^{NE}) \geq  \E p_i(\pi_i, \pi_{-i}^{NE}) 
\label{eqn:ne}
\end{split}
\end{equation}\end{definition}

Thus, under a Nash equilibrium, no player $i$ can improve her payoff by unilaterally changing her strategy. The same definition applies for pure-strategy profiles. \citet{Nash1951Non} has proven that, allowing the use of mixed-strategies, any finite NFG has at least one Nash equilibrium.

A correlated equilibrium is a game theoretic solution concept proposed by \citet{aumann1974subjectivity} in order to capture correlation options available to the players when some form of communication can be established prior to the action selection phase (i.e, the players receive signals from an external device, according to a known distribution, allowing them to correlate their strategies). For the current work, we look at correlation signals taking the form of action recommendations.

\subsubsection*{Correlated strategy}
A correlated strategy represents a probability vector $\sigma$ on $\mathcal{A}$, that assigns probabilities for each possible action profile, i.e., $\sigma \colon \mathcal{A} \to [0,1]$. The expected payoff of player $i$, given a correlated strategy $\sigma$ is calculated as: $$\E p_i(\sigma) = \sum_{\mathbf{a}\in \mathcal{A}}\sigma(\mathbf{a})p_i(\mathbf{a}) $$

\subsubsection*{Strategy modification}
A strategy modification for player $i$ is a function $\delta_i\colon A_{i}\to A_{i}$, such that given a recommendation $a_i$, player $i$ will play action $\delta_i(a_i)$ instead. The expected payoff of player $i$, given a correlated strategy $\sigma$ and a strategy modification $\delta_i$ is calculated as: $$\E p_i(\delta(\sigma)) = \sum_{\mathbf{a}\in \mathcal{A}}\sigma(\mathbf{a})p_i(\delta_i(a_i), a_{-i}) $$

\begin{definition}[Correlated equilibrium]
\label{def:ce}
A correlated strategy $\sigma^{CE}$ of a game G is a correlated equilibrium if for each player $i \in \{1,...,N\}$ and for any possible strategy modification $\delta_i$:
\begin{equation}
\begin{split}
\E p_i(\sigma^{CE}) \geq  \E p_i(\delta_i(\sigma^{CE})) 
\label{eqn:ce}
\end{split}
\end{equation}
\end{definition}

Thus, a correlated equilibrium ensures that no player can gain additional payoff by deviating from the suggestions, given that the other players follow them as well. Although this definition strongly resembles the one of NE, there is one important aspect we emphasise here, namely the distinction between a mixed-strategy profile and a correlated strategy. Mixed-strategy profiles are composed of independent probability factors, while the action probabilities in correlated strategies are jointly defined.

Correlated equilibria can be computed via linear programming in polynomial time \citep{papadimitriou2008computing}. It has been also shown that no-regret algorithms converge to CE \citep{foster1999regret}. Furthermore, CE has the same existence guarantees in finite NFGs \citep{hart1989existence} as NE, and any Nash equilibrium is an instance of a correlated equilibrium \citep{aumann1987correlated}. 

\subsubsection*{Example}

Consider the game of Chicken with the payoffs described in Table~\ref{table:chicken}. Each player has two actions: to continue driving towards the other player (D) or to swerve the car (S).

\begin{table}[h]
\centering
\begin{tabular}{ccc}
                       & S                         & D                         \\ \cline{2-3} 
\multicolumn{1}{c|}{S} & \multicolumn{1}{c|}{6, 6} & \multicolumn{1}{c|}{2, 7} \\ \cline{2-3} 
\multicolumn{1}{c|}{D} & \multicolumn{1}{c|}{7, 2} & \multicolumn{1}{c|}{0, 0} \\ \cline{2-3} 
\end{tabular}
\vskip 1em
\caption{Payoff matrix for the game of Chicken.}
\label{table:chicken}
\end{table}

There are three well-known Nash equilibria for this game with expected payoffs $(7,2)$, $(2,7)$ -- pure strategy NE -- and $(4\frac{2}{3}, 4\frac{2}{3})$ -- mixed strategy NE where each player selects S and D with probabilities $\frac{2}{3}$ and $\frac{1}{3}$ respectively. 

\begin{table}[h]
\centering
\begin{tabular}{ccc}
                       & S                         & D                         \\ \cline{2-3} 
\multicolumn{1}{c|}{S} & \multicolumn{1}{c|}{0.5} & \multicolumn{1}{c|}{0.25} \\ \cline{2-3} 
\multicolumn{1}{c|}{D} & \multicolumn{1}{c|}{0.25} & \multicolumn{1}{c|}{0} \\ \cline{2-3} 
\end{tabular}
\vskip 1em
\caption{A possible correlated equilibrium for the game of Chicken.}
\label{table:chicken_ce}
\end{table}
A possible correlated equilibrium is represented in Table~\ref{table:chicken_ce}, by assigning $0.5$ probability for the joint action $(S,S)$, $0.25$ for $(D,S)$ and finally $0.25$ for $(S,D)$. The expected payoff for this CE is $(5\frac{1}{4},5\frac{1}{4})$, values higher than the ones obtained under any NE. Thus, the notion of correlated equilibrium not only extends Nash equilibrium, but it also offers the potential for obtaining higher expected payoffs when players are able to receive a correlation signal (e.g., a recommended action).

\subsection{Multi-Objective Normal-Form Games}
\begin{definition}[Multi-objective normal-form game]
\label{def:monfg}
An $n$-person finite multi-objective normal-form game $G$ is a tuple $(N, \mathcal{A}, \mathbf{p})$, with $n\ge2$ and $d\ge2$ objectives, where:
\begin{itemize}
    \item $N=\{1, \ldots, n\}$ is a finite set of players. 
    \item $\mathcal{A} = A_1 \times \dots \times A_n$, where $A_i$ is the finite action set of player~$i$ (i.e., the pure strategies of $i$). An \textit{action (pure strategy) profile} is a vector $\mathbf{a}=(a_1, \ldots, a_n) \in \mathcal{A}$. 
    \item $\mathbf{p}=(\mathbf{p_1}, \ldots, \mathbf{p_n})$, where $\mathbf{p_i} \colon \mathcal{A} \to \mathbb{R}^d$ is the vectorial payoff of player $i$, given an action profile. 
\end{itemize}
\end{definition}
In this work we adopt a utility-based perspective \citep{roijers2013survey} and assume that each agent has a utility function that maps her vectorial payoff to a scalar utility value. A more detailed discussion of utility functions can be found in Section~\ref{sec:utility}.

\subsection{Multi-Objective Optimisation Criteria}
\label{sec:optimisation}
When agents consider multiple conflicting objectives, they should balance these in such a way that the user utility derived from the outcome of a decision problem (such as a MONFG) is maximised. This is known as the utility-based approach \citep{roijers2013survey}. Following this approach, we assume that there exists a utility function that maps a vector with a value for each objective to a scalar utility: 
\begin{equation}
    p_{u,i} = u_i({\bf p}_i)
    \label{eqn:utilityfunction}
\end{equation}

\noindent where $p_{u,i}$ is the utility that agent $i$ derives from the vector ${\bf p}_i$.  
When deciding what to optimise in a multi-objective normal form game, we thus need to apply this function to the vector-valued outcomes of the decision problem in some way. 
There are two choices for how to do this \citep{roijers2013survey,roijers2017modem}. Computing the expected value of the payoffs of a joint strategy first and then applying the utility function, leads to the \emph{scalarised expected returns} (SER) optimisation criterion, i.e., 
\begin{equation}
    p_{u,i} = u(\E[{\bf p}_i~|~ \boldsymbol{\pi}])
    \label{eqn:ser}
\end{equation}
\noindent where $\pi$ is the joint strategy for all the agents in a MONFG, and ${\bf p}_i$ is the payoff received by agent $i$. SER is employed in most of the multi-objective planning and reinforcement learning literature. Alternatively, the utility function can be applied before computing the expectation, leading to the \emph{expected scalarised returns} (ESR) optimisation criterion \citep{roijers2018multi}, i.e., 
\begin{equation}
    p_{u,i} = \E[u({\bf p}_i)~|~ \boldsymbol{\pi}]
    \label{eqn:esr}
\end{equation}
Which of these criteria should be considered best depends on how the games are used in practice. SER is the correct criterion if a game is played multiple times, and it is the average payoff over multiple plays that determines the user's utility. ESR is the correct formulation if the payoff of a single play is what is important to the user. 

\subsection{Utility Functions}
\label{sec:utility}

From a single-objective game theoretic perspective the notions of utility and payoff functions are generally used interchangeably. When transitioning to the multi-objective domain, we usually denote by payoff function the vectorial return (containing a real-valued payoff for each objective) received by a player, given an action profile. The utility (scalarisation) function is then used to denote the mapping from this vectorial return to a scalar utility value for a player $i$: $u_i \colon \mathbb{R}^d \to \mathbb{R}$.

Linear combinations are a widely used canonical example of a scalarisation function:
\begin{equation}
    u_{i}(\mathbf{p}_{i}) = \sum\limits_{d \in D} w_d p_{i,d}
    \label{eqn:linear_utility}
\end{equation}
\noindent where $D$ is the set of objectives, $\mathbf{w}$ is a weight vector\footnote{A vector whose coordinates are all non-negative and sum up to 1.}, $w_d \in [0,1]$ is the weight for objective $d$ and $p_{i,d}$ is the payoff for objective $d$ received by agent $i$.
Non-linear, discontinuous utility functions may arise in the case where it is important for an agent to achieve a minimum payoff on one of the objectives; such a utility function may look like the following:
\begin{equation}
        u_i(\mathbf{p}_{i}) = \begin{cases}
        p_{i,t_d} & \text{if } p_{i,d} \geq t_d\\
        0 & \text{otherwise } 
        \end{cases}
        \label{eqn:nonlinear_utility_example}
\end{equation}
\noindent where $p_{i,d}$ represents the expected payoff for agent $i$ on objective $d$, $t_d$ is the required threshold value for $d$, and $p_{i,t_d}$ is the utility to agent $i$ of reaching the threshold value on $d$.

Utility functions may not always be known \emph{a priori} and/or may not be easy to define depending on the setting. For example, in the \emph{decision support scenario} \citep{roijers2013survey} it may not be possible for users to specify utility functions directly; instead users may be asked to provide their preferences by scoring or ranking different possible outcomes. After the preference elicitation process is complete, users' responses may then be used to model their utility functions \citep{zintgraf2018ordered}. 
\section{Related Work}
\label{sec:relatedwork}
Since their introduction in \citet{blackwell1956analog}, multi-objective (multicriteria) games have been discussed extensively in the literature. Below we present a non-exhaustive overview of this work, highlighting a few differences with the current considered perspective.

Most previous work in multi-objective games considers utility-function agnostic equilibria, i.e., the agents do not know their preferences. For this case, \citet{shapley1959equilibrium} extend and characterise the set of mixed-strategy agnostic Nash equilibria for multicriteria two-person zero-sum games for linear utility functions: joint strategies that are undominated w.r.t.\ unilateral changes by either agent. They also note that if the preference functions differ, the scalarised game (implicitly assuming ESR) can possibly be no longer zero-sum. 
While the idea that utility functions could also be non-linear is discussed by \citet{bergstresser1977domination}, for analysis purposes they only consider linear utility functions and derive solution points from the resulting trade-off games. This is important because, as we will discuss in Section \ref{sec:theory}, there is no in-practice difference between ESR and SER in the linear case. 
The existence of Pareto\footnote{While the original paper refers to this type of equilibrium as ``Pareto'', we note that Pareto is a too loose domination concept when considering only linear utility functions, and would prefer ``Convex'' in this case. For consistency however, we keep the original term.} equilibria for two-person multi-objective games under linear utility functions is proven by \citet{borm1990pareto}. A further characterisation of Pareto equilibria can be found in \cite{voorneveld1999axiomatizations}. 

Considering non-cooperative games, \citet{wierzbicki1995multiple} states that, in realistic scenarios, how to aggregate criteria might not be known, however some form of scalarisation function is necessary in order to compute possible solutions. This corresponds to explicitly taking the user utility into account, and we therefore fully agree with this approach. Conflict escalation and solution selection are discussed when considering linear or order-consistent scalarisation functions.
\citet{lozovanu2005multiobjective} formulate an algorithm for finding Pareto-Nash equilibria in multi-objective non-cooperative games (i.e. for every linear utility function for which the weights sum to one, compute the trade-off game, then find its NE). Finally, \citet{lozan2013computing} propose a method for computing Pareto-Nash equilibrium sets, also under linear utility functions.
A third approach is to elicit preferences, i.e., information about the utility function, while determining equilibria \citep{igarashi2017multi}. As far as we know however, this also has only been done for linear utility functions. 

Notice that, despite the fact that many works admit that it might not always be desirable for a player to share full information about her utility function or that utility functions could take any form (including non-linear), most analysis and theoretical contributions use linear utility functions only. Furthermore, the utility function is directly applied on the original game in order to derive and analyse the corresponding trade-off game, which corresponds to the expected scalarised return (ESR) case. However, due to the use of linear utility functions, there is no distinction to be made between the ESR and SER optimisation criteria, as we will show in Section~\ref{sec:theory}. Interestingly enough, the field of multi-objective (single-agent) reinforcement learning typically focuses on the SER case \citep{van2014multi,vamplew2011empirical,zintgraf2015quality,mossalam2016multi}, while in either field this vital choice is typically not made explicitly or explained in the individual papers. In this paper, we aim to make the choice between an ESR and SER perspective explicit, and show that this choice has profound consequences in multi-objective multi-agent systems. 

For a comprehensive overview of prior work on multi-objective multi-agent decision making, the interested reader is referred to a recent survey article by \cite{Radulescu2020Survey}.

\section{Computing Equilibria in MONFGS}
\label{sec:computing}
Now that we have covered the necessary background, we begin our exploration of the differences between the ESR and SER optimisation criteria in MOMAS. In Section \ref{sec:definitions} we formally define Nash and correlated equilibria in MONFGs under either ESR or SER. In Section \ref{sec:theory} we discuss several important theoretical considerations arising from these definitions, {and introduce a new MONFG for this purpose. Section \ref{sec:additional_games} introduces some additional games, which we analyse from the SER perspective.}

\subsection{Definitions}
\label{sec:definitions}
As agents in MOMAS seek to optimise the utility of their vector-valued payoffs, rather than the value of scalar payoffs in single-objective settings, the standard solution concepts must be redefined based on the agents' utilities. Incentives to deviate from an equilibrium strategy may be defined based on utility, specifically the difference between the utility of an equilibrium action and the utilities of other possible actions. Here, we reformulate the conditions for Nash equilibria (Eqn. \ref{eqn:ne}) and correlated equilibria (Eqn. \ref{eqn:ce}) under the ESR optimisation criterion (Eqn. \ref{eqn:esr}) and the SER optimisation criterion (Eqn. \ref{eqn:ser}). 

\begin{definition}[Nash equilibrium in a MONFG under ESR]
\label{def:ne_esr}
A mixed-strategy strategy profile $\pi^{NE}$ is a Nash equilibrium in a MONFG under ESR if for all $i \in \{1,...,N\}$ and all $\pi_i \in \Pi_i$:
\begin{equation}
\begin{split}
\E u_i\big[\mathbf{p}_i(\pi_{i}^{NE}, \pi_{-i}^{NE}) \big] \geq  \E u_i\big[\mathbf{p}_i(\pi_i, \pi_{-i}^{NE}) \big]
\label{eqn:ne_esr}
\end{split}
\end{equation}
\noindent i.e. $\pi^{NE}$ is a Nash equilibrium under ESR if no agent can increase the \emph{expected utility of her payoffs} by deviating unilaterally from $\pi^{NE}$.
\end{definition}

\begin{definition}[Nash equilibrium in a MONFG under SER]
\label{def:ne_ser}
A mixed-strategy strategy profile $\pi^{NE}$ is a Nash equilibrium in a MONFG under SER if for all $i \in \{1,...,N\}$ and all $\pi_i \in \Pi_i$:
\begin{equation}
u_i \big[ \E \mathbf{p}_i(\pi_{i}^{NE}, \pi_{-i}^{NE}) \big] \geq  u_i \big[ \E \mathbf{p}_i(\pi_i, \pi_{-i}^{NE}) \big]
\label{eqn:ne_ser}
\end{equation}
\noindent i.e. $\pi^{NE}$ is a Nash equilibrium under SER if no agent can increase the \emph{utility of her expected payoffs} by deviating unilaterally from $\pi^{NE}$.
\end{definition}

\begin{definition}[Correlated equilibrium in a MONFG under ESR]
\label{def:ce_esr}
A probability vector $\sigma^{CE}$ on $\mathcal{A}$ is a correlated equilibrium in a MONFG under ESR if for all players $i \in \{1,...,N\}$ and for all
strategy modifications $\delta_i$:
\begin{equation}
    \E u_i\big[\mathbf{p}_i(\sigma^{CE}) \big] \geq  \E u_i\big[\mathbf{p}_i(\delta_i(\sigma^{CE})) \big]
    \label{eqn:ce_esr}
\end{equation}
\noindent i.e. $\sigma^{CE}$ is a correlated equilibrium under ESR if no agent can increase the \emph{expected utility of her payoffs} by deviating unilaterally from the action recommendations in $\sigma^{CE}$.
\end{definition}

When applying the SER optimisation criterion for correlated equilibrium, there are two cases we can distinguish between, due to the two expectations that CE incorporates for every player $i$. First, we can define the expected payoff given a signal $a^r_i$ due to the uncertainty about the other players' actions. Second, we can define the expected payoff given the correlated strategy (i.e., a certain probability distribution over the joint action space). Depending on where we place the utility function for taking the scalarised expectation, we distinguish between the \emph{single-signal} and \emph{multi-signal} cases.

\subsubsection*{Single-signal CE under SER}

In the case of a single-signal correlated equilibrium, we assume that the signal is only given once, and that the expected payoffs over which the utility must be computed is conditioned on the signal. Even if the MONFG is played multiple times, the signal does not change. An example of a single persistent signal in a multi-agent decision problem can be a smart-grid in which the correlation signal
corresponds to the price of electricity in a longer interval (e.g., one or more hours), and the actions of the agents are whether to perform a given task or not within a small interval (e.g., 10 min). In such cases, the utility of the other signals that might have been possible do not matter; they did not occur. Hence, the agent must maximise the utility of its expected vector-valued payoff given a single signal. Or, if the signal is not known at plan-time, for each signal separately.  

\begin{definition}[Single-signal CE in a MONFG under SER]

\label{def:ce_ser_single}
A probability vector $\sigma^{CE}$ on $\mathcal{A}$ is a single-signal correlated equilibrium in a MONFG under SER if for all players $i \in \{1,...,N\}$, given a recommended action $a_i^r$, and for any alternative action $a_i$:

\begin{equation}
    u_i\bigg[\cfrac{\sum_{a_{-i} \in \mathcal{A}_{-i}} \sigma^{CE}(a_{-i},a_i^{r}) \mathbf{p}_i(a_{-i},a_i^{r})}{\sum_{a_{-i} \in \mathcal{A}_{-i}} \sigma^{CE}(a_{-i},a_i^{r})}\bigg]  \geq 
    u_i\bigg[ \cfrac{\sum_{a_{-i} \in \mathcal{A}_{-i}} \sigma^{CE}(a_{-i},a_i^{r}) \mathbf{p}_i(a_{-i},a_i)}{\sum_{a_{-i} \in \mathcal{A}_{-i}} \sigma^{CE}(a_{-i},a_i^{r})}\bigg]
    \label{eqn:ce_SER_single}
\end{equation}

i.e. $\sigma^{CE}$ is a single-signal correlated equilibrium under SER if no agent can increase the \emph{utility of her expected payoffs} by deviating unilaterally from the given action recommendation in $\sigma^{CE}$.
\end{definition}

\subsubsection*{Multi-signal CE under SER}
The single-signal CE for MONFGs assumes that even if the MONFG is played multiple times, there will be one possible signal. Alternatively, the signal may change every time the game is played. I.e., the scalarisation is performed after marginalising over the entire correlated strategy probability distribution.

\begin{definition}[Multi-signal CE in a MONFG under SER]
\label{def:ce_ser_multi}
A probability vector $\sigma^{CE}$ on $\mathcal{A}$ is a multi-signal correlated equilibrium in a MONFG under SER if for all players $i \in \{1,...,N\}$ and for any strategy modification $\delta_i$:
\begin{equation}
u_i \big[ \E \mathbf{p}_i(\sigma^{CE}) \big] \geq  u_i \big[ \E \mathbf{p}_i(\delta_i(\sigma^{CE})) \big]
\label{eqn:ce_ser}
\end{equation}
\noindent i.e. $\sigma^{CE}$ is a multi-signal correlated equilibrium under SER if no agent can increase the \emph{utility of her expected payoffs} by deviating unilaterally from the given action recommendations in $\sigma^{CE}$.
\end{definition}

Notice that while the ESR case is equivalent to solving the CE for the corresponding single-objective trade-off game, the SER case leads to a much more complicated situation. In a general case, when no restriction is imposed on the form of the utility function, we may end up having to solve a non-linear optimisation problem.

\subsection{Theoretical Considerations}
\label{sec:theory}
\begin{theorem} 
Every finite MONFG where each agent seeks to maximise the expected utility of its payoff vectors (ESR) has at least one Nash equilibrium.

\end{theorem}

\begin{proof}
In the ESR case, any MONFG can be reduced to its corresponding single-objective trade-off game $G'$, as players will apply the utility function on their payoff vectors after every interaction. We proceed with showing how one can construct $G'$.

Consider the following finite normal-form game $G' = (N, \mathcal{A}, f)$, where $N$ and $\mathcal{A}$ are the same as in the original MONFG. According to Definition~\ref{def:nfg}, the payoff function for $G'$: $f = (f_1, \ldots, f_n)$. 

We define each component $f_i\colon \mathcal{A} \to \mathbb{R}$ as the composition between player's $i$ utility function $u_i \colon \mathbb{R}^d \to \mathbb{R}$ and her vectorial payoff function $\mathbf{p}_i \colon \mathcal{A} \to \mathbb{R}^d$: $$ f_i(a) = (u_i \circ \mathbf{p}_i)(a) = u_i(\mathbf{p}_i(a)), \forall a \in \mathcal{A}$$

Thus, in the ESR case, any MONFG is reduced to a corresponding single-objective trade-off finite NFG that can be constructed as shown above. According to the Nash equilibrium existence theorem \cite{Nash1951Non}, the resulting finite NFG $G'$ has at least one Nash equilibrium. 
\end{proof}

\begin{theorem}
In finite MONFGs, when linear utility functions are used, the ESR and SER optimisation criteria are equivalent.\footnote{As is the case for single-agent decision problems \citep{roijers2013survey,roijersPhD}.}
\end{theorem}
\begin{proof} 
Let $\pi^{NE}$ be the NE strategy profile under the ESR optimisation criteria, according to Definition~\ref{def:ne_esr} and for each player $i$ let $u_i$ be a linear scalarisation function, according to Equation~\ref{eqn:linear_utility}.

Due to the fact that $u_i$ is a linear function,  Jensen's inequality \cite{jensen1906fonctions} allows us to rewrite each term of Equation~\ref{eqn:ne_esr} as follows: 
\begin{equation}
\begin{split}
\E u_i\big[\mathbf{p}_i(\pi_{i}^{NE} \cup \pi_{-i}^{NE}) \big] = 
u_i \big[ \E \mathbf{p}_i(\pi_{i}^{NE} \cup \pi_{-i}^{NE}) \big] 
\end{split}
\label{eqn:ne_linear1}
\end{equation}
\begin{equation}
\begin{split}
\E u_i\big[\mathbf{p}_i(\pi_i \cup \pi_{-i}^{NE}) \big] = u_i \big[ \E \mathbf{p}_i(\pi_i \cup \pi_{-i}^{NE}) \big]
\end{split}
\label{eqn:ne_linear2}
\end{equation}
Notice that by replacing the terms from Equation~\ref{eqn:ne_esr} according to Equations~\ref{eqn:ne_linear1} and \ref{eqn:ne_linear2} we obtain the definition of the NE under SER (Equation~\ref{eqn:ne_ser}). 
The same procedure can be applied for CE, to transition from Equation~\ref{eqn:ce_esr} to \ref{eqn:ce_ser} and prove that, under a linear utility function, the ESR and SER criteria are also equivalent for CE. 
\end{proof}

When considering a more general case, with $u_i$ being a non-linear function, despite the fact that Jensen's inequality \citep{jensen1906fonctions} would allow us to define inequality relations between the terms in Equations~\ref{eqn:ne_linear1} and \ref{eqn:ne_linear2} (when constraining $u_i$ to be convex or concave), we have no guarantee that the set of NE and CE remains the same under the two optimisation criteria ESR and SER. Thus, no clear conclusions can be drawn when generalising the form of the utility function. Furthermore, as we show below using a concrete example, in the general case, the ESR and SER criteria are not equivalent.

\begin{theorem}
In finite MONFGs, where each agent seeks to maximise the utility of its expected payoff vectors (SER), Nash equilibria need not exist.
\label{th:ser_nash}
\end{theorem}

\begin{proof} 
Consider the following game. There are two agents that can each choose from three actions: {\it left}, {\it middle}, or {\it right}. The payoff vectors are identical for both agents, and are specified by the payoff matrix in Table \ref{table:balance}. 

The utility functions of the agents are given by $u_1( [p^1,p^2] ) = p^1 \cdot p^1 +  p^2 \cdot p^2$ for agent 1, and $u_2( [p^1,p^2] ) = p^1 \cdot p^2$ for agent 2.\footnote{Please note that this is a monotonically increasing payoff function for positive-only payoffs. In the case of negative payoffs we can set the utility to $0$ as soon as the payoff value for one of the objectives becomes negative.}
\begin{table}[ht!]
  \centering
    \setlength{\extrarowheight}{2pt}
    \begin{tabular}{cc|c|c|c|}
      & \multicolumn{1}{c}{} & \multicolumn{1}{c}{$L$}  & \multicolumn{1}{c}{$M$} & \multicolumn{1}{c}{$R$} \\\cline{3-5}
      \multirow{2}*{}  & $L$ & $(4,0)$ & $(3,1)$ & $(2,2)$ \\\cline{3-5}
      & $M$ & $(3,1)$ & $(2,2)$ & $(1,3)$ \\\cline{3-5}
      & $R$ & $(2,2)$ & $(1,3)$ & $(0,4)$ \\\cline{3-5}
    \end{tabular}
    \vskip 1em
    \caption{The (Im)balancing act game.}
    \label{table:balance}
  \end{table}
In this game, it is easy to see that agent 1 will always want to move towards an as imbalanced payoff vector as possible, i.e., concentrate as much of the value in one objective, while agent 2 will always want to move to a balanced solution, i.e, spread out the value across the objectives equally. Under SER, the expectation is taken before the utility function is applied. Therefore, a mixed strategy will lead to an expected payoff vector for both agents. If the expected payoff vector is balanced, i.e., $[2,2]$, agent 1 will have an incentive to deterministically take action $L$ or $R$, irrespective of its current strategy. If the payoff vector is imbalanced, e.g., $[2-x, 2+x]$, agent 2 will have an incentive to compensate for this imbalance, and play \textit{left} more often to compensate if $x$ is positive, and \textit{right} more often if $x$ is negative, and he is always able to do so. Hence, at least one of the agents will always have an incentive to deviate from its strategy, and therefore there is no Nash equilibrium under SER.
\end{proof}

\begin{table}[ht]
  \centering
    \setlength{\extrarowheight}{2pt}
    \begin{tabular}{cc|c|c|c|}
      & \multicolumn{1}{c}{} & \multicolumn{1}{c}{$L$}  & \multicolumn{1}{c}{$M$} & \multicolumn{1}{c}{$R$} \\\cline{3-5}
      \multirow{2}*{}  & $L$ & $(16,0)$ & $(10,3)$ & $(8,4)$ \\\cline{3-5}
      & $M$ & $(10,3)$ & $(8,4)$ & $(10,3)$ \\\cline{3-5}
      & $R$ & $(8,4)$ & $(10,3)$ & $(16,0)$ \\\cline{3-5}
    \end{tabular}
    \vskip 1em
    \caption{The (Im)balancing act game under ESR with utility functions $u_1(\mathbf{p}) = p^1 \cdot p^1 +  p^2 \cdot p^2$ and $u_2(\mathbf{p}) = p^1 \cdot p^2$ applied.}
    \label{table:balance_esr}
  \end{table}

We also note that under ESR there is a mixed Nash equilibrium for the game in Table \ref{table:balance}, i.e., agent 2 plays {\it middle} deterministically, and agent 1 plays {\it left} with a probability $0.5$ and {\it right} with a probability $0.5$, leading to an expected utility of $3^2+1^2 = 10$ for agent 1, and $3 \cdot 1 = 3$ for agent 2. This is not a Nash equilibrium under SER, as the expected payoff vector is $[2,2]$ for this strategy, and agent 1 has an incentive to play either  {\it left} or  {\it right} deterministically, which would lead to an expected payoff vector of $[3,1]$ or $[1,3]$, yielding a higher utility for agent 1 if agent 2 does not adjust its strategy. Hence, the SER and ESR case are fundamentally different. 

\begin{theorem}
In finite MONFGs, where each agent seeks to maximise the utility of its expected payoff vectors given a signal (single-signal CE under SER), correlated equilibria can exist when Nash equilibria do not.
\label{th:ser_ce_single}
\end{theorem}

\begin{proof} 
Consider the action suggestions in Table~\ref{table:balance_ser_ce1} for the (Im)balancing act game. 

\begin{table}[h!]
  \centering
    \setlength{\extrarowheight}{2pt}
    \begin{tabular}{cc|c|c|c|}
      & \multicolumn{1}{c}{} & \multicolumn{1}{c}{$L$}  & \multicolumn{1}{c}{$M$} & \multicolumn{1}{c}{$R$} \\\cline{3-5}
      \multirow{2}*{}  & $L$ & $0$ & $0.75$ & $0$ \\\cline{3-5}
      & $M$ & $0$ & $0$ & $0$ \\\cline{3-5}
      & $R$ & $0$ & $0.25$ & $0$ \\\cline{3-5}
    \end{tabular}
    \vskip 1em
    \caption{A correlated equilibrium in the (Im)balancing act game under SER.}
    \label{table:balance_ser_ce1}
  \end{table}
It may easily be shown that the action suggestions in Table \ref{table:balance_ser_ce1} satisfy the conditions given in Eqn. \ref{eqn:ce_SER_single} for a single-signal CE in a MONFG under SER:
\begin{itemize}
\item When L is suggested to the row player, the expected payoff vectors and SER for it to play L, M or R are: 
        \begin{itemize}
            \item L: $\E(\mathbf{p}) = (0.75 \cdot [3,1])/0.75 = [3,1]$, SER $=3^2 + 1^2 = 10$ 
            \item M: $\E(\mathbf{p}) = (0.75 \cdot [2,2])/0.75 = [2,2]$, SER $=2^2 + 2^2 = 8$
            \item R: $\E(\mathbf{p}) = (0.75 \cdot [1,3])/0.75 = [1,3]$, SER $=1^2 + 3^2 = 10$
        \end{itemize}
    \item When R is suggested to the row player, the expected payoff vectors and SER for it to play L, M or R are:
        \begin{itemize}
            \item L: $\E(\mathbf{p}) = (0.25 \cdot [3,1])/0.25 = [3,1]$, SER $=3^2 + 1^2 = 10$
            \item M: $\E(\mathbf{p}) = (0.25 \cdot [2,2])/0.25 = [2,2]$, SER $=2^2 + 2^2 = 8$
            \item R: $\E(\mathbf{p}) = (0.25 \cdot [1,3])/0.25 = [1,3]$, SER $=1^2 + 3^2 = 10$
        \end{itemize}
    \item When M is suggested to the column player, the expected payoff vectors and SER for it to play L, M or R are:
    \begin{itemize}
        \item L: $\E(\mathbf{p}) = (0.75 \cdot [4,0] + 0.25 \cdot [2,2])/(0.75+0.25) = [3.5,0.5]$, SER $=3.5 \cdot 0.5 = 1.75$
        \item M: $\E(\mathbf{p}) = (0.75 \cdot [3,1] + 0.25 \cdot [1,3])/(0.75+0.25) = [2.5,1.5]$, SER $=2.5 \cdot 1.5 = 3.75$
        \item R: $\E(\mathbf{p}) = (0.75 \cdot [2,2] + 0.25 \cdot [0,4])/(0.75+0.25) = [1.5,2.5]$, SER $=1.5 \cdot 2.5 = 3.75$
    \end{itemize}
\end{itemize}

In all the cases above, neither of the agents may increase the utility of their expected payoff vectors given the recommendations, by deviating from the suggested actions in Table \ref{table:balance_ser_ce1}, assuming that the other agent follows the suggestions. Therefore CE may exist in MONFGs under SER when conditioning the expectation on a given signal, even in cases where Nash equlilibria do not exist.
\end{proof}

\begin{theorem}
In finite MONFGs, where each agent seeks to maximise the utility of its expected payoff vectors over all the given signals (multi-signal CE under SER), correlated equilibria need not exist.
\label{th:ser_ce}
\end{theorem}

\begin{proof} 
In the case of a multi-signal CE, the agents are interested in their expected payoff vectors across all possible signals. In other words, to compute the expected payoff vectors, the signal must be marginalised out first. Therefore, the CE previously discussed for the single-signal case (Table~\ref{table:balance_ser_ce1}) is not a CE for the multi-signal case, i.e., 

Player 1 will have an incentive to deterministically take action $L$ or $R$, irrespective of the given signal. If the correlated strategy tries to incorporate this tendency, player 2 will have an incentive to deviate towards the options that offer her the most balanced outcome. Hence, similar to the proof for the non-existence of Nash-equilibria under SER, at least one of the agents will always have an incentive to deviate from the given recommendation, and therefore there is no multi-signal correlated equilibrium under SER. 
\end{proof}

We thus conclude that an MONFG under ESR with \textit{known} utility functions is equivalent to a single-objective NFG, and therefore all theory, including the existence of Nash equilibria and correlated equilibria, is implied. Under SER however, Nash equilibria and multi-signal correlated equilibria need not exist, and MONFGs are fundamentally more difficult than single-objective NFGs, even when the utility functions are known in advance.

\subsection{Additional Games for SER Analysis}
\label{sec:additional_games}
To further investigate the existence of Nash, single- and multi-signal correlated equilibria under scalarised expected returns (SER), we introduce two additional games that demonstrate different characteristics under these criteria. We consider for analysis the same non-linear utility functions as above: $u_1( [p^1,p^2] ) = p^1 \cdot p^1 +  p^2 \cdot p^2$ for player 1, and $u_2( [p^1,p^2] ) = p^1 \cdot p^2$ for player 2.

\subsubsection{The (Im)balancing act game without action M}
\begin{table}[h!]
\centering
\begin{tabular}{ccc}
 & L & R \\ \cline{2-3} 
\multicolumn{1}{c|}{L} & \multicolumn{1}{c|}{(4, 0)} & \multicolumn{1}{c|}{(2, 2)} \\ \cline{2-3} 
\multicolumn{1}{l|}{R} & \multicolumn{1}{l|}{(2, 2)} & \multicolumn{1}{l|}{(0, 4)} \\ \cline{2-3} 
\end{tabular}
\quad\quad
\begin{tabular}{ccc}
 & L & R \\ \cline{2-3} 
\multicolumn{1}{c|}{L} & \multicolumn{1}{c|}{0.25} & \multicolumn{1}{c|}{0.25} \\ \cline{2-3} 
\multicolumn{1}{l|}{R} & \multicolumn{1}{l|}{0.25} & \multicolumn{1}{l|}{0.25} \\ \cline{2-3} 
\end{tabular}
    \vskip 1em
    \caption{The (Im)balancing act game without action M (left), together with the corresponding correlated strategy (right).}
\label{table:balance_noM}
\end{table}

First, we derive a 2-player, 2-action, 2-objective game from the (Im)balancing act game (Table~\ref{table:balance}), by removing the middle action. The (Im)balancing act game without action M is presented in Table~\ref{table:balance_noM} (left).

Notice that in the case of NE and multi-signal CE, the dynamics of the game remain unchanged from the original 3-action version: player 1 will always have an incentive to deviate towards imbalanced payoffs, while player 2 desires the exact opposite. For the single-signal CE, we have the opportunity to offer the agents the chance to coordinate their actions and obtain a fair outcome over the two possible situations (i.e., a balanced outcome (2, 2) and an imbalanced outcome (4, 0) or (0, 4)), as shown in the right side of Table~\ref{table:balance_noM}. 

It may be shown that the correlated strategy in Table~\ref{table:balance_noM} (right) satisfies the conditions given in Eqn. \ref{eqn:ce_SER_single} for a single-signal CE in a MONFG under SER:
\begin{itemize}
    \item When L is suggested to the row player, the expected payoff vectors and SER for it to play L or R are: 
        \begin{itemize}
            \item L: $\E(\mathbf{p}) = (0.25 \cdot [4,0] + 0.25 \cdot [2,2])/(0.25+0.25) = [3,1]$, SER $=3^2 + 1^2 = 10$ 
            \item R: $\E(\mathbf{p}) = (0.25 \cdot [4,0] + 0.25 \cdot [2,2])/(0.25+0.25) = [3,1]$, SER $=3^2 + 1^2 = 10$ 
        \end{itemize}
    \item When R is suggested to the row player, the expected payoff vectors and SER for it to play L or R are:
        \begin{itemize}
            \item L: $\E(\mathbf{p}) =  (0.25 \cdot [4,0] + 0.25 \cdot [2,2])/(0.25+0.25) = [3,1]$, SER $=3^2 + 1^2 = 10$ 
            \item R: $\E(\mathbf{p}) =  (0.25 \cdot [4,0] + 0.25 \cdot [2,2])/(0.25+0.25) = [3,1]$, SER $=3^2 + 1^2 = 10$ 
        \end{itemize}
    \item When L is suggested to the column player, the expected payoff vectors and SER for it to play Lor R are:
    \begin{itemize}
        \item L: $\E(\mathbf{p}) = (0.25 \cdot [4,0] + 0.25 \cdot [2,2])/(0.25+0.25) = [3,1]$, SER $=3 \cdot 1 = 3$
        \item R: $\E(\mathbf{p}) = (0.25 \cdot [4,0] + 0.25 \cdot [2,2])/(0.25+0.25) = [3,1]$, SER $=3 \cdot 1 = 3$
    \end{itemize}
        \item When R is suggested to the column player, the expected payoff vectors and SER for it to play Lor R are:
    \begin{itemize}
        \item L: $\E(\mathbf{p}) = (0.25 \cdot [4,0] + 0.25 \cdot [2,2])/(0.25+0.25) = [3,1]$, SER $=3 \cdot 1 = 3$
        \item R: $\E(\mathbf{p}) = (0.25 \cdot [4,0] + 0.25 \cdot [2,2])/(0.25+0.25) = [3,1]$, SER $=3 \cdot 1 = 3$
    \end{itemize}
\end{itemize}

In all the cases above, neither of the agents may increase the utility of their expected payoff vectors given the recommendations, by deviating from the suggested actions, assuming that the other agent follows the suggestions. Therefore the signal suggested in the right side of Table~\ref{table:balance_noM} represents a single-signal correlated equilibrium for the (Im)balancing act game without action M. 

\subsubsection{A 3-action MONFG with NE and CE under SER} 
\label{sec:NEgame}
\begin{table}[h!]
  \centering
    \setlength{\extrarowheight}{2pt}
    \begin{tabular}{cc|c|c|c|}
      & \multicolumn{1}{c}{} & \multicolumn{1}{c}{$L$}  & \multicolumn{1}{c}{$M$} & \multicolumn{1}{c}{$R$} \\\cline{3-5}
      \multirow{2}*{}  & $L$ & $(4,1)$ & $(1,2)$ & $(2,1)$ \\\cline{3-5}
      & $M$ & $(3,1)$ & $(3,2)$ & $(1,2)$ \\\cline{3-5}
      & $R$ & $(1,2)$ & $(2,1)$ & $(1,3)$ \\\cline{3-5}
    \end{tabular}
    \quad
        \begin{tabular}{cc|c|c|c|}
      & \multicolumn{1}{c}{} & \multicolumn{1}{c}{$L$}  & \multicolumn{1}{c}{$M$} & \multicolumn{1}{c}{$R$} \\\cline{3-5}
      \multirow{2}*{}  & $L$ & $0.5$ & $0$ & $0$ \\\cline{3-5}
      & $M$ & $0$ & $0.5$ & $0$ \\\cline{3-5}
      & $R$ & $0$ & $0$ & $0$ \\\cline{3-5}
    \end{tabular}
        \vskip 1em
    \caption{A 3-action MONFG which has 3 pure strategy NE (left) -- (L,L), (M,M) and (R,R) -- when the row player uses utility function $u_1( [p^1,p^2] ) = p^1 \cdot p^1 +  p^2 \cdot p^2$ and the column player uses utility function $u_2( [p^1,p^2] ) = p^1 \cdot p^2$, with the corresponding proposed correlated strategy (right).}
    \label{table:MONFG_with_NE}
  \end{table}
  
  The final game we introduce for this work presents an example of a MONFG for which all the studied equilibria (i.e., NE, single- and multi-signal CE) exist under SER (Table~\ref{table:MONFG_with_NE}). There are 3 pure-strategy NE -- (L,L), (M,M) and (R,R), under the non-linear utility functions specified above.
Notice that player 1 will receive the highest SER under (L, L), while player 2 will prefer the (M, M) outcome. (R, R) is also a NE, but it is Pareto dominated by (L,L) and (M,M) and does not offer the best possible SER for either agent.

  Let us turn our attention to the single-signal CE. It may be shown that the correlated strategy proposed in Table \ref{table:MONFG_with_NE} (right) satisfies the conditions given in Eqn. \ref{eqn:ce_SER_single} for a single-signal CE in a MONFG under SER:
\begin{itemize}
    \item When L is suggested to the row player, the expected payoff vectors and SER for it to play L, M or R are: 
        \begin{itemize}
            \item L: $\E(\mathbf{p}) = (0.5 \cdot [4,1] )/0.5 = [4,1]$, SER $=4^2 + 1^2 = 17$ 
            \item M: $\E(\mathbf{p}) = (0.5 \cdot [3,1])/0.5 = [3,1]$, SER $=3^2 + 1^2 = 10$
            \item R: $\E(\mathbf{p}) = (0.5 \cdot [1,2])/0.5 = [1,2]$, SER $=1^2 + 2^2 = 5$
        \end{itemize}
    \item When M is suggested to the row player, the expected payoff vectors and SER for it to play L, M or R are:
        \begin{itemize}
            \item L: $\E(\mathbf{p}) = (0.5 \cdot [1,2])/0.5 = [1,2]$, SER $=1^2 + 2^2 = 5$
            \item M: $\E(\mathbf{p}) = (0.5 \cdot [3,2])/0.5 = [3,2]$, SER $=3^2 + 2^2 = 13$
            \item R: $\E(\mathbf{p}) = (0.5 \cdot [2,1])/0.5 = [2,1]$, SER $=2^2 + 1^2 = 5$
        \end{itemize}
    \item When L is suggested to the column player, the expected payoff vectors and SER for it to play L, M or R are:
    \begin{itemize}
        \item L: $\E(\mathbf{p}) = (0.5 \cdot [4,1])/0.5 = [4,1]$, SER $=4 \cdot 1 = 4$
        \item M: $\E(\mathbf{p}) = (0.5 \cdot [1,2])/(0.5 = [1,2]$, SER $=1 \cdot 2 = 2$
        \item R: $\E(\mathbf{p}) = (0.5 \cdot [2,1])/0.5 = [2,1]$, SER $=2 \cdot 1 = 2$
    \end{itemize}
    \item When M is suggested to the column player, the expected payoff vectors and SER for it to play L, M or R are:
    \begin{itemize}
        \item L: $\E(\mathbf{p}) = (0.5 \cdot [3,1])/0.5 = [3,1]$, SER $=3 \cdot 1 = 3$
        \item M: $\E(\mathbf{p}) = (0.5 \cdot [3,2])/0.5 = [3,2]$, SER $=3 \cdot 2 = 6$
        \item R: $\E(\mathbf{p}) = (0.5 \cdot [1,2])/0.5 = [1,2]$, SER $=1 \cdot 2 = 2$
    \end{itemize}
\end{itemize}

In all the cases above, neither of the agents may increase the utility of their expected payoff vectors given the recommendations, by deviating from the suggested actions, assuming that the other agent follows the suggestions. Therefore the signal suggested in the right side of Table~\ref{table:MONFG_with_NE} represents a single-signal correlated equilibria.

In single-objective normal form games, it is known that any convex combination of Nash equilibrium payoff profiles can be reached or achieved by a correlated equilibrium \citep{aumann1974subjectivity}. The relationship between Nash and correlated equilibria in multi-objective normal form games remains, however, an open question. In Section~\ref{sec:experiments} we empirically test whether the proposed correlated strategy, representing a convex combination between 2 pure NE under SER, is a multi-signal correlated equilibrium as well. We also note that in the single-objective case CE can achieve payoffs that lie outside the convex hull of NE payoffs, again a property not validated in the case of MONFGs.

\section{Experiments}
\label{sec:experiments}
To demonstrate the effect of the SER optimisation criterion on equilibria in MONFGs, with no action recommendations and in the case of a \emph{single- and multi-signal correlated equilibrium}, we conducted a series of experiments using the games introduced in the previous section in Tables~\ref{table:balance}, \ref{table:balance_noM} and \ref{table:MONFG_with_NE}. All experiments were repeated 100 times and had a duration of 10,000 episodes, where the MONFG game was played once per episode. 

Agents implemented a simple algorithm\footnote{We note that specialised algorithms exist to learn mixed-strategy Nash equilibria (e.g. \cite{fudenberg1993learning}) or correlated equilibria (e.g. \cite{arifovic2016learning}) in single-objective MAS. We leave the design and empirical evaluation of versions of these algorithms for learning or approximating equilibria in MOMAS under SER for future work.} to learn estimates of the expected vectors for each action according to the following update rule (i.e. a ``one-shot'' vectorial version of Q-learning \citep{Watkins89}):
\begin{equation}
    \mathbf{Q}(s_i,a_i) \leftarrow \mathbf{Q}(s_i,a_i) + \alpha [\mathbf{p}_{i}(s_i,a_i) - \mathbf{Q}(s_i,a_i)]
    \label{eqn:ser_learning}
\end{equation}
\noindent where $\mathbf{Q}(s_i,a_i)$ is an estimate of the expected value vector for selecting action $a_i$ when a private signal $s_i$ is received, $\mathbf{p}_{i}(s_i,a_i)$ is the payoff vector received by agent $i$ for selecting action $a_i$ when observing $s_i$, and $\alpha$ is the learning rate.

The private signals given to each agent allow us to test empirically whether agents will have an incentive to deviate from a single- or multi-signal correlate equilibrium in a MONFG under SER. For the experiments marked as ``No action recommendations'', in each episode agents received unchanging private signals with probability 1 (i.e. equivalent to the case where no private signals are present). Otherwise, the private signals received by each agent corresponded to the correlated action recommendations indicated for each considered MONFG. When signals were given, for the first 500 episodes, both agents followed the action recommendations in their private signals deterministically, so that the correlated equilibrium behaviour could be learned. For the last 9,500 episodes, agents continued to receive action recommendations, but selected their actions autonomously.

Agents implemented the $\epsilon$-greedy exploration strategy. As agents seek to optimise their action choices with respect to scalarised expected returns, they will determine the optimal mixed strategy (given the recommendation, where applicable), with probability $1-\epsilon$, or chose a random action with probability $\epsilon$. Each agents determines their optimal mixed strategy by solving a non-linear optimisation problem with the goal of maximising their scalarised expected returns, under their utility function and current Q-values\footnote{This non-linear optimisation problem is solved using the ``optimize'' module of the Scipy Python package \citep{2019arXiv190710121V}}. For all the experiments, the estimates of expected value vectors for each action were scalarised using the same utility functions as in Section \ref{sec:theory}. In the case of the single-signal CE, this expectation is taken under the given action recommendation, while for the multi-signal CE, the expectation is derived with respect to the entire CE signal the agent received, following Definitions~\ref{def:ce_ser_single} and \ref{def:ce_ser_multi}, respectively. This also implies that for the multi-signal correlated equilibria, each agent has information regarding the CE distribution over her own actions, but not over the entire joint-action space. For example, in the case of the (Im)balancing Act Game, player~1 knows that the CE distribution over her actions is $[0.75, 0, 0.25]$, but is not aware that player~2 will be recommended action ``M'' with probability 1, leaving this information to be acquired through the learning process.

All agents used a constant value of $\alpha=0.05$ for the learning rate. For the experiments without action recommendations, $\epsilon$ was initially set to $0.1$ in the first episode, and decayed by a factor $0.999$ in each subsequent episode. For the experiments where agents receive action recommendations, $\epsilon$ was set to $0.0$ in for the first $500$ episodes where the agents deterministically followed the recommendations from their private signals, after which $\epsilon$ was set to $0.1$ for episode $501$ and decayed by a factor $0.999$ in each subsequent episode. No attempt was made to conduct comprehensive parameter sweeps to optimise the values of $\alpha$ and $\epsilon$ which were used in either experiment.

\begin{figure}[ht!]
    \centering
    \begin{subfigure}[b]{0.32\textwidth}
        \includegraphics[width=\textwidth]{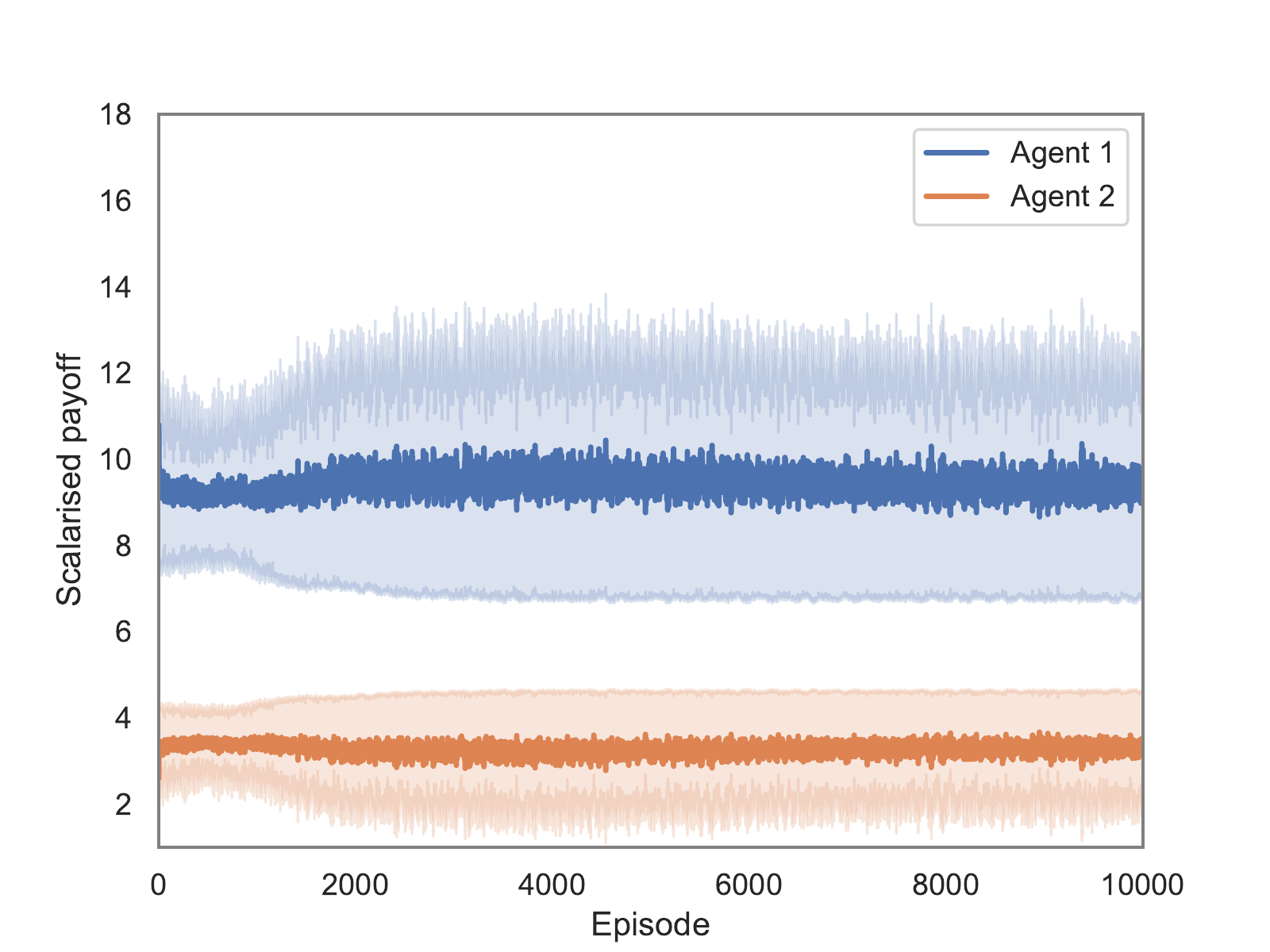}
    \caption{Scalarised payoffs obtained by each agent.}
    \label{fig:game1_NE_SER}
    \end{subfigure}
    \vspace{\baselineskip}
    \begin{subfigure}[b]{0.32\textwidth}
       \includegraphics[width=\textwidth]{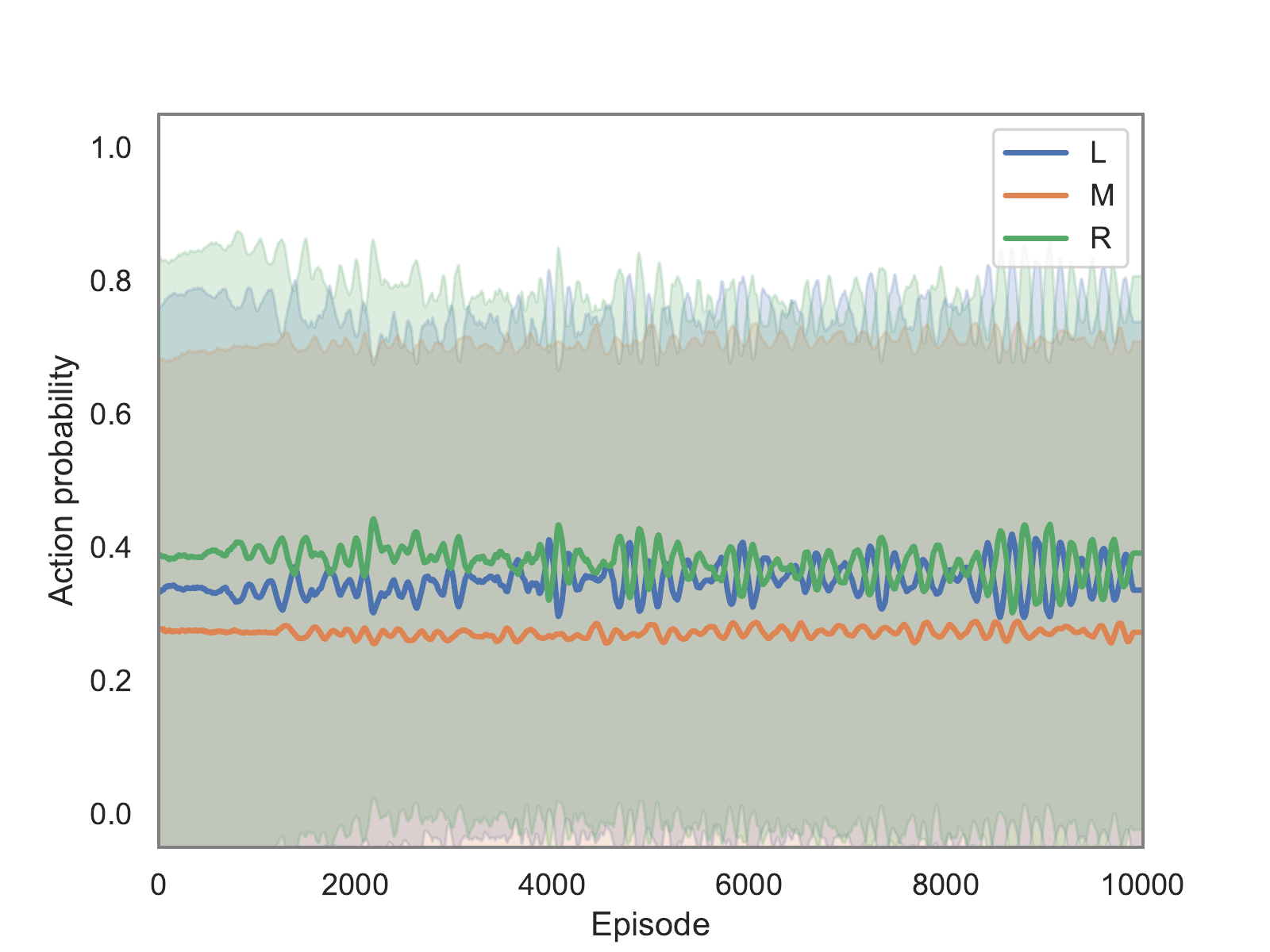}
    \caption[width=\textwidth]{Action selection probabilities of Agent 1.}
    \label{fig:game1_NE_A1}
    \end{subfigure}
    \begin{subfigure}[b]{0.32\textwidth}
        \includegraphics[width=\textwidth]{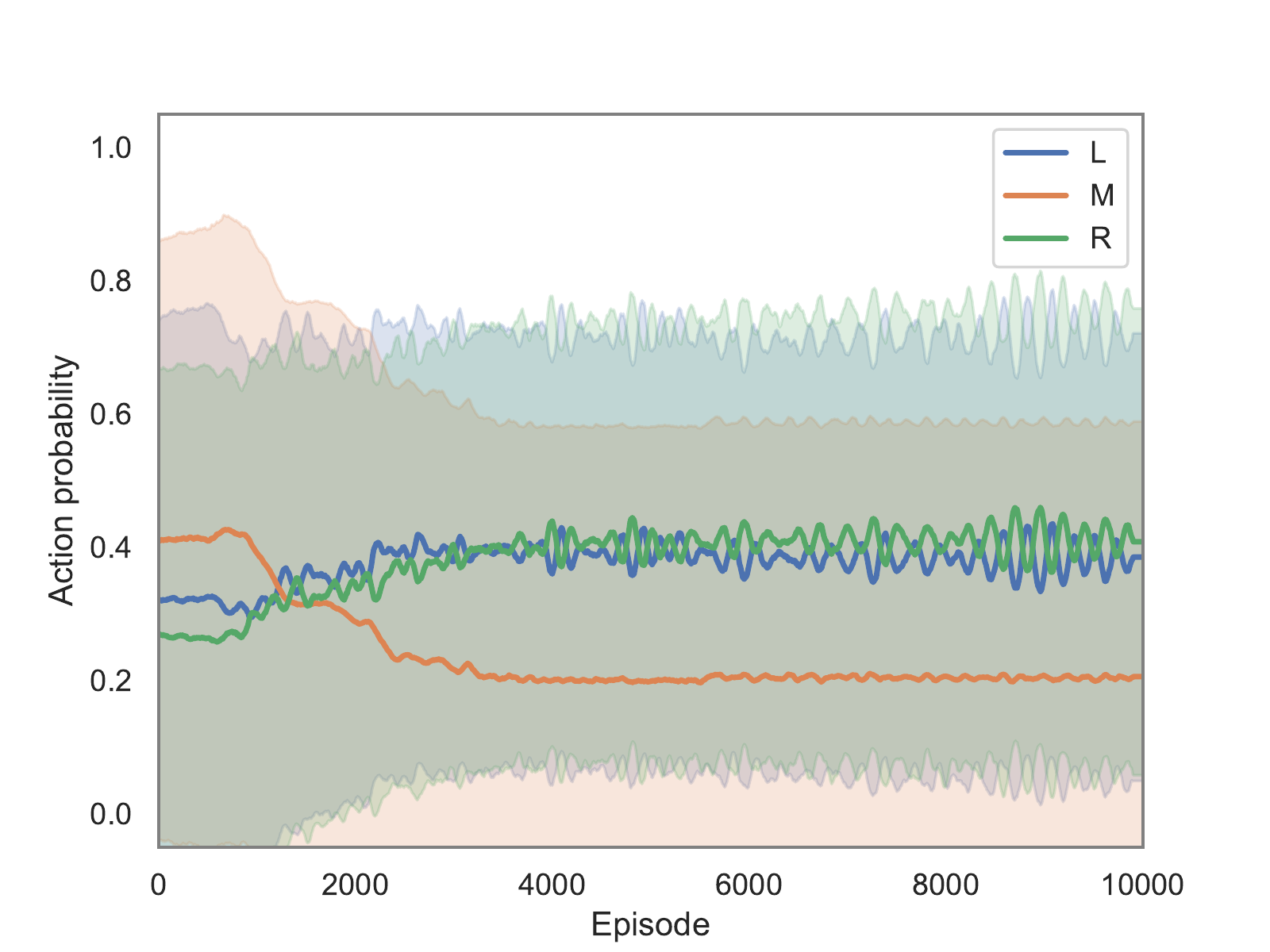}
    \caption{Action selection probabilities of Agent 2.}
    \label{fig:game1_NE_A2}
    \end{subfigure}
    \caption{Game 1 under SER with no action recommendations.}
    \label{fig:game1_NE}
\end{figure}
\begin{figure}[ht!]
    \centering
    \begin{subfigure}[b]{0.32\textwidth}
        \includegraphics[width=\textwidth]{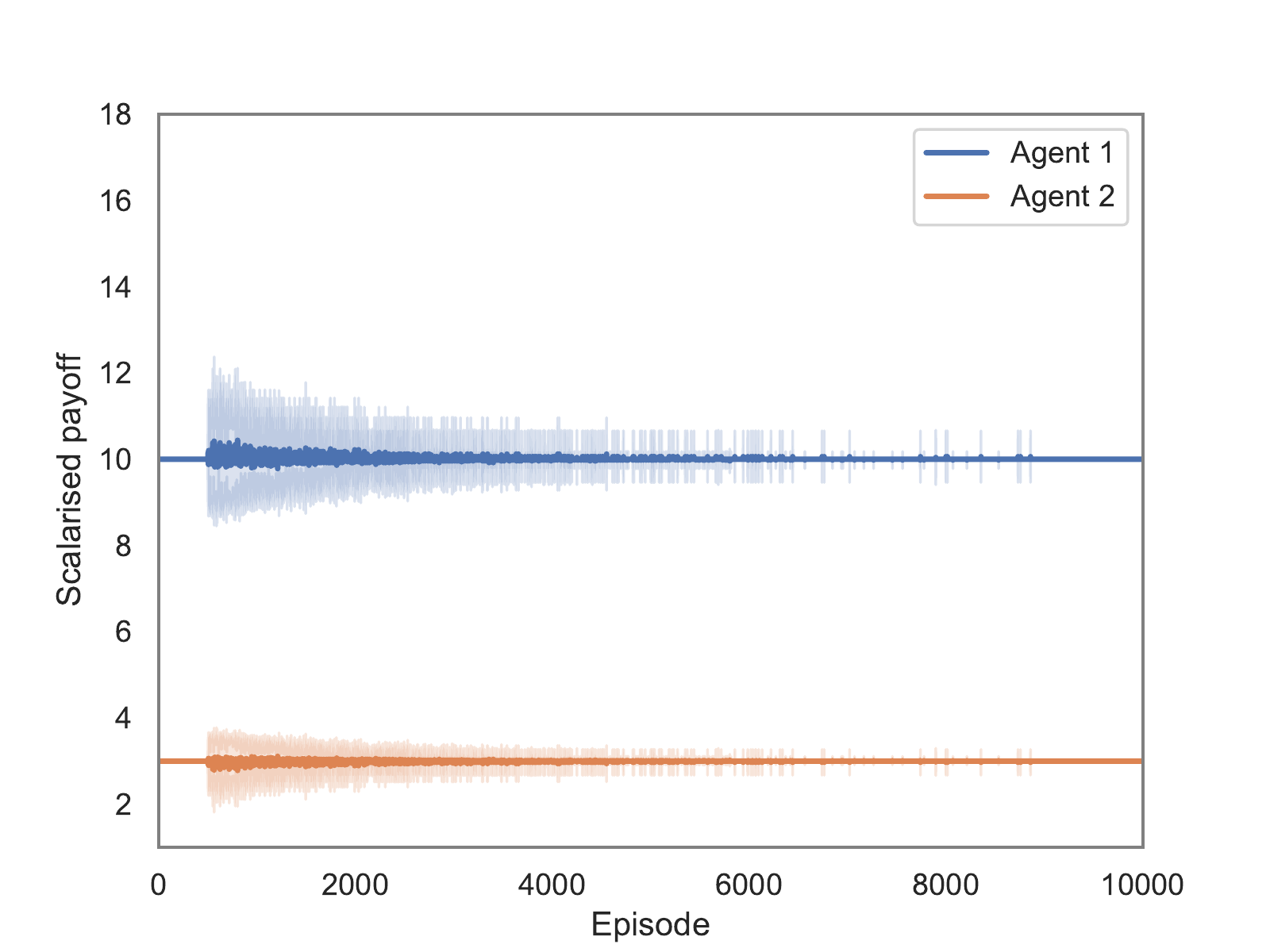}
    \caption{Scalarised payoffs obtained by each agent.}
    \label{fig:game1_sCE_SER}
    \end{subfigure}
    \begin{subfigure}[b]{0.32\textwidth}
       \includegraphics[width=\textwidth]{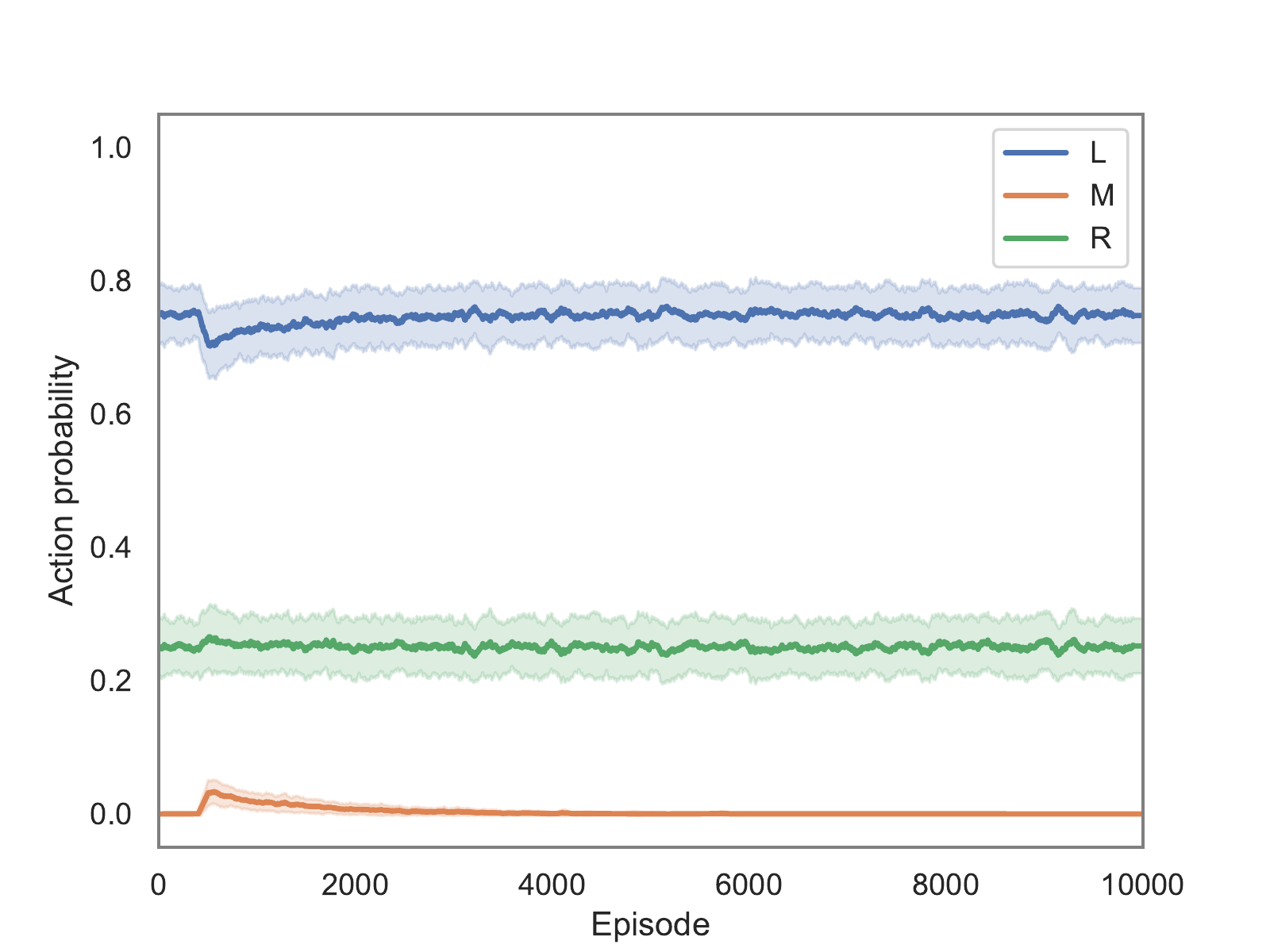}
    \caption{Action selection probabilities of Agent 1.}
    \label{fig:game1_sCE_A1}
    \end{subfigure}
    \begin{subfigure}[b]{0.32\textwidth}
        \includegraphics[width=\textwidth]{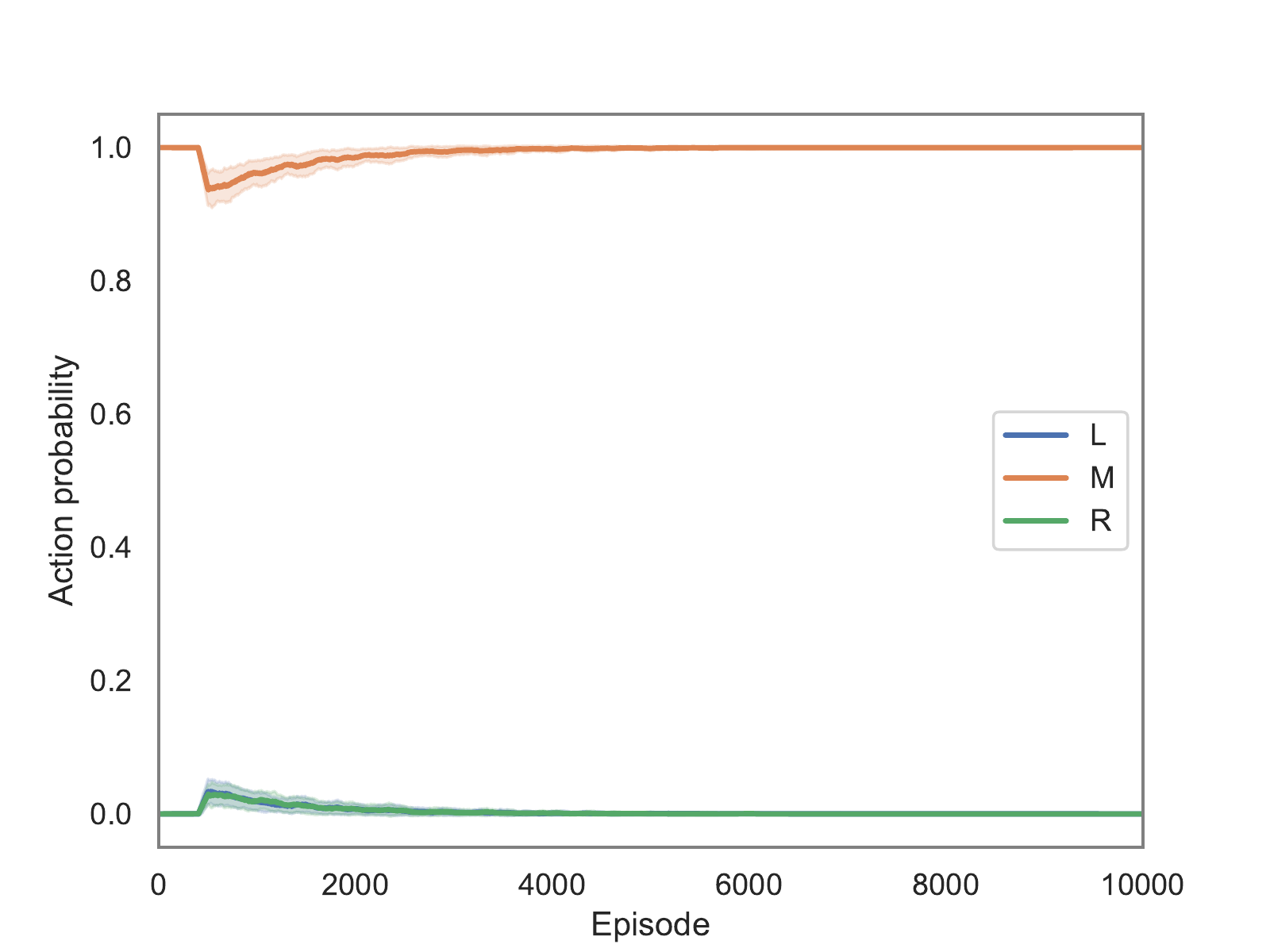}
    \caption{Action selection probabilities of Agent 2.}
    \label{fig:game1_sCE_A2}
    \end{subfigure}
    \caption{Game 1: single-signal CE under SER with action recommendations provided according to Table \ref{table:balance_ser_ce1}.}
    \label{fig:game1_sCE}
\end{figure}
\begin{figure}[ht!]
    \centering
    \begin{subfigure}[b]{0.32\textwidth}
        \includegraphics[width=\textwidth]{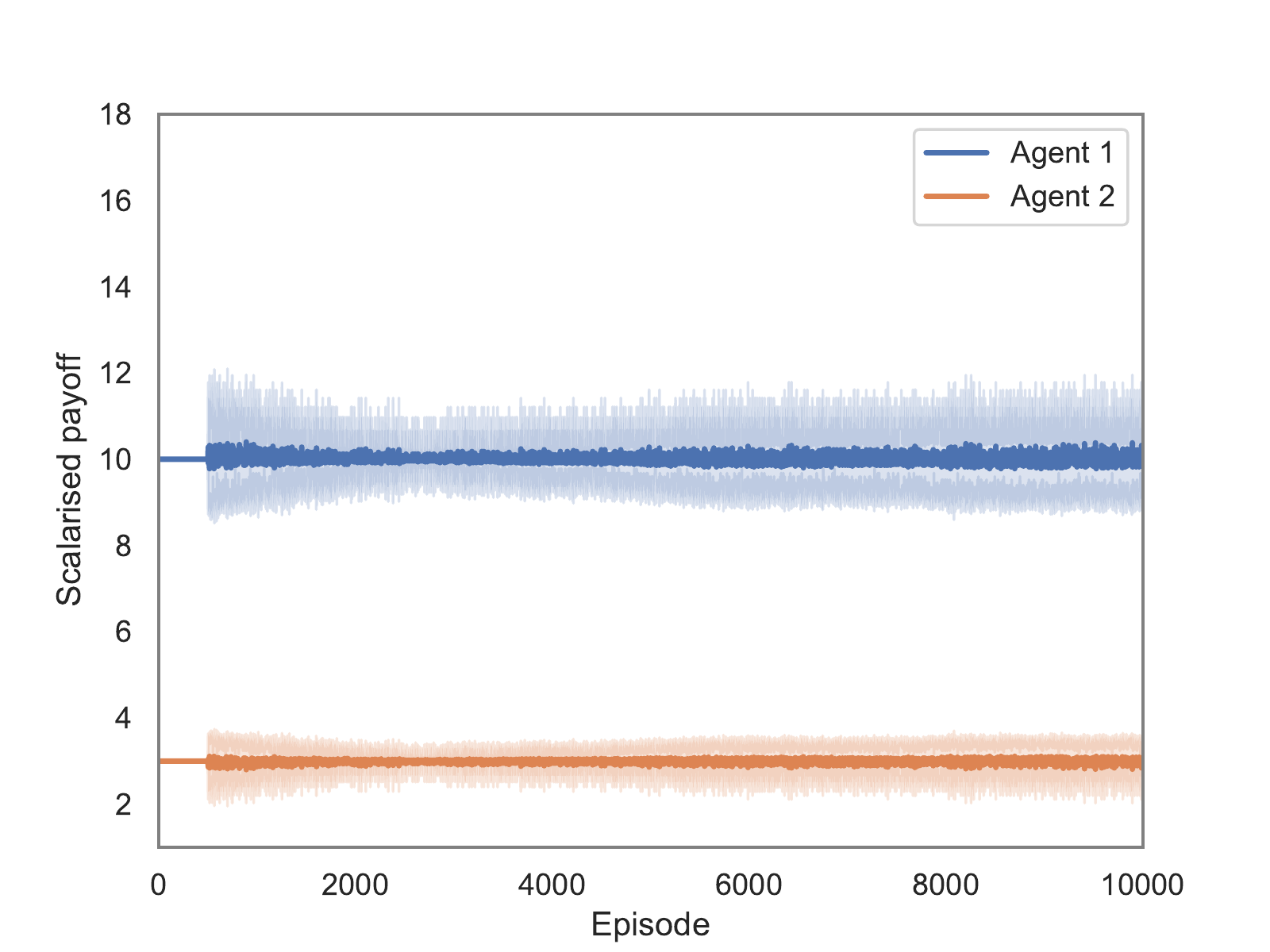}
    \caption{Scalarised payoffs obtained by each agent.}
    \label{fig:game1_mCE_SER}
    \end{subfigure}
    \begin{subfigure}[b]{0.32\textwidth}
       \includegraphics[width=\textwidth]{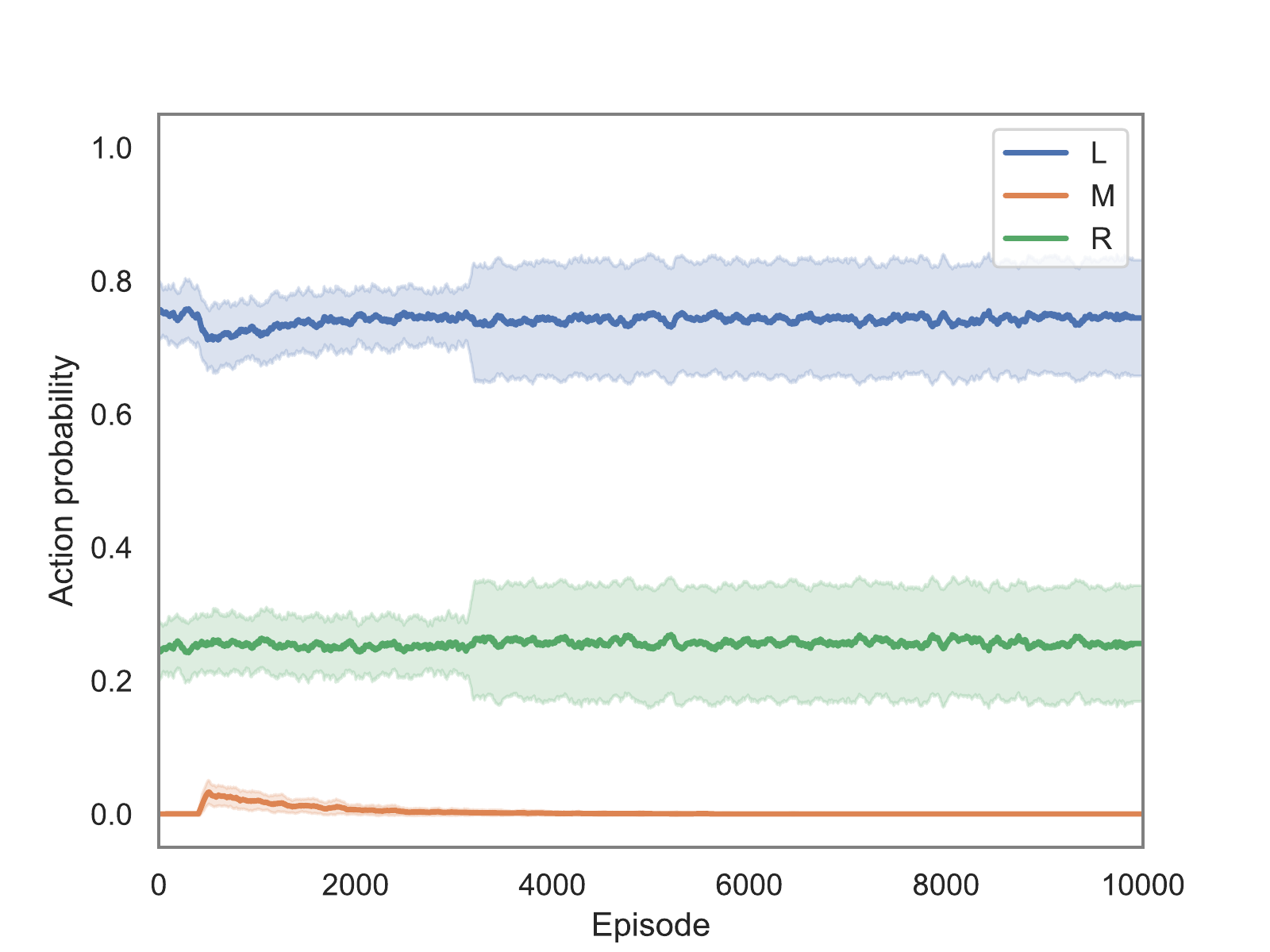}
    \caption{Action selection probabilities of Agent 1.}
    \label{fig:game1_mCE_A1}
    \end{subfigure}
    \begin{subfigure}[b]{0.32\textwidth}
        \includegraphics[width=\textwidth]{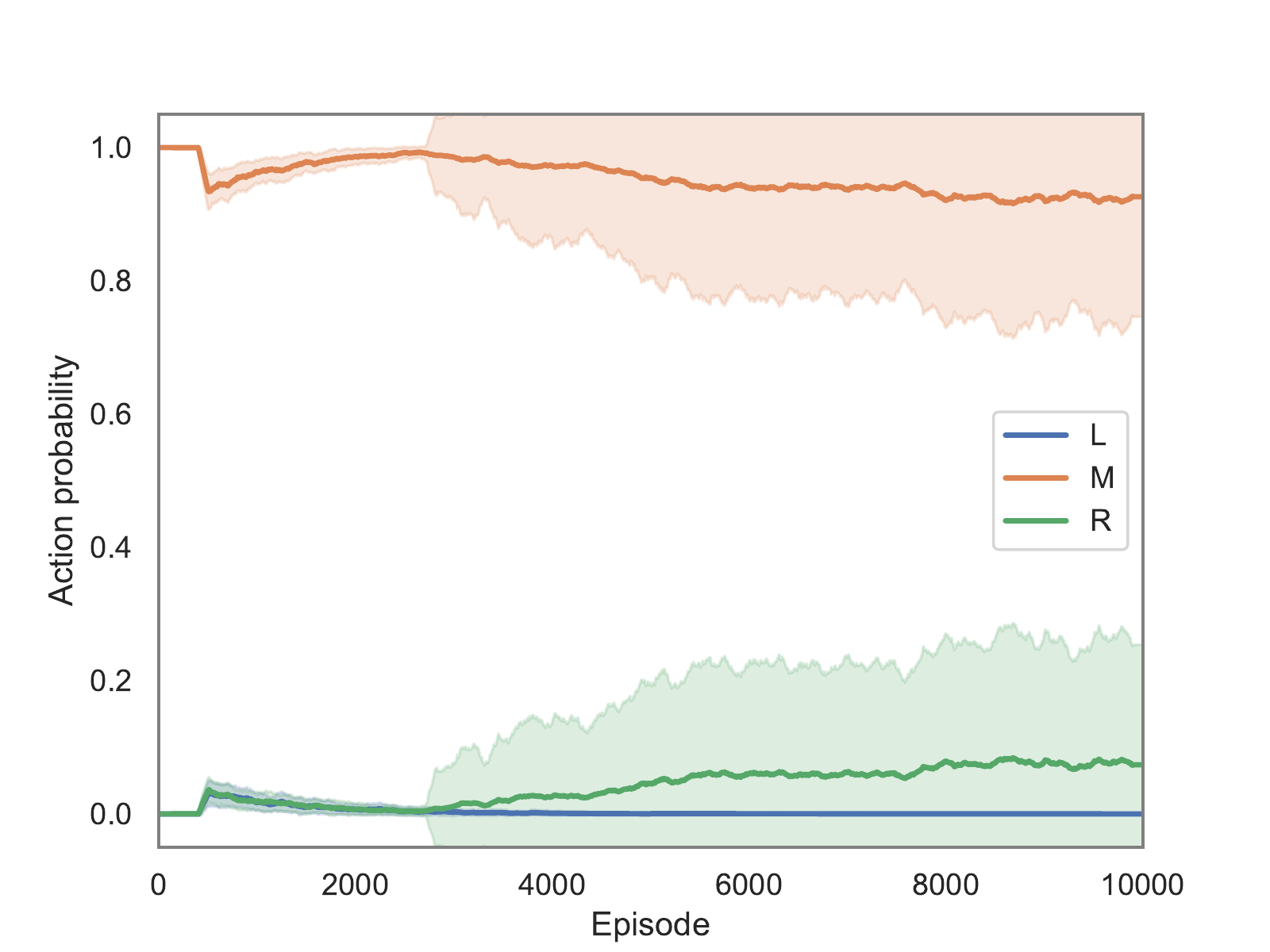}
    \caption{Action selection probabilities of Agent 2.}
    \label{fig:game1_mCE_A2}
    \end{subfigure}
    \caption{Game 1: multi-signal CE under SER with action recommendations provided according to Table \ref{table:balance_ser_ce1}.}
    \label{fig:game1_mCE}
\end{figure}
\begin{figure}[ht!]
    \centering
    \begin{subfigure}[b]{0.32\textwidth}
        \includegraphics[width=\textwidth]{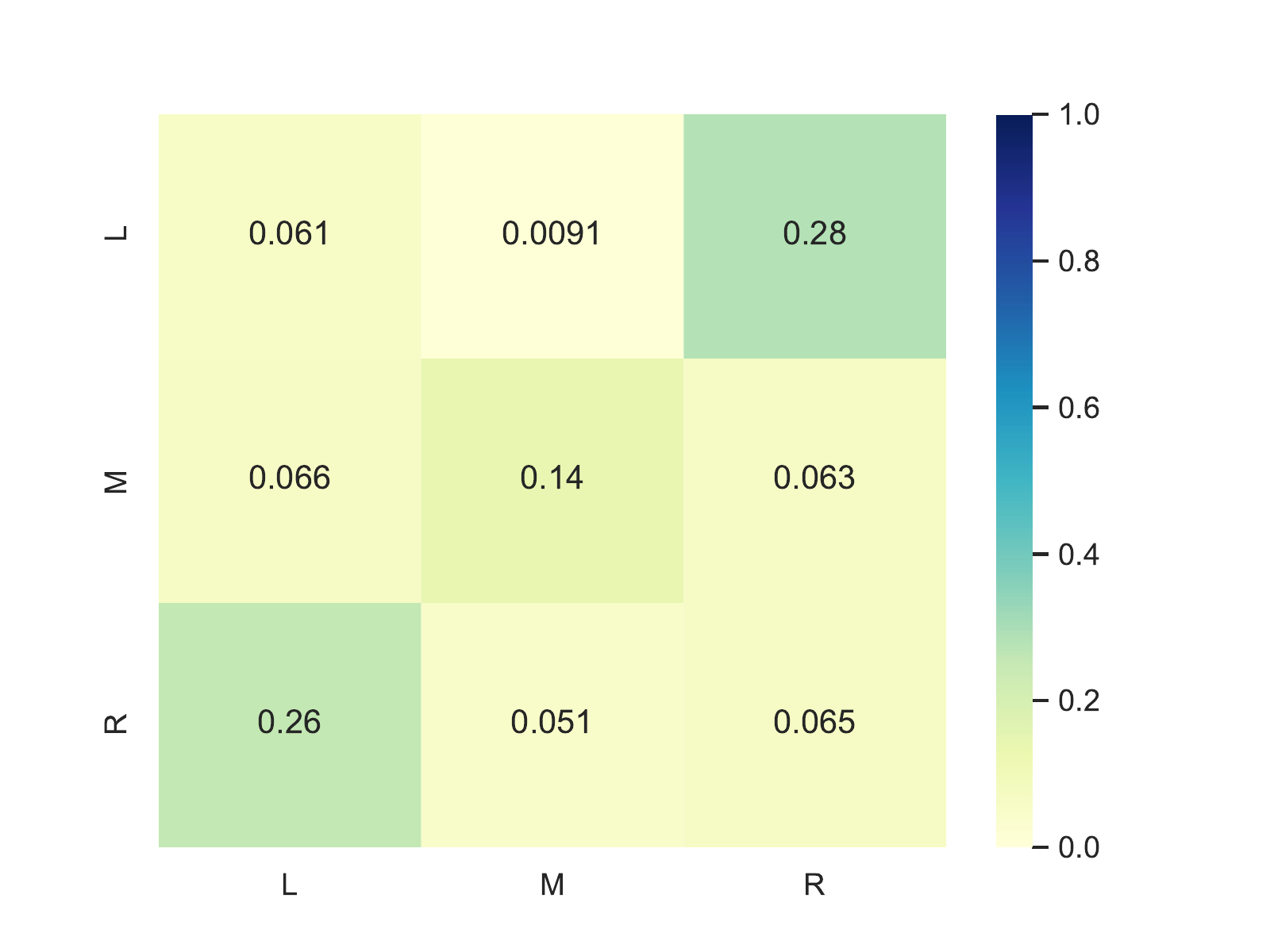}
    \caption{No action recommendations}
    \label{fig:game1_NE_states}
    \end{subfigure}
    \begin{subfigure}[b]{0.32\textwidth}
       \includegraphics[width=\textwidth]{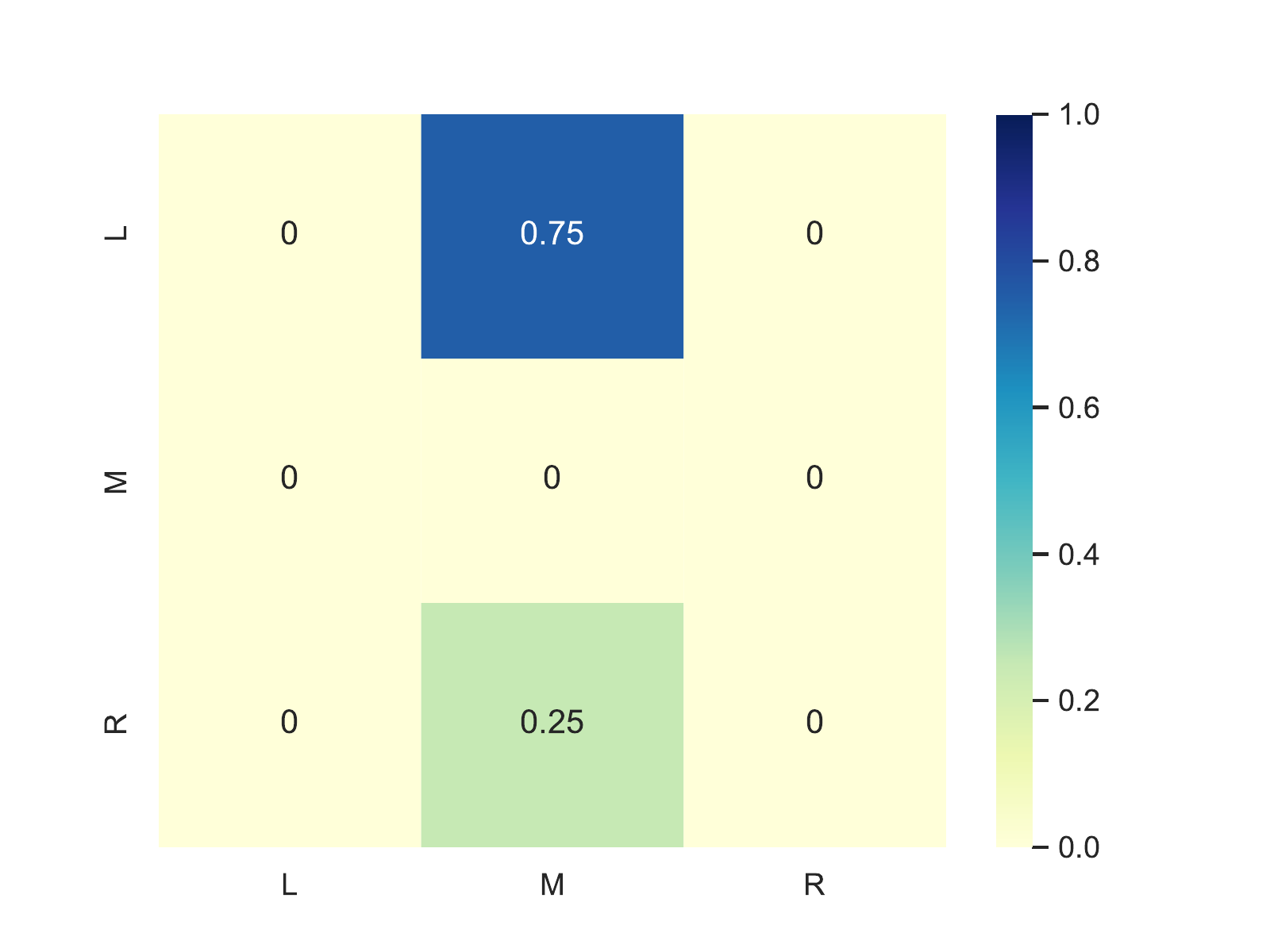}
    \caption{Single-signal CE}
    \label{fig:game1_sCE_states}
    \end{subfigure}
    \begin{subfigure}[b]{0.32\textwidth}
        \includegraphics[width=\textwidth]{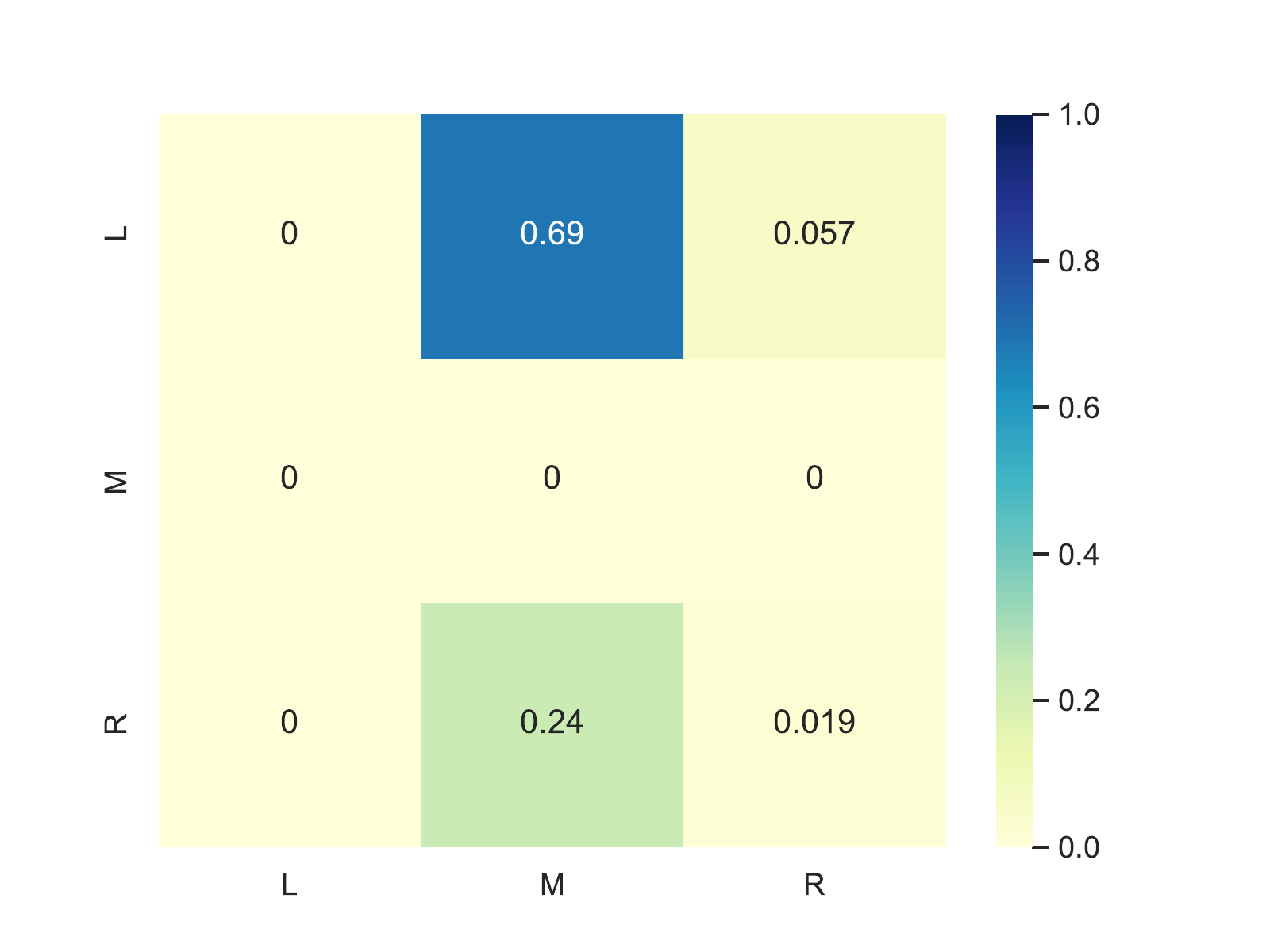}
    \caption{Multi-signal CE}
    \label{fig:game1_mCE_states}
    \end{subfigure}
    \caption{Game 1: Joint action probabilities over the last 1000 episodes under SER.}
    \label{fig:game1_states}
\end{figure}

\subsection{Game 1 - The (Im)balancing Act Game}
For Game 1, the correlated signal was given in accordance to Table \ref{table:balance_ser_ce1}, i.e., in a given episode, (L,M) was recommended with probability 0.75, or else (R,M) was recommended with probability 0.25.

The experimental results in terms of scalarised payoff are shown in Figs.~\ref{fig:game1_NE}, \ref{fig:game1_sCE} and \ref{fig:game1_mCE}. All figures show the scalarised payoffs received by the agents in each episode, averaged over 100 trials. Figure~\ref{fig:game1_states} presents the distribution of outcomes over the joint-action space for the last 1000 interactions, averaged again over 100 trials.
For each experiment we also plot the action selection probabilities for each of the two players (Figs.~\ref{fig:game1_NE_A1}, \ref{fig:game1_NE_A2}, \ref{fig:game1_sCE_A1}, \ref{fig:game1_sCE_A2}, \ref{fig:game1_mCE_A1} and \ref{fig:game1_mCE_A2}). The probabilities are computed using a sliding window of size 100 over the history of taken actions and are also averaged over 100 trials. 
The shaded region around each plot shows one standard deviation from the mean. No smoothing was applied to any of the plots.

It is clear to see from the high standard deviations in Fig. \ref{fig:game1_NE_SER} that agents do not reliably converge on any one joint strategy when no correlated action recommendations are provided. This conclusion is further strengthened when observing the action selection probabilities of player~1 (Fig.~\ref{fig:game1_NE_A1}) and player 2 (Fig.~\ref{fig:game1_NE_A2}). Given our analysis in Theorem \ref{th:ser_nash}, this is to be expected, as agents will always have some incentive to deviate from a potential Nash equilibrium point in this game. As $\epsilon$ is decayed, the agents' behaviour does not converge to any stable point, and the joint strategies learned in each run seem to always cycle among a few possibilities (e.g., predominant joint-actions are (R, L), (L, R) and (M, M) as it can be seen from  Figure~\ref{fig:game1_states}).

In Figure \ref{fig:game1_sCE}, the effect of the single-signal correlated equilibrium may clearly be seen. As we would expect, for the first 500 episodes a consistent scalarised payoff is received by both agents while they learn the correlated equilibrium. From episode $501$ both agents are free to select actions autonomously and to explore and learn the effects of deviating from the action suggestions. As $\epsilon$ is gradually decayed towards zero, the agents consistently converge back to the correlated equilibrium, evidenced by the low standard deviations around the means of the scalarised payoffs near episode 10,000. Furthermore, Fig.~\ref{fig:game1_sCE_A1}, \ref{fig:game1_sCE_A2} and \ref{fig:game1_sCE_states} show that the action selection probabilities for each player nicely converge to the probabilities of the correlated equilibrium in Table~\ref{table:balance_ser_ce1} (i.e., agent 1 will select L with 25\% probability and R with 75\% probability, while agent 2 ends up selecting M 100\% of the time). This provides empirical support for our claim in Theorem \ref{th:ser_ce_single} that single-signal correlated equilibria can exist in MONFGs under SER, demonstrating that neither agent has an incentive to deviate unilaterally given an action recommendation, when learning in this MONFG under SER.

For the case of multi-signal correlated equilibrium, Figure~\ref{fig:game1_mCE} clearly indicates how, after the initial 500 episodes, the agents slowly diverge from the given recommendations. From Fig. \ref{fig:game1_mCE_A2} we can notice how agent 2, decays the use of the recommended action M, replacing it consistently with R, as it is trying to push the outcome towards the more imbalanced payoff outcome (L, R), given that his opponent is initially still taking the recommended actions L with 75\% probability. We can then notice from Fig.~\ref{fig:game1_mCE_A1} an attempt from agent 1 to coordinate their actions to obtain (R, R), but with less success according to the join-action distribution outcome presented in Fig.~\ref{fig:game1_mCE_states}. This provides empirical support for our claim in Theorem \ref{th:ser_ce} that multi-signal correlated equilibria need not exist in MONFGs under SER, demonstrating that the agents have incentives to deviate from the given action recommendations, when learning in this MONFG under SER.

\begin{figure}[ht!]
    \centering
    \begin{subfigure}[b]{0.32\textwidth}
        \includegraphics[width=\textwidth]{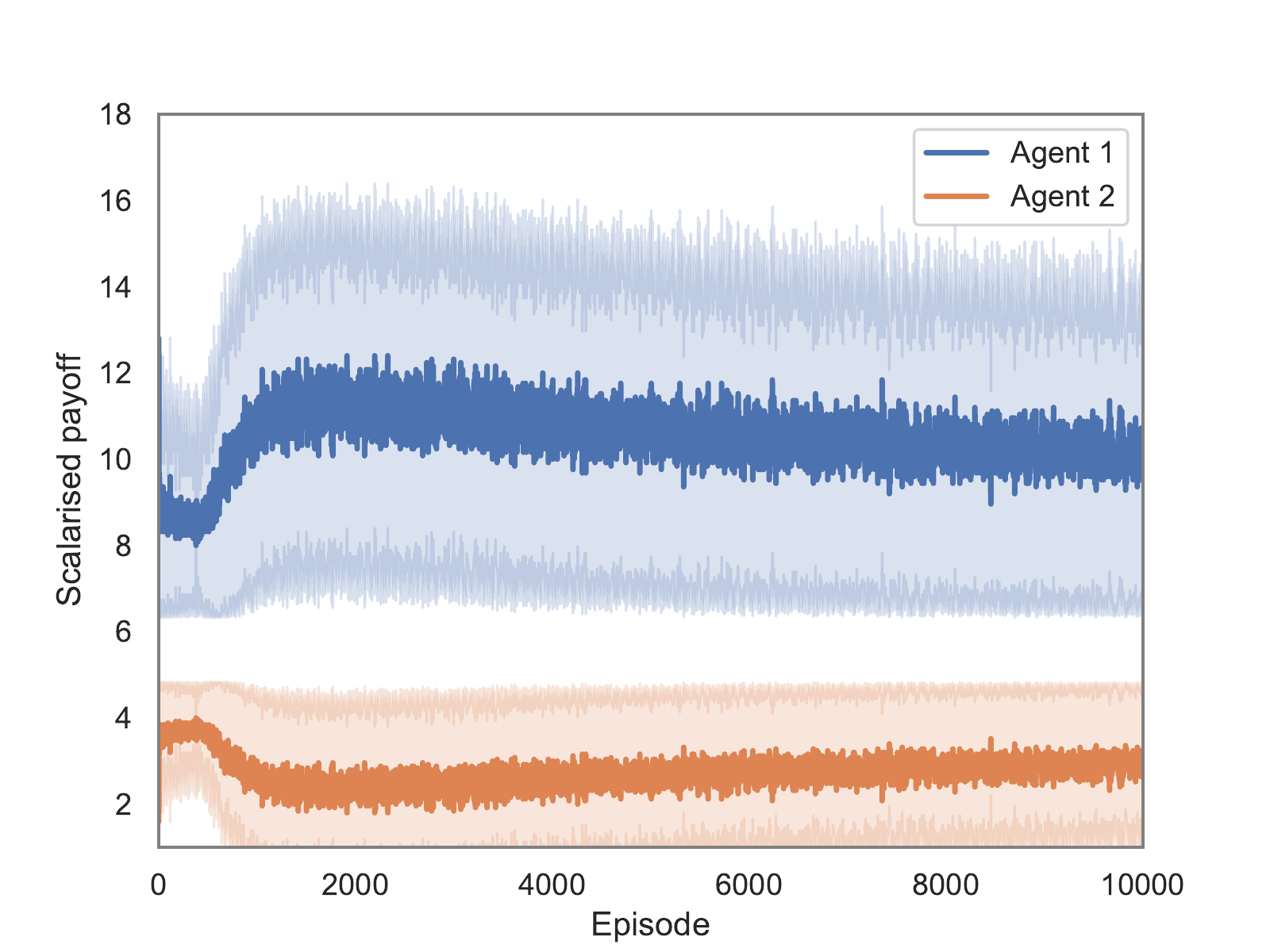}
    \caption{Scalarised payoffs obtained by each agent.}
    \label{fig:game2noM_NE_SER}
    \end{subfigure}
    \vspace{\baselineskip}
    \begin{subfigure}[b]{0.32\textwidth}
       \includegraphics[width=\textwidth]{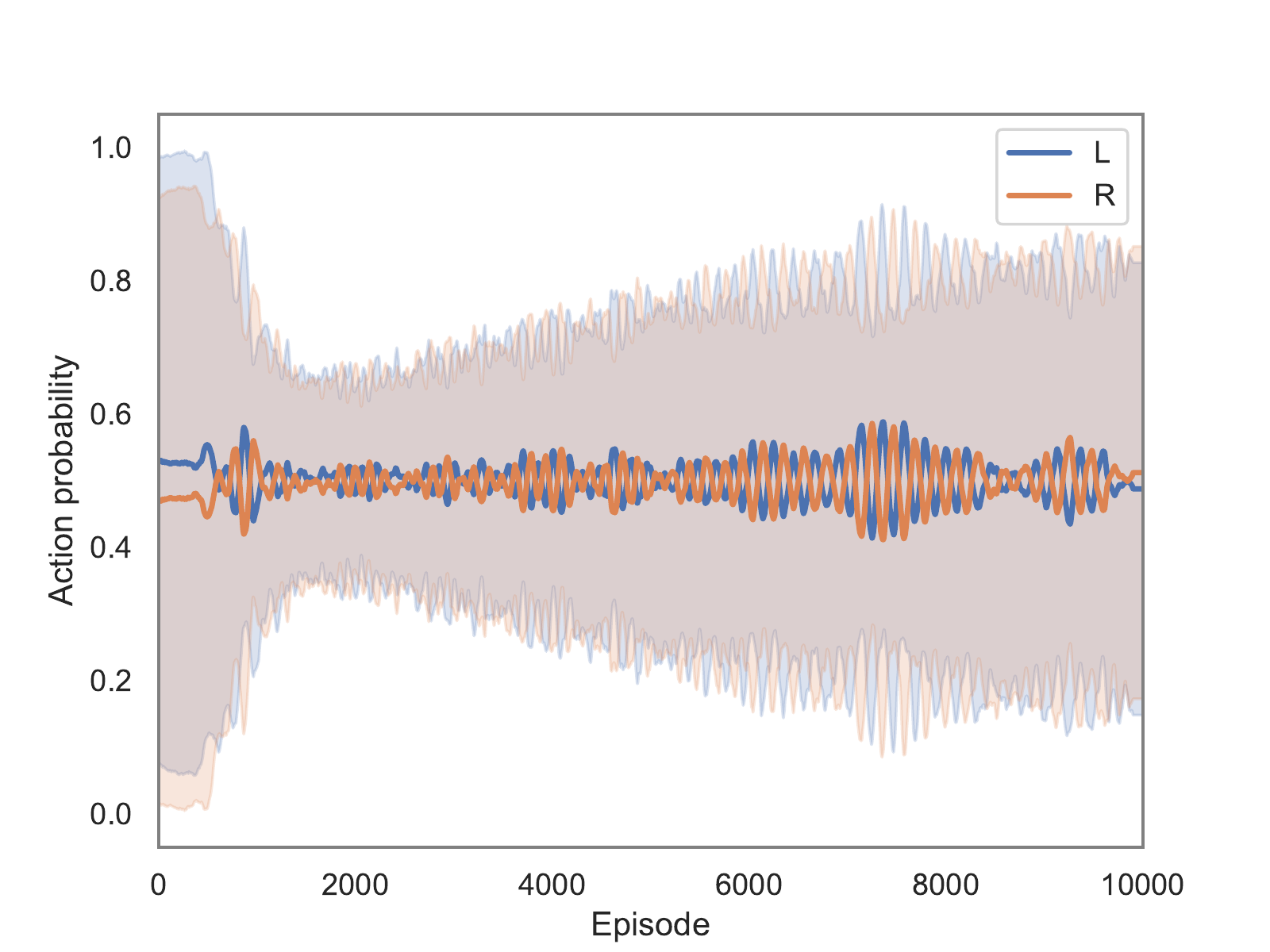}
    \caption[width=\textwidth]{Action selection probabilities of Agent 1.}
    \label{fig:game2noM_NE_A1}
    \end{subfigure}
    \begin{subfigure}[b]{0.32\textwidth}
        \includegraphics[width=\textwidth]{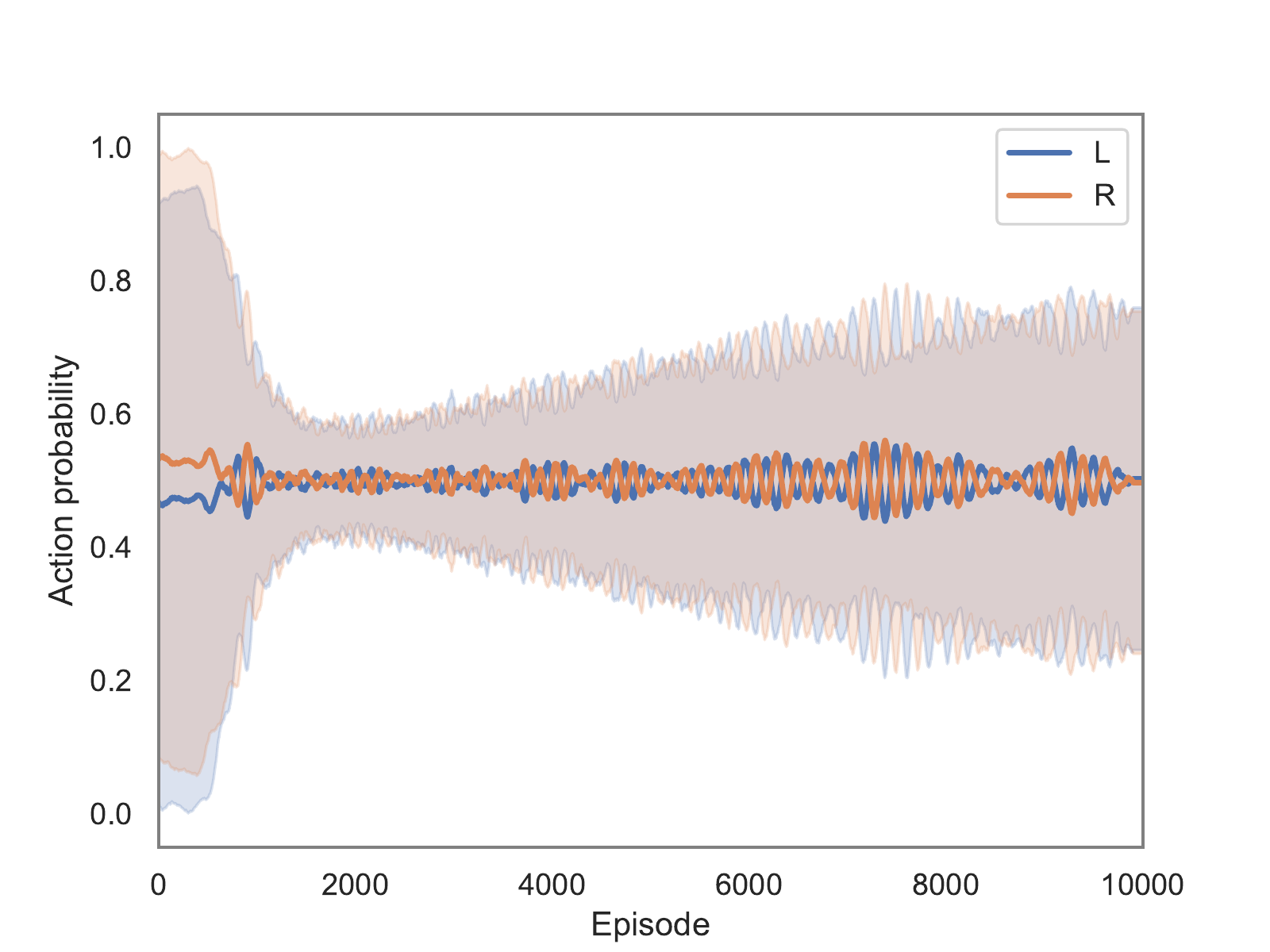}
    \caption{Action selection probabilities of Agent 2.}
    \label{fig:game2noM_NE_A2}
    \end{subfigure}
    \caption{Game 2 under SER with no action recommendations.}
    \label{fig:game2noM_NE}
\end{figure}
\begin{figure}[ht!]
    \centering
    \begin{subfigure}[b]{0.32\textwidth}
        \includegraphics[width=\textwidth]{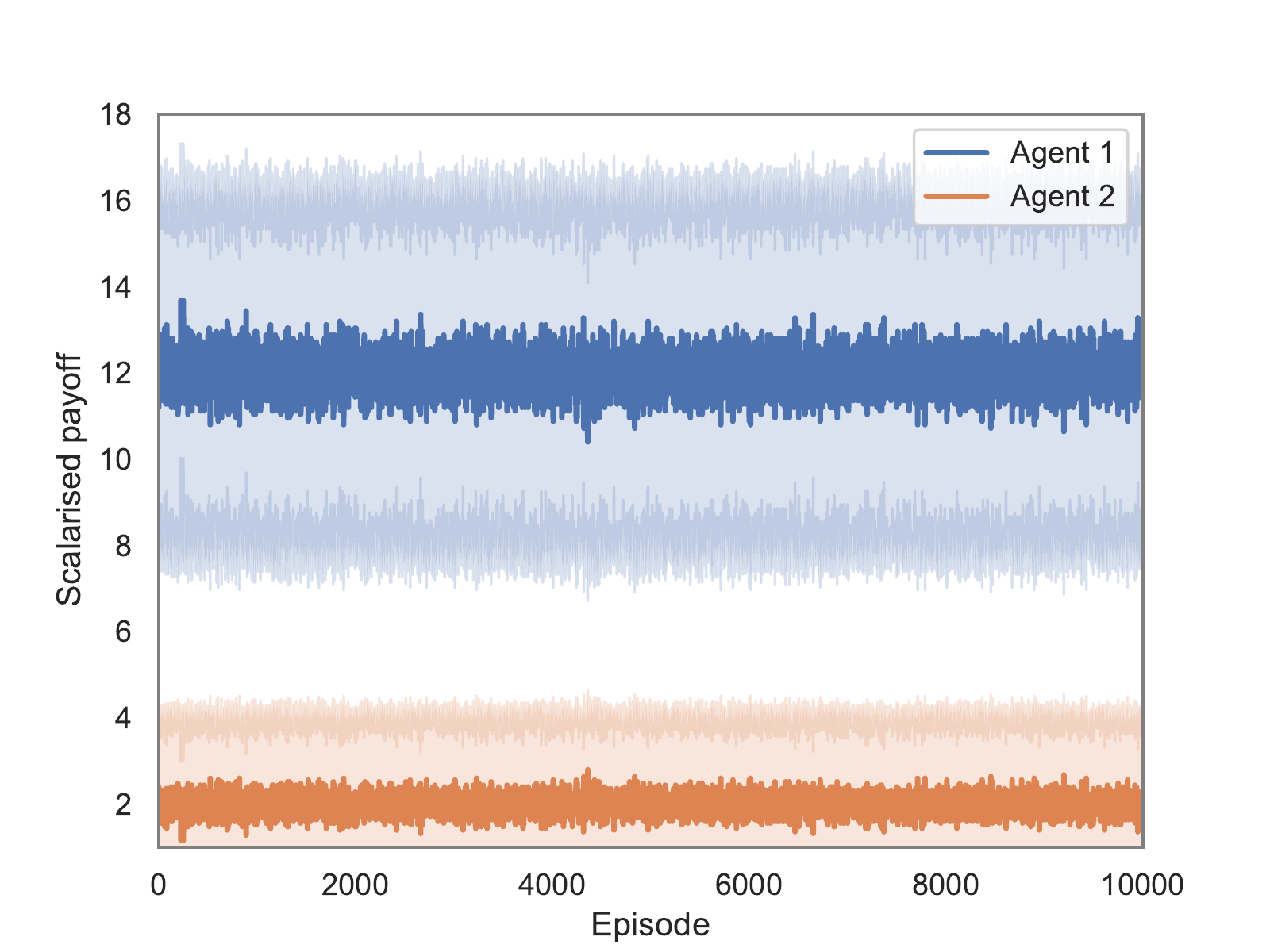}
    \caption{Scalarised payoffs obtained by each agent.}
    \label{fig:game2noM_sCE_SER}
    \end{subfigure}
    \begin{subfigure}[b]{0.32\textwidth}
       \includegraphics[width=\textwidth]{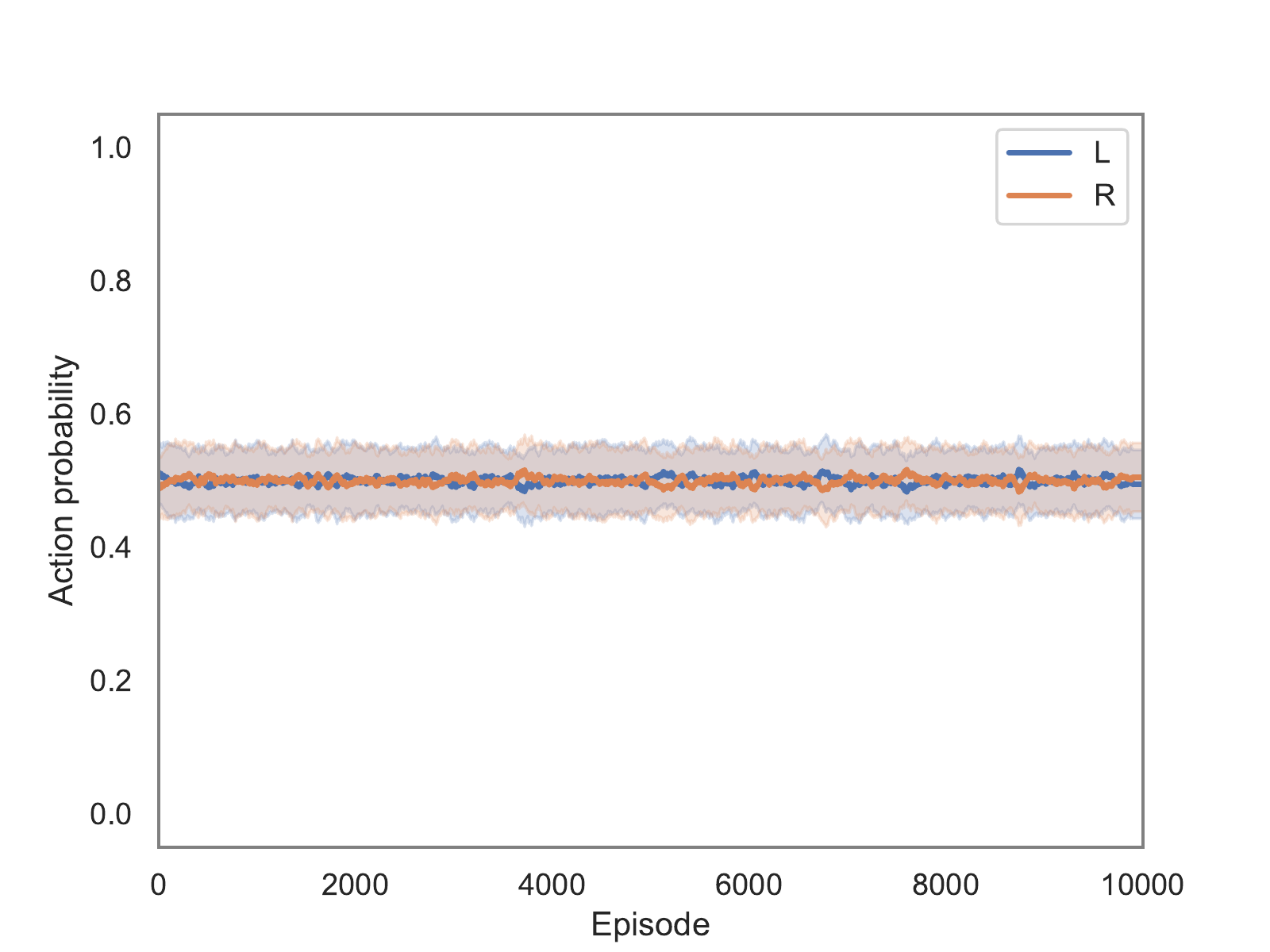}
    \caption{Action selection probabilities of Agent 1.}
    \label{fig:game2noM_sCE_A1}
    \end{subfigure}
    \begin{subfigure}[b]{0.32\textwidth}
        \includegraphics[width=\textwidth]{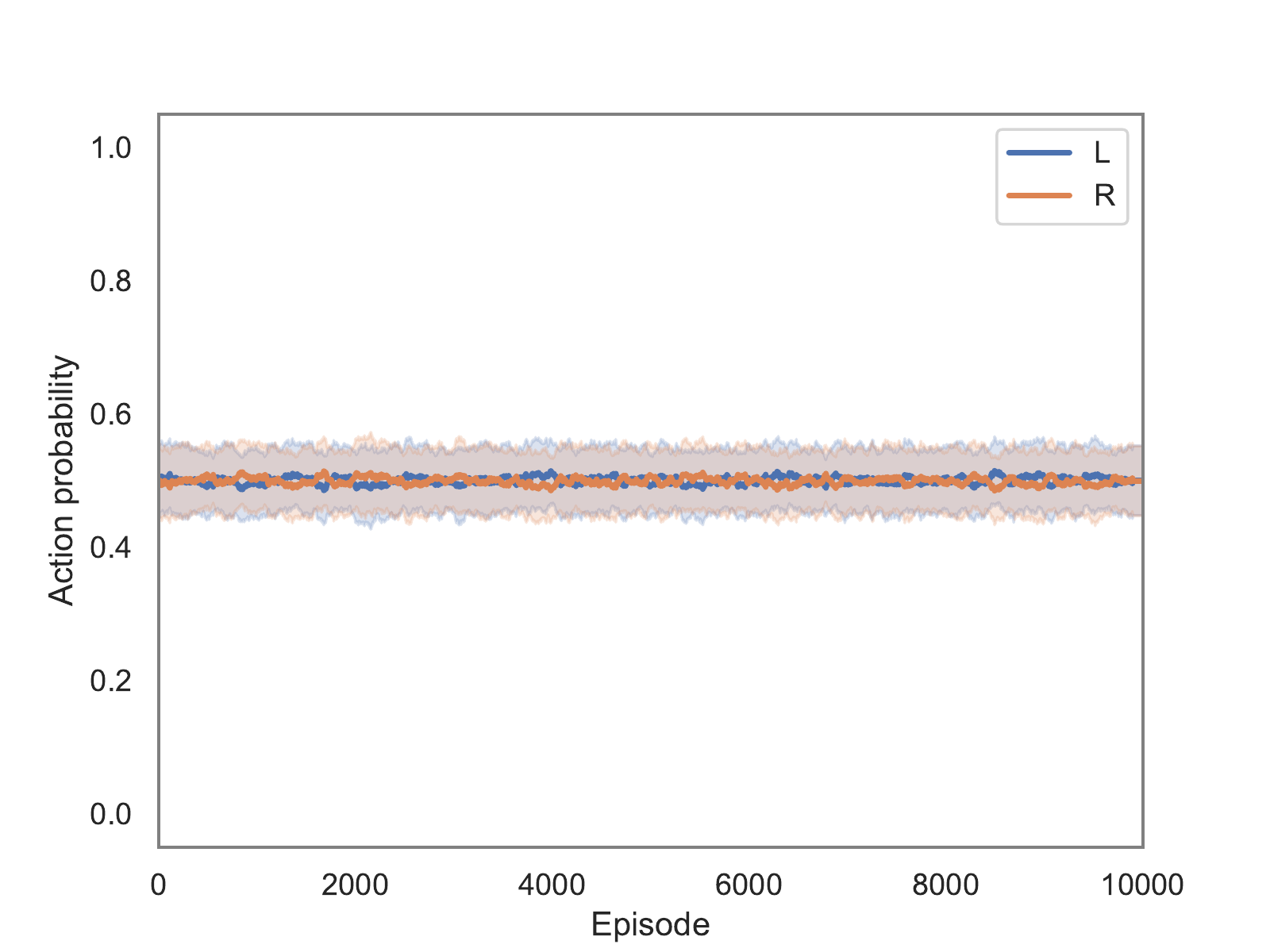}
    \caption{Action selection probabilities of Agent 2.}
    \label{fig:game2noM_sCE_A2}
    \end{subfigure}
    \caption{Game 2: single-signal CE under SER with action recommendations provided according to Table \ref{table:balance_ser_ce1}.}
    \label{fig:game2noM_sCE}
\end{figure}
\begin{figure}[ht!]
    \centering
    \begin{subfigure}[b]{0.32\textwidth}
        \includegraphics[width=\textwidth]{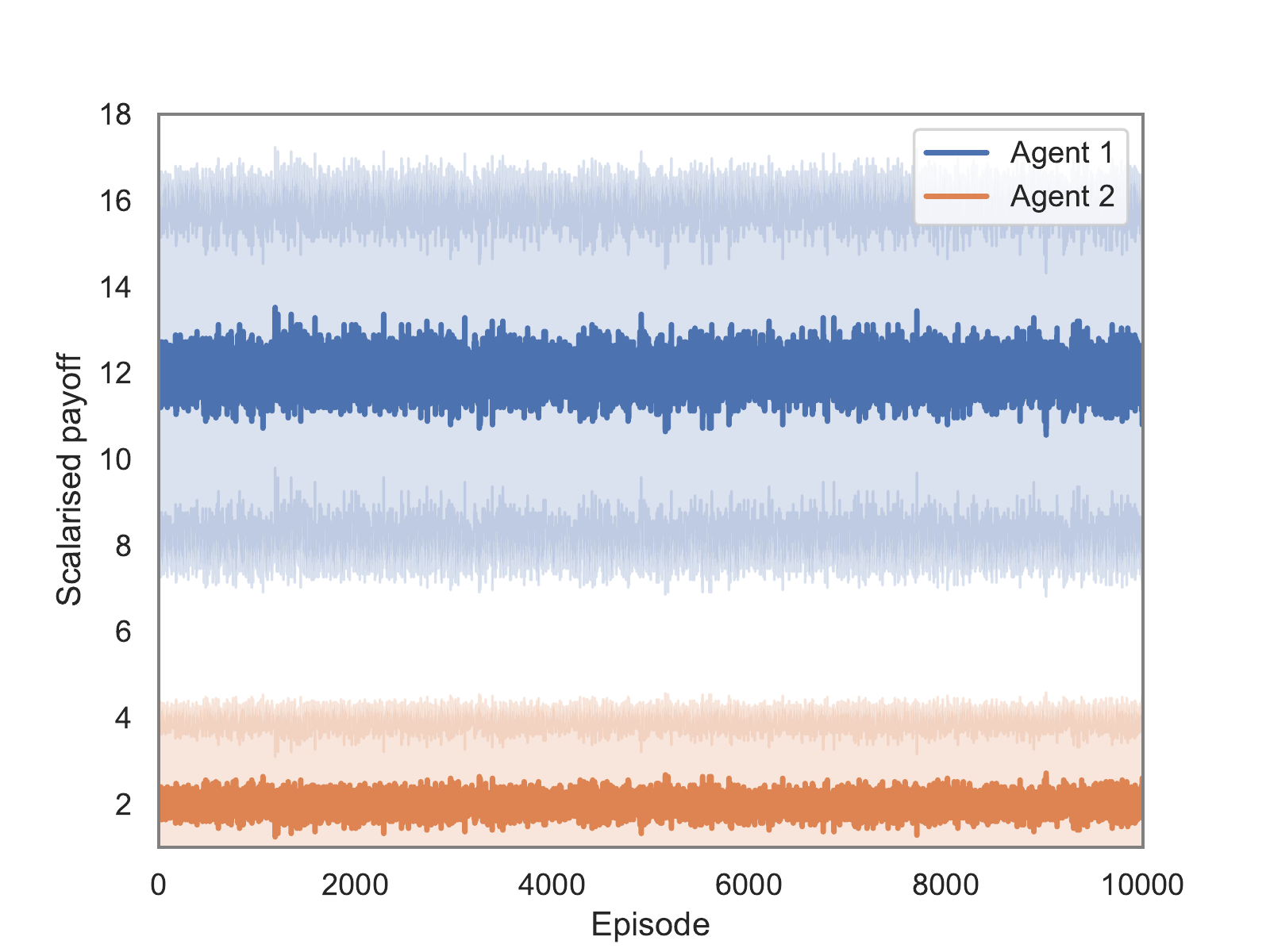}
    \caption{Scalarised payoffs obtained by each agent.}
    \label{fig:game2noM_mCE_SER}
    \end{subfigure}
    \begin{subfigure}[b]{0.32\textwidth}
       \includegraphics[width=\textwidth]{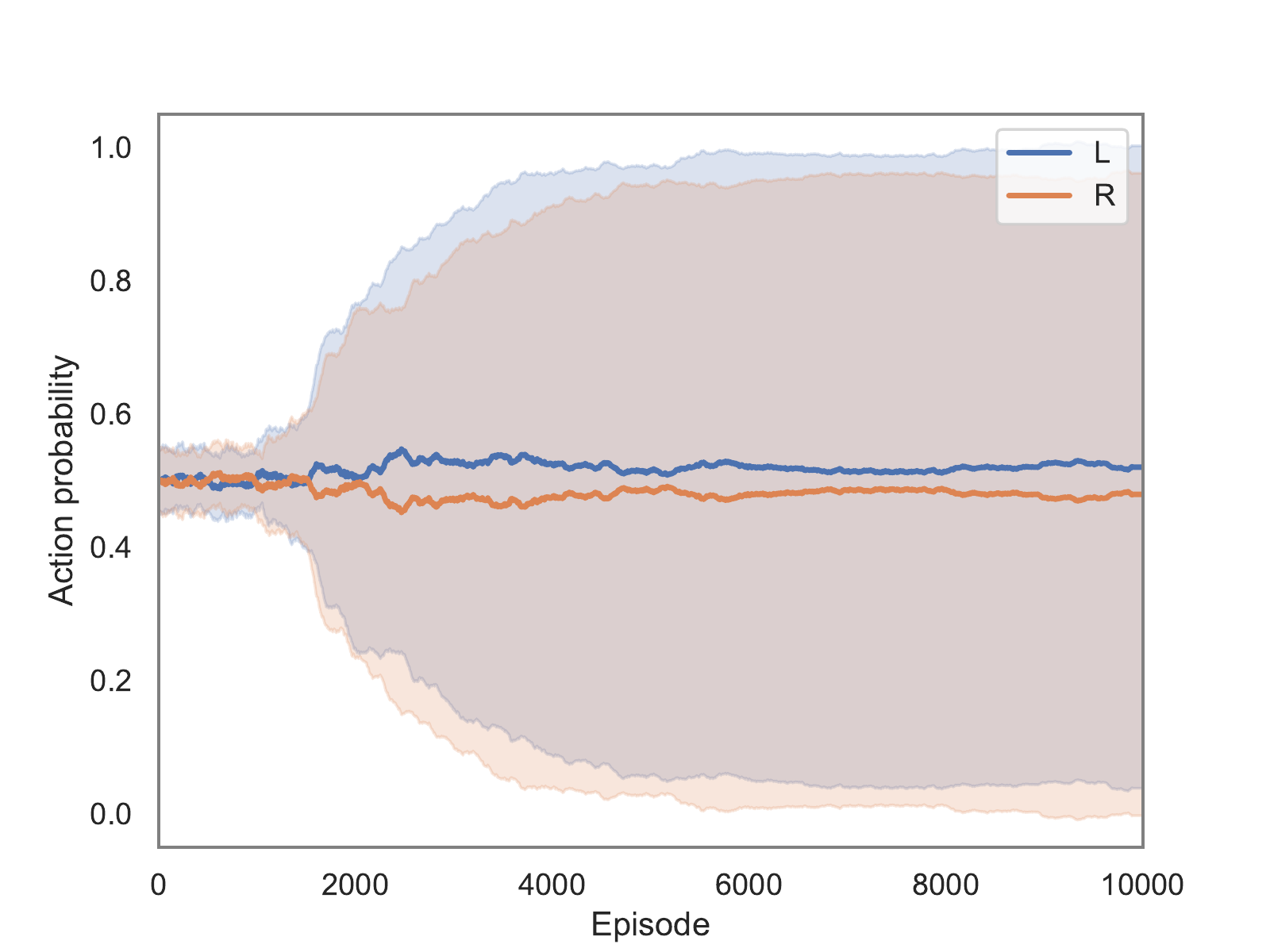}
    \caption{Action selection probabilities of Agent 1.}
    \label{fig:game2noM_mCE_A1}
    \end{subfigure}
    \begin{subfigure}[b]{0.32\textwidth}
        \includegraphics[width=\textwidth]{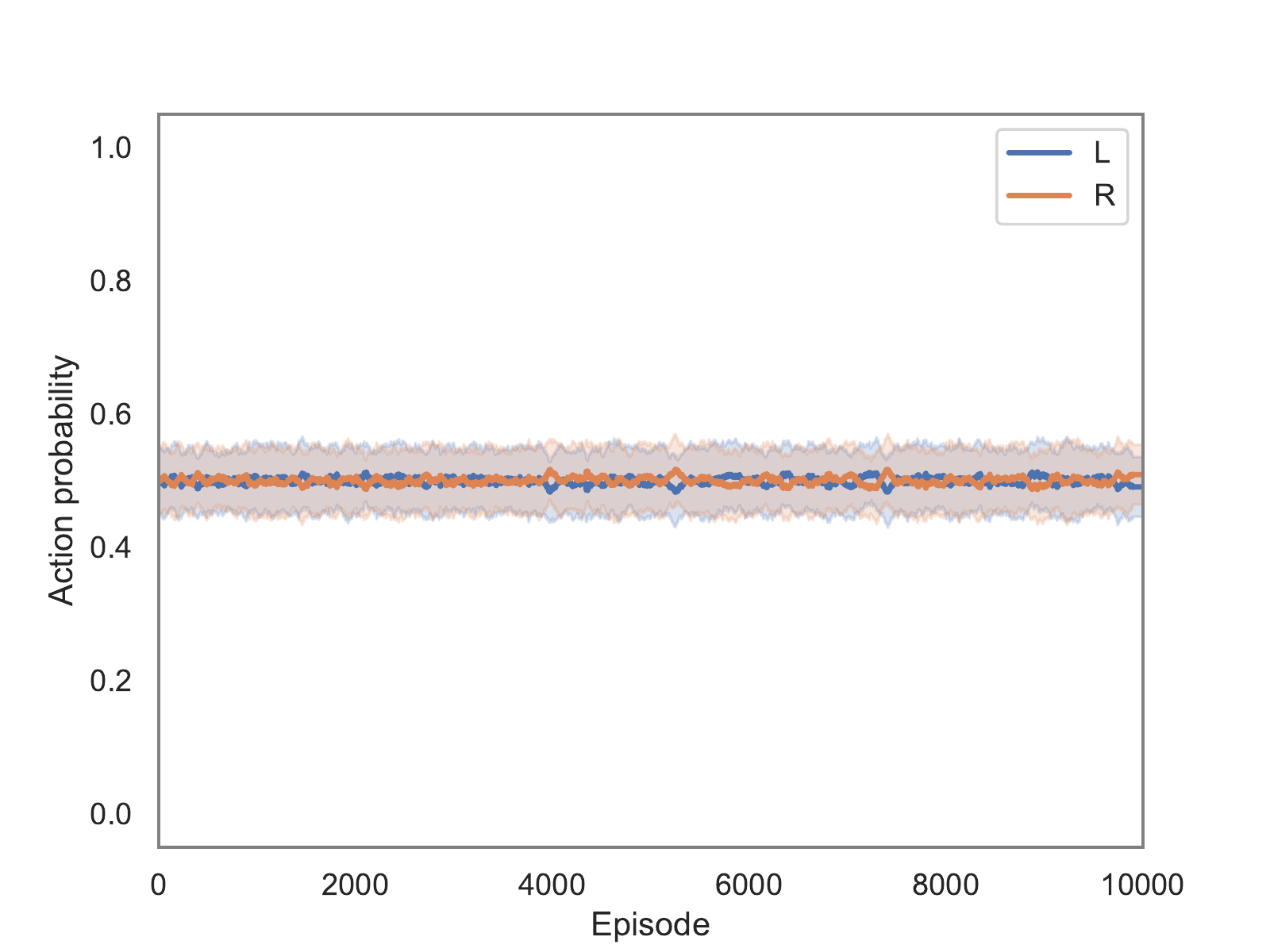}
    \caption{Action selection probabilities of Agent 2.}
    \label{fig:game2noM_mCE_A2}
    \end{subfigure}
    \caption{Game 2: multi-signal CE under SER with action recommendations provided according to Table \ref{table:balance_ser_ce1}.}
    \label{fig:game2noM_mCE}
\end{figure}
\begin{figure}[ht!]
    \centering
    \begin{subfigure}[b]{0.32\textwidth}
        \includegraphics[width=\textwidth]{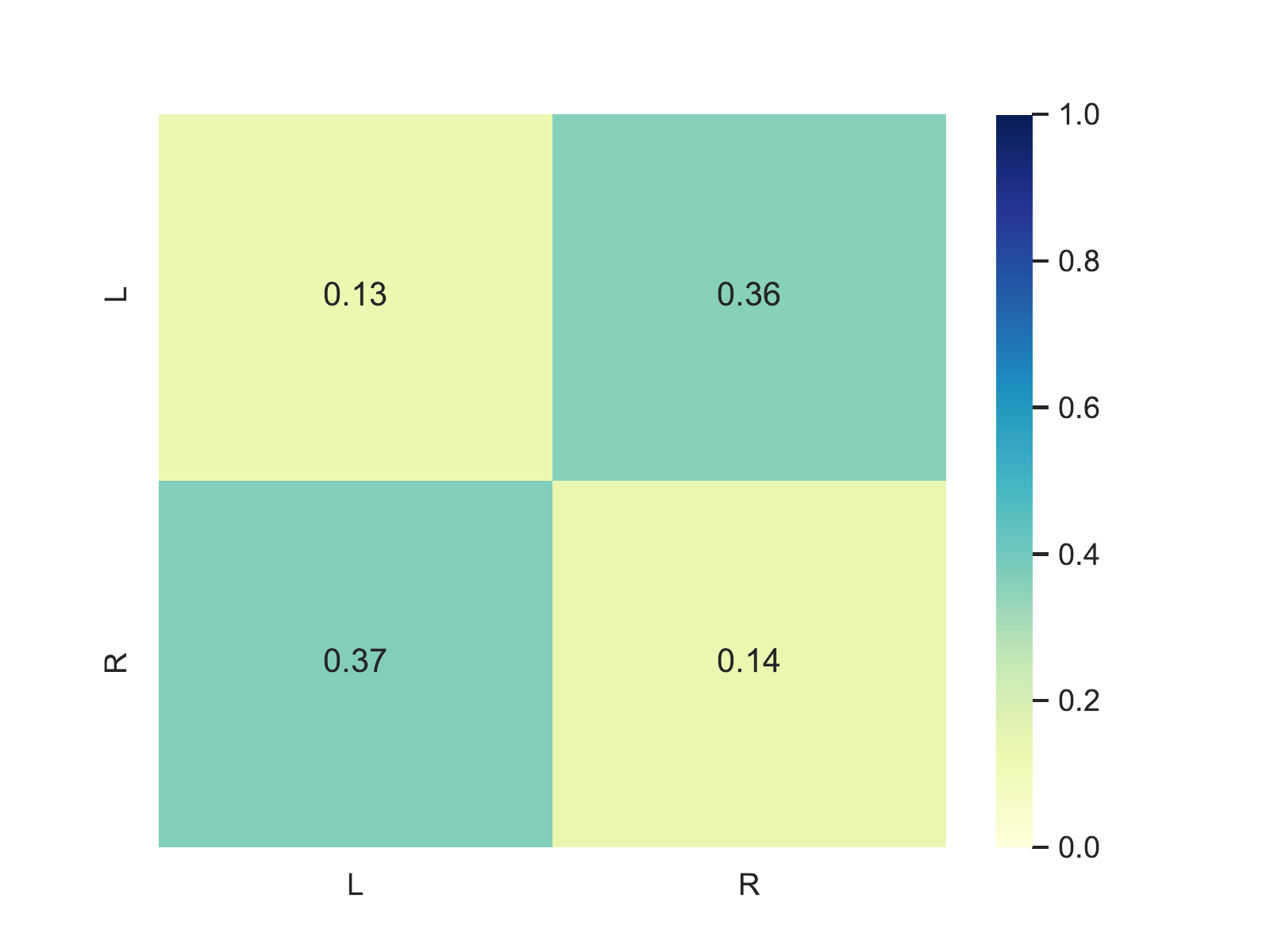}
    \caption{No action recommendations}
    \label{fig:game2noM_NE_states}
    \end{subfigure}
    \begin{subfigure}[b]{0.32\textwidth}
       \includegraphics[width=\textwidth]{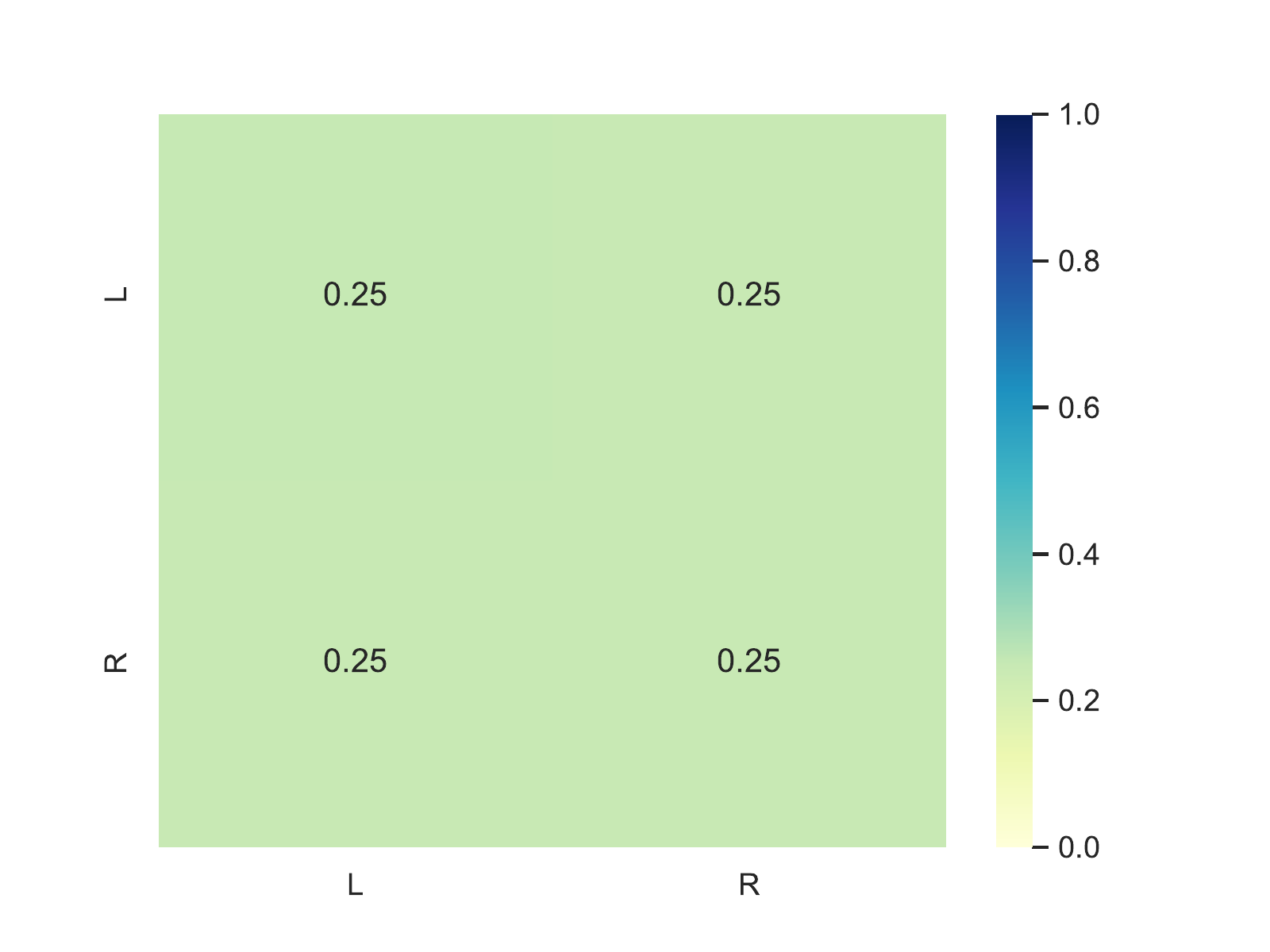}
    \caption{Single-signal CE}
    \label{fig:game2noM_sCE_states}
    \end{subfigure}
    \begin{subfigure}[b]{0.32\textwidth}
        \includegraphics[width=\textwidth]{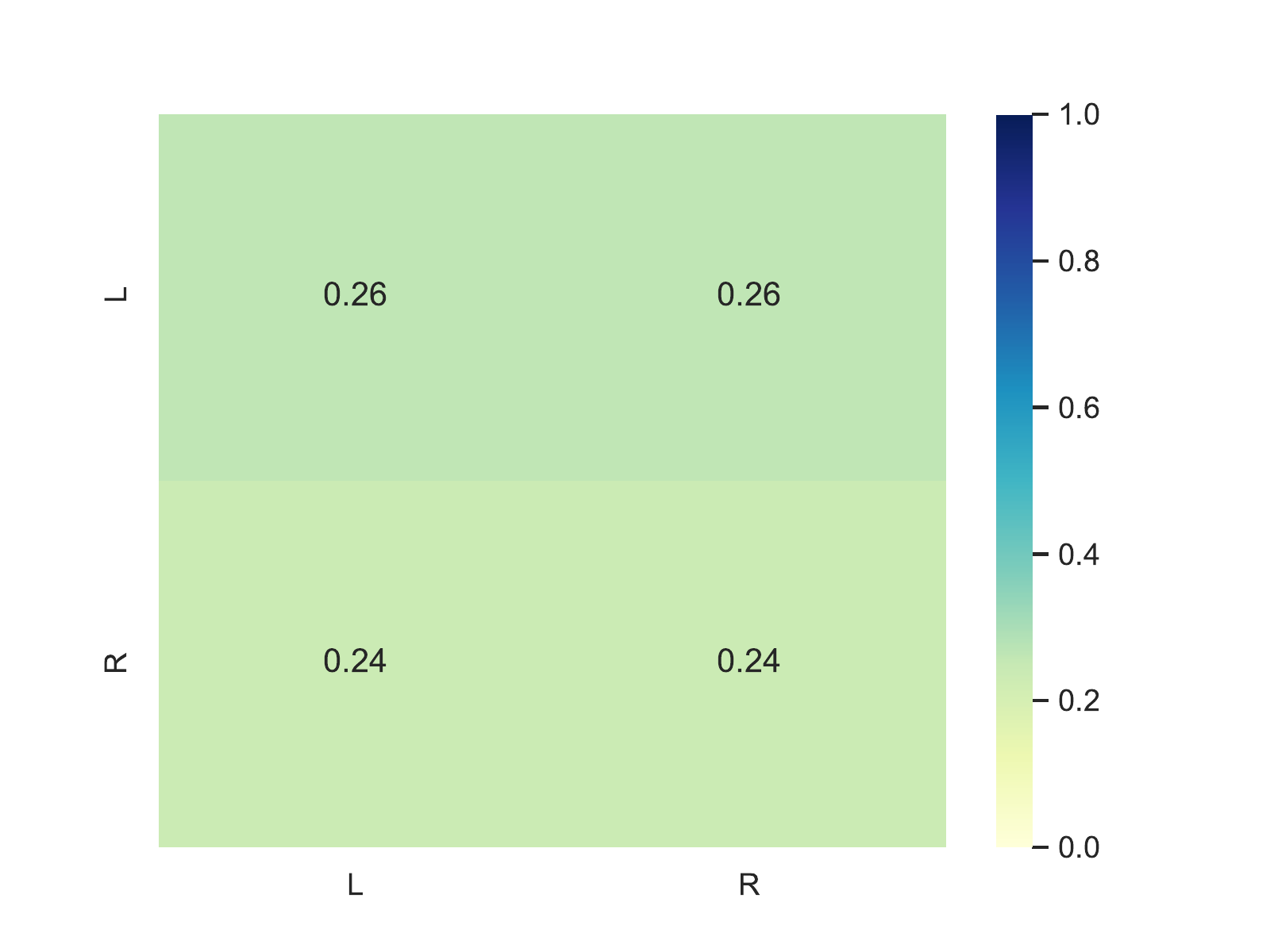}
    \caption{Multi-signal CE}
    \label{fig:game2noM_mCE_states}
    \end{subfigure}
    \caption{Game 2: Joint action probabilities over the last 1000 episodes under SER.}
    \label{fig:game2noM_states}
\end{figure}

\subsection{Game 2 - The (Im)balancing Act Game without action M} 
For Game 2, the correlated signal was given in accordance to Table \ref{table:balance_noM} (right), i.e., in a given episode, each possible joint-action among the four (i.e., (L, L), (L, R), (R, L) or (R, R)) is recommended with equal probability.

The experimental results in terms of scalarised payoff are shown in Figs.~\ref{fig:game2noM_NE}, \ref{fig:game2noM_sCE} and \ref{fig:game2noM_mCE} respectively. Figure~\ref{fig:game2noM_states} presents the distribution over the joint-action space for the last 1000 interactions. Again, all experiments are run for 10,000 interactions, averaged over 100 trials. 

Figs.~\ref{fig:game2noM_NE_A1} and \ref{fig:game2noM_NE_A2}  highlight the dynamics between the agents of shifting between balanced and imbalanced outcomes, without being able to converge to any stable equilibrium strategies when no action recommendations are given, implying that there are no NE present in this game. According to Fig.~\ref{fig:game2noM_NE_states}, player 2 seems to be more successful in obtaining her desired outcomes (L,R) or (R,L). Regarding the multi-signal CE, we see from Fig.~\ref{fig:game2noM_mCE_A1} that player 1 has a stronger incentive to deviate from the recommendations, probably obtained due to the higher loss in utility incurred when switching between the possible outcomes. In any case, similar to the previous experiment, agents are not able to converge to any stable strategy, implying that the set of action recommendations in Table \ref{table:balance_noM} (right) do not constitute a multi-signal CE for this game. Finally, similar to the 3-action (Im)balancing Act game, agents have no incentive to deviate from the action recommendations in Table \ref{table:balance_noM} (right) in the case of a single-signal CE, as can be seen from Figs.~\ref{fig:game2noM_sCE_A1}, \ref{fig:game2noM_sCE_A2}. Additionally, the distribution of outcomes over the joint-action space (Fig.~\ref{fig:game2noM_sCE_states}) also closely aligns with the action recommendations in Table \ref{table:balance_noM} (right), thus allowing the agents to fairly coordinate between ending up half of the time in the imbalanced payoff outcomes, preferred by player 1, and the balanced payoff outcomes, preferred by player 2.

\subsection{Game 3 - A 3-action MONFG with pure NE} 
For Game 3, the correlated signal was given in accordance to Table \ref{table:MONFG_with_NE} (right), i.e., in a given episode, (L,L) was recommended with probability 0.5, or else (M, M) was recommended with the same probability.

Compared to the previous two games, we now have the opportunity to study the learning outcomes of the agents when all the considered equilibria exist. Figures~\ref{fig:game4_NE} and \ref{fig:game4_NE_states} present the results for the setting in which the agents do not receive any action recommendations. Although the action selection probabilities (Figs.~\ref{fig:game4_sCE_A1} and \ref{fig:game4_sCE_A2}) might not exhibit any regular behaviour over to considered trials, when looking at the distribution over the joint-action space in Fig. \ref{fig:game4_NE_states} more structure emerges. We notice that for about $95\%$ of the time the agents converge to one of the pure NE, described in Section~\ref{sec:NEgame} -- (L,L), (M, M) or (R, R) -- with the Pareto-dominated outcome, (R, R), having the least probability mass. This indicates that our learning algorithm that combines a ``one-shot'' vectorial Q-learning update rule with a $\varepsilon$-greedy action selection method, while allowing agents to determine their best mixed strategy by solving non-linear optimisation problems with respect to their Q-values and utility function, is quite successful in converging to NE in the case of independent learners.

When agents are able to receive action recommendations, we can notice that the selected correlated strategy is both a single-signal CE (Figs.~\ref{fig:game4_sCE} and \ref{fig:game4_sCE_states}) and a multi-signal CE (Figs.~\ref{fig:game4_mCE} and \ref{fig:game4_mCE_states}). By looking at the scalarised payoffs in Figs.~\ref{fig:game4_sCE_SER} and \ref{fig:game4_mCE_SER}, we can also notice that even in a multi-objective setting, correlated equilibria can allow one to obtain better compromises between conflicting utility functions (i.e., a SER of 14.99 for agent 1 and 5 for agent 2 in the case of single and multi-signal CE) compared to the Nash equilibrium case (i.e., SER of 13.98 for agent 1 and 4.38 for agent 2), given that the agents are able to receive a correlation signal. This empirical result demonstrates that the well known previous findings that CE can provide better payoffs than NE in single objective games (see e.g. \cite{aumann1974subjectivity}) can also apply in the more general class of multi-objective games, i.e. that in a MOMAS where a coordination signal can be established CE can potentially lead to higher individual utilities than NE.

 \begin{figure}[ht!]
    \centering
    \begin{subfigure}[b]{0.32\textwidth}
        \includegraphics[width=\textwidth]{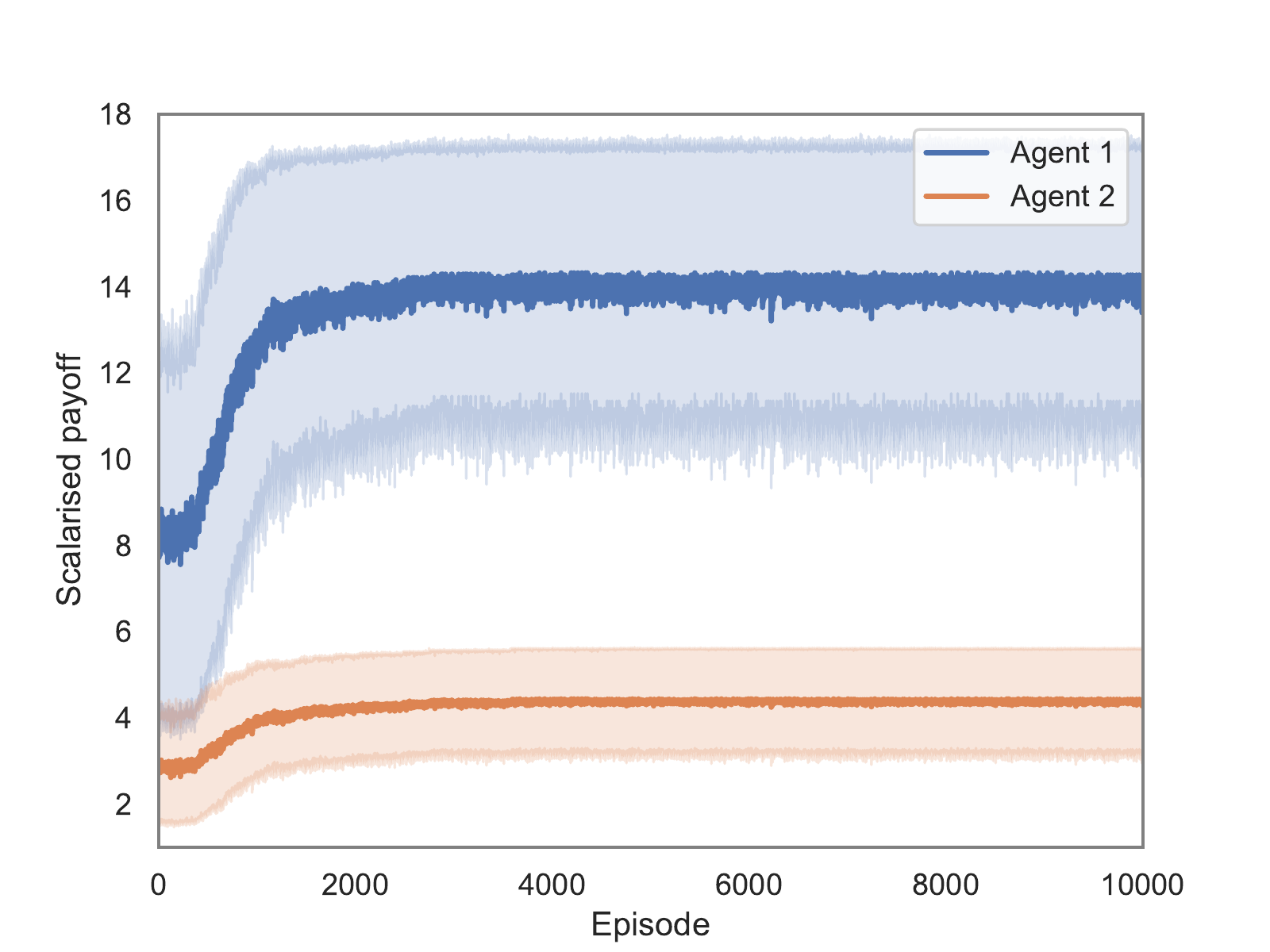}
    \caption{Scalarised payoffs obtained by each agent.}
    \label{fig:game4_NE_SER}
    \end{subfigure}
    \vspace{\baselineskip}
    \begin{subfigure}[b]{0.32\textwidth}
       \includegraphics[width=\textwidth]{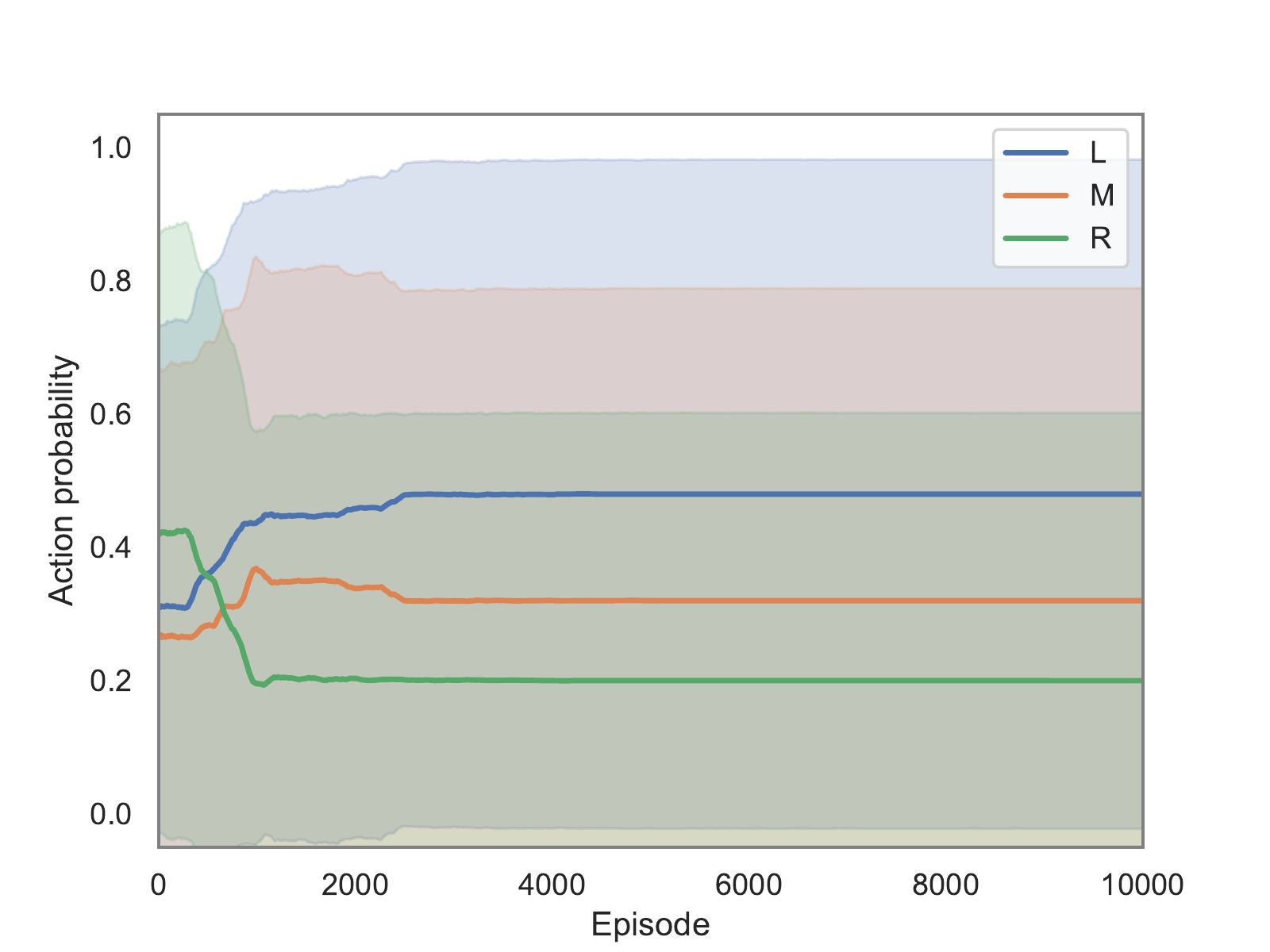}
    \caption[width=\textwidth]{Action selection probabilities of Agent 1.}
    \label{fig:game4_NE_A1}
    \end{subfigure}
    \begin{subfigure}[b]{0.32\textwidth}
        \includegraphics[width=\textwidth]{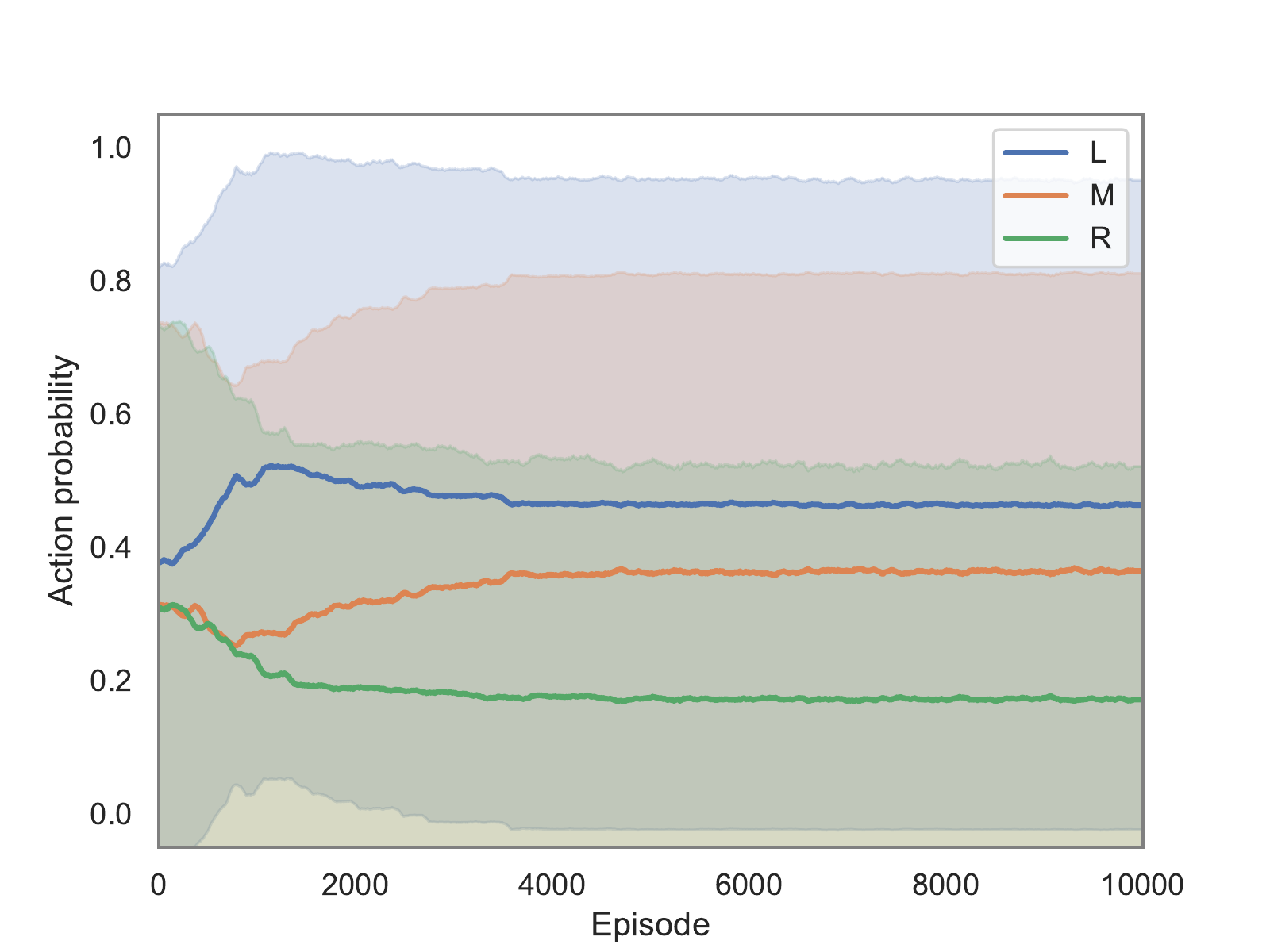}
    \caption{Action selection probabilities of Agent 2.}
    \label{fig:game4_NE_A2}
    \end{subfigure}
    \caption{Game 3 under SER with no action recommendations.}
    \label{fig:game4_NE}
\end{figure}
\begin{figure}[ht!]
    \centering
    \begin{subfigure}[b]{0.32\textwidth}
        \includegraphics[width=\textwidth]{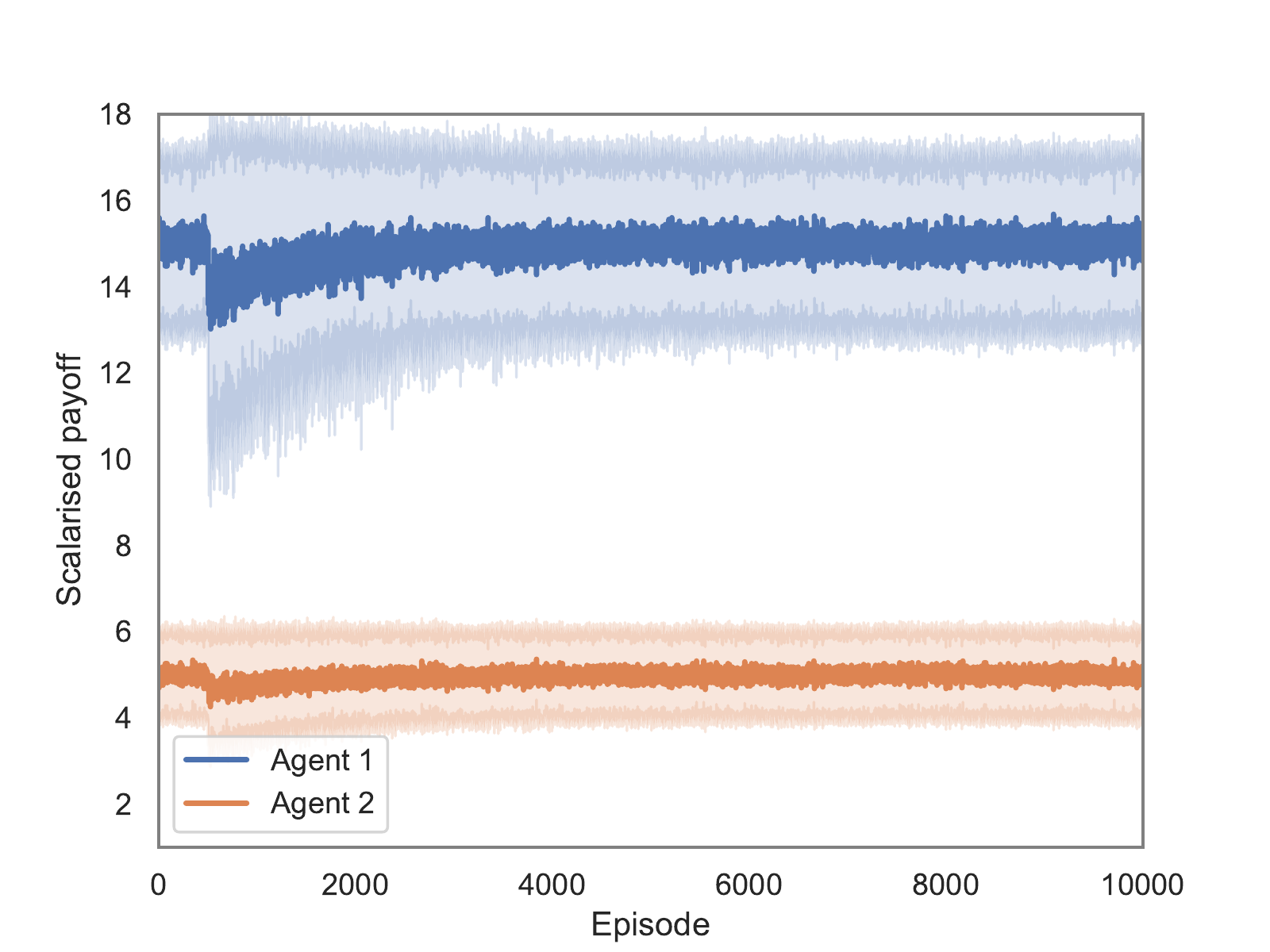}
    \caption{Scalarised payoffs obtained by each agent.}
    \label{fig:game4_sCE_SER}
    \end{subfigure}
    \begin{subfigure}[b]{0.32\textwidth}
       \includegraphics[width=\textwidth]{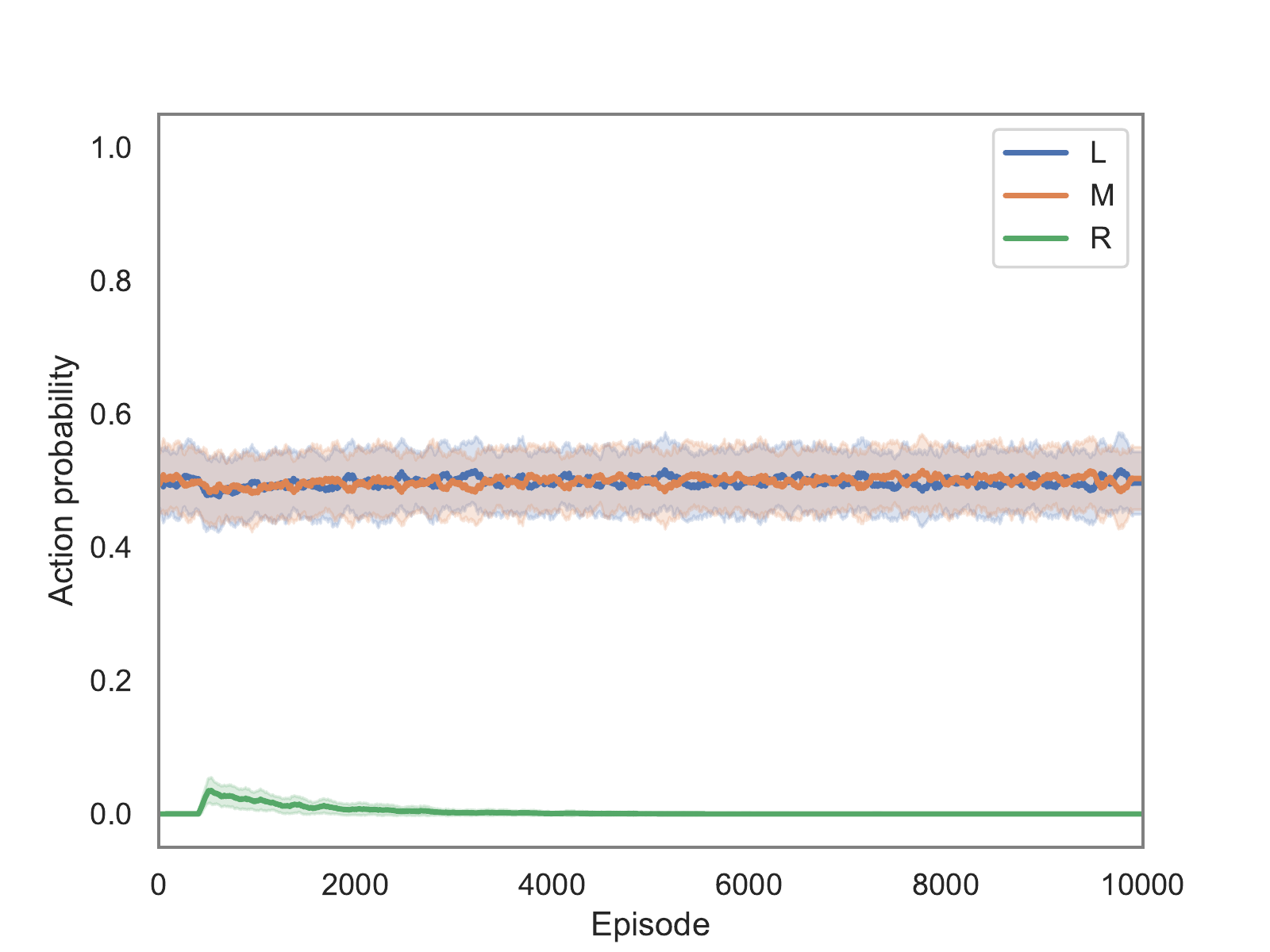}
    \caption{Action selection probabilities of Agent 1.}
    \label{fig:game4_sCE_A1}
    \end{subfigure}
    \begin{subfigure}[b]{0.32\textwidth}
        \includegraphics[width=\textwidth]{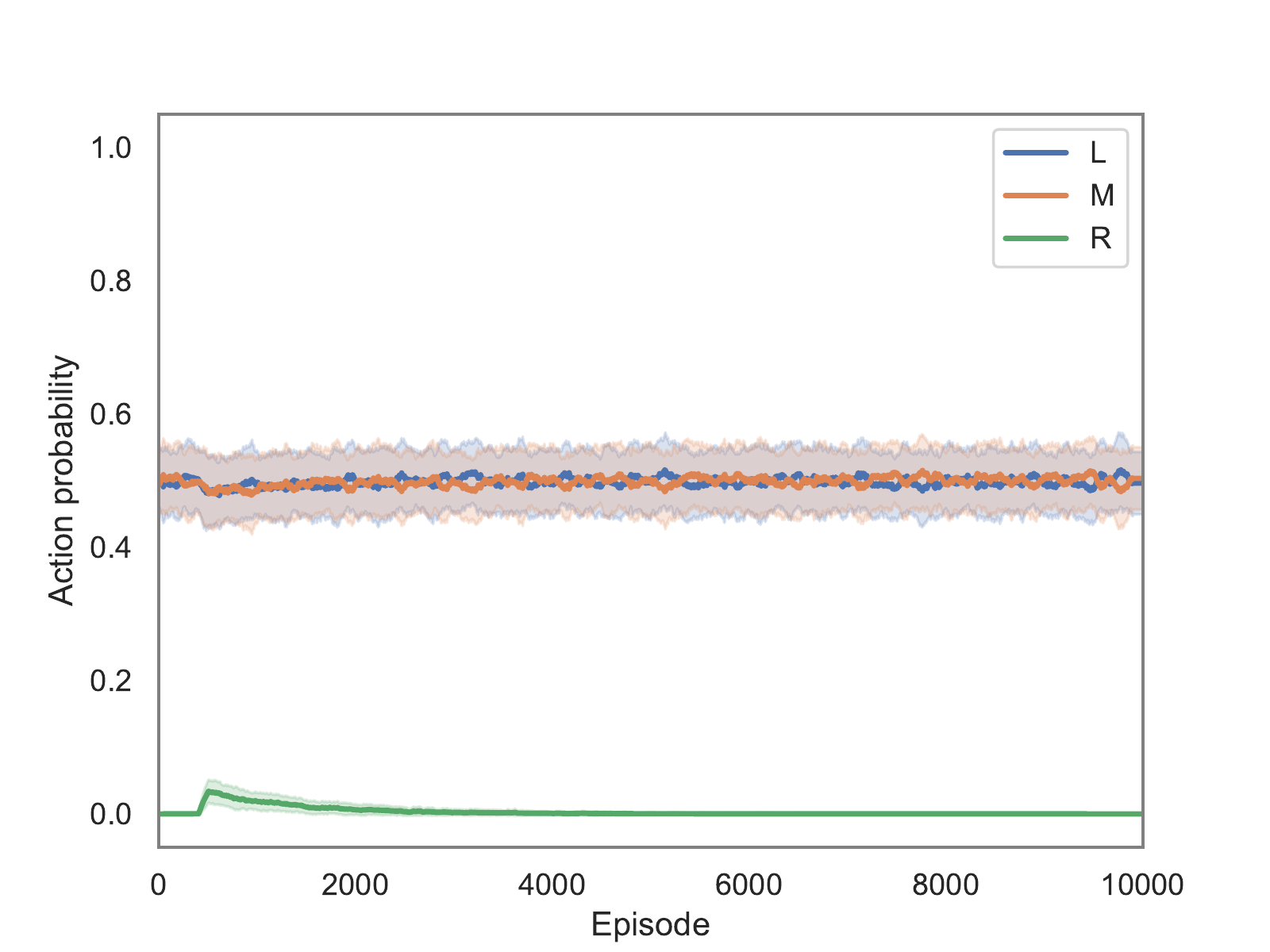}
    \caption{Action selection probabilities of Agent 2.}
    \label{fig:game4_sCE_A2}
    \end{subfigure}
    \caption{Game 3: single-signal CE under SER with action recommendations provided according to Table \ref{table:balance_ser_ce1}.}
    \label{fig:game4_sCE}
\end{figure}
\begin{figure}[ht!]
    \centering
    \begin{subfigure}[b]{0.32\textwidth}
        \includegraphics[width=\textwidth]{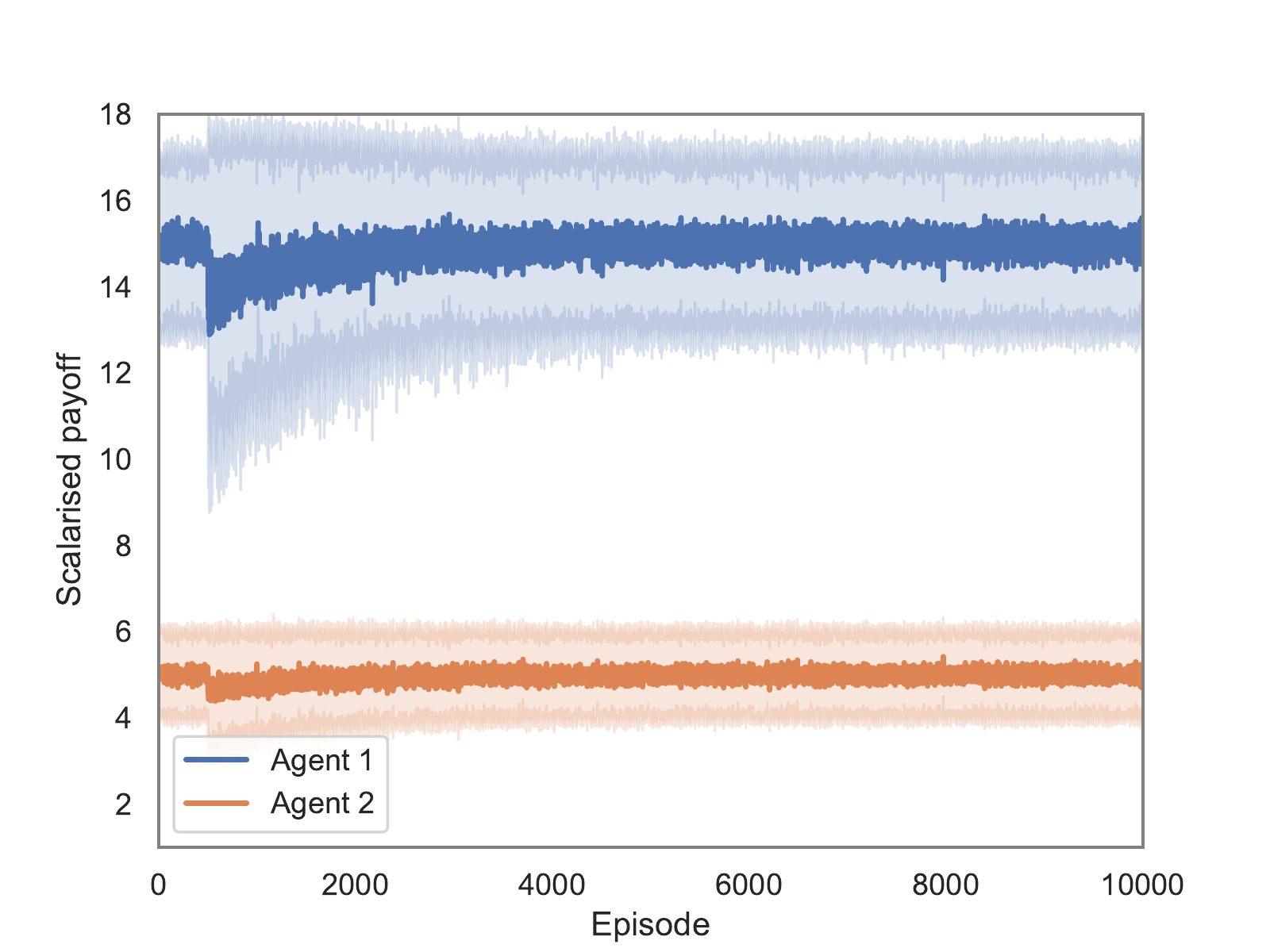}
    \caption{Scalarised payoffs obtained by each agent.}
    \label{fig:game4_mCE_SER}
    \end{subfigure}
    \begin{subfigure}[b]{0.32\textwidth}
       \includegraphics[width=\textwidth]{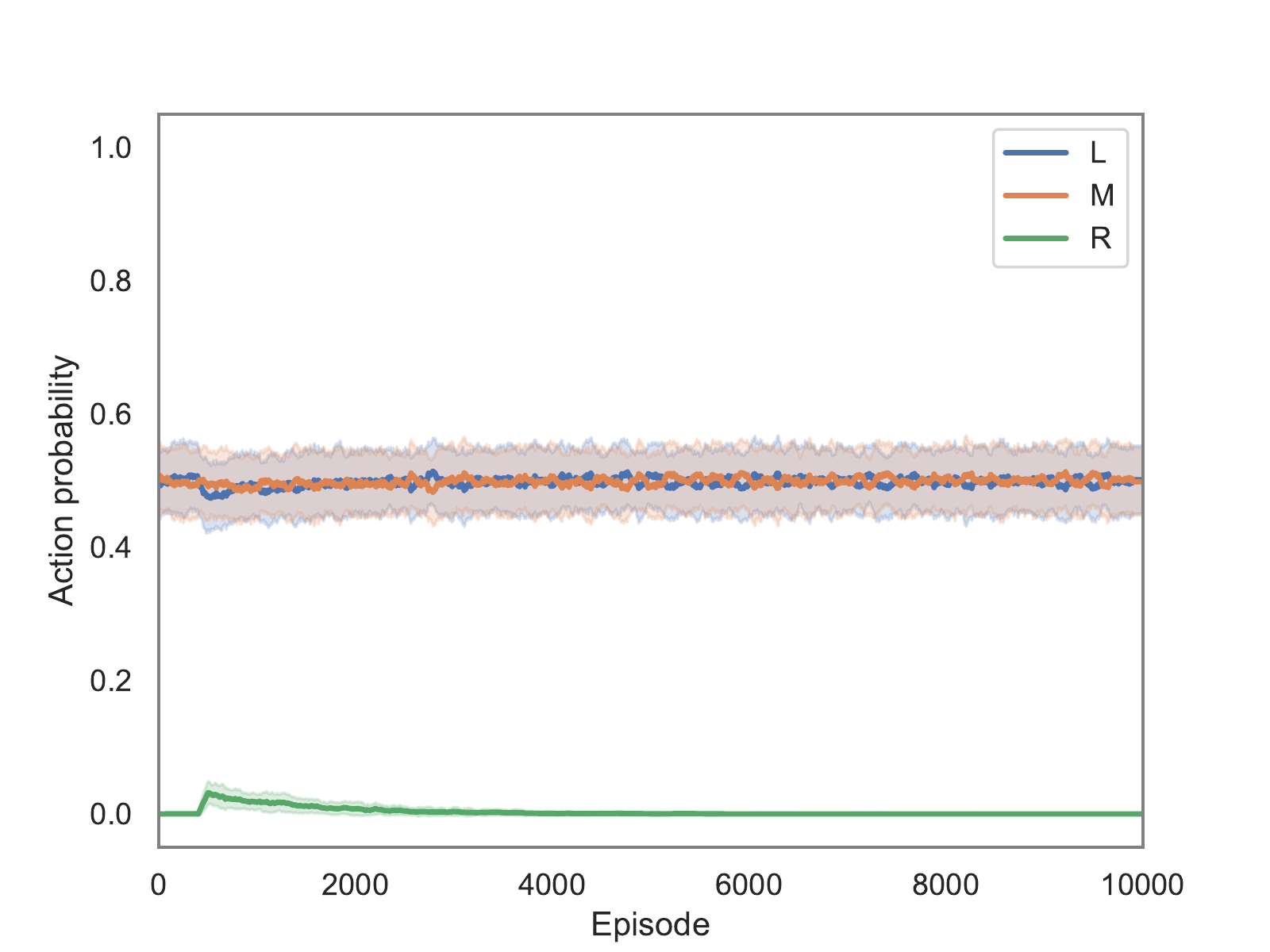}
    \caption{Action selection probabilities of Agent 1.}
    \label{fig:game4_mCE_A1}
    \end{subfigure}
    \begin{subfigure}[b]{0.32\textwidth}
        \includegraphics[width=\textwidth]{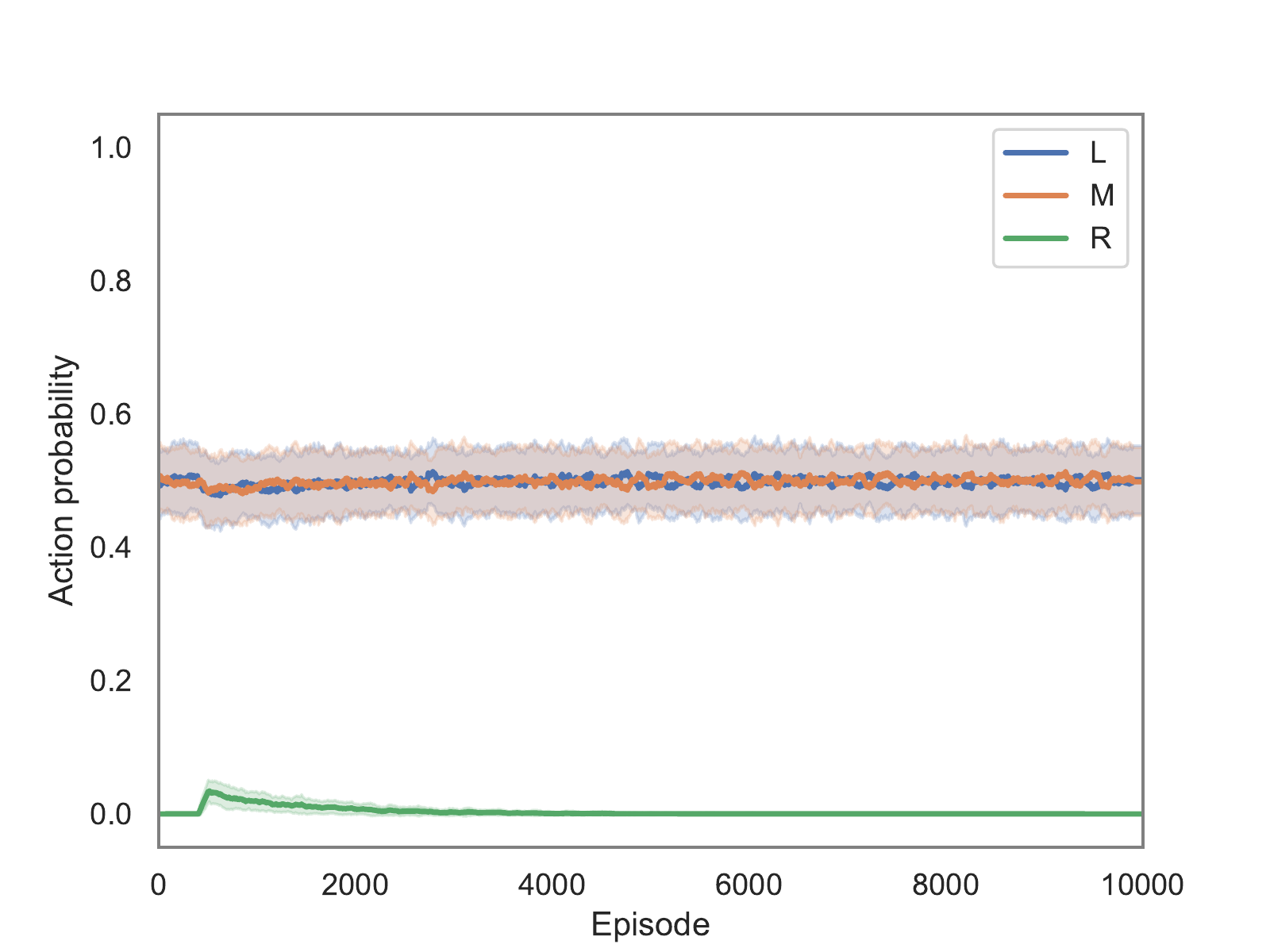}
    \caption{Action selection probabilities of Agent 2.}
    \label{fig:game4_mCE_A2}
    \end{subfigure}
    \caption{Game 3: multi-signal CE under SER with action recommendations provided according to Table \ref{table:balance_ser_ce1}.}
    \label{fig:game4_mCE}
\end{figure}
\begin{figure}[ht!]
    \centering
    \begin{subfigure}[b]{0.32\textwidth}
        \includegraphics[width=\textwidth]{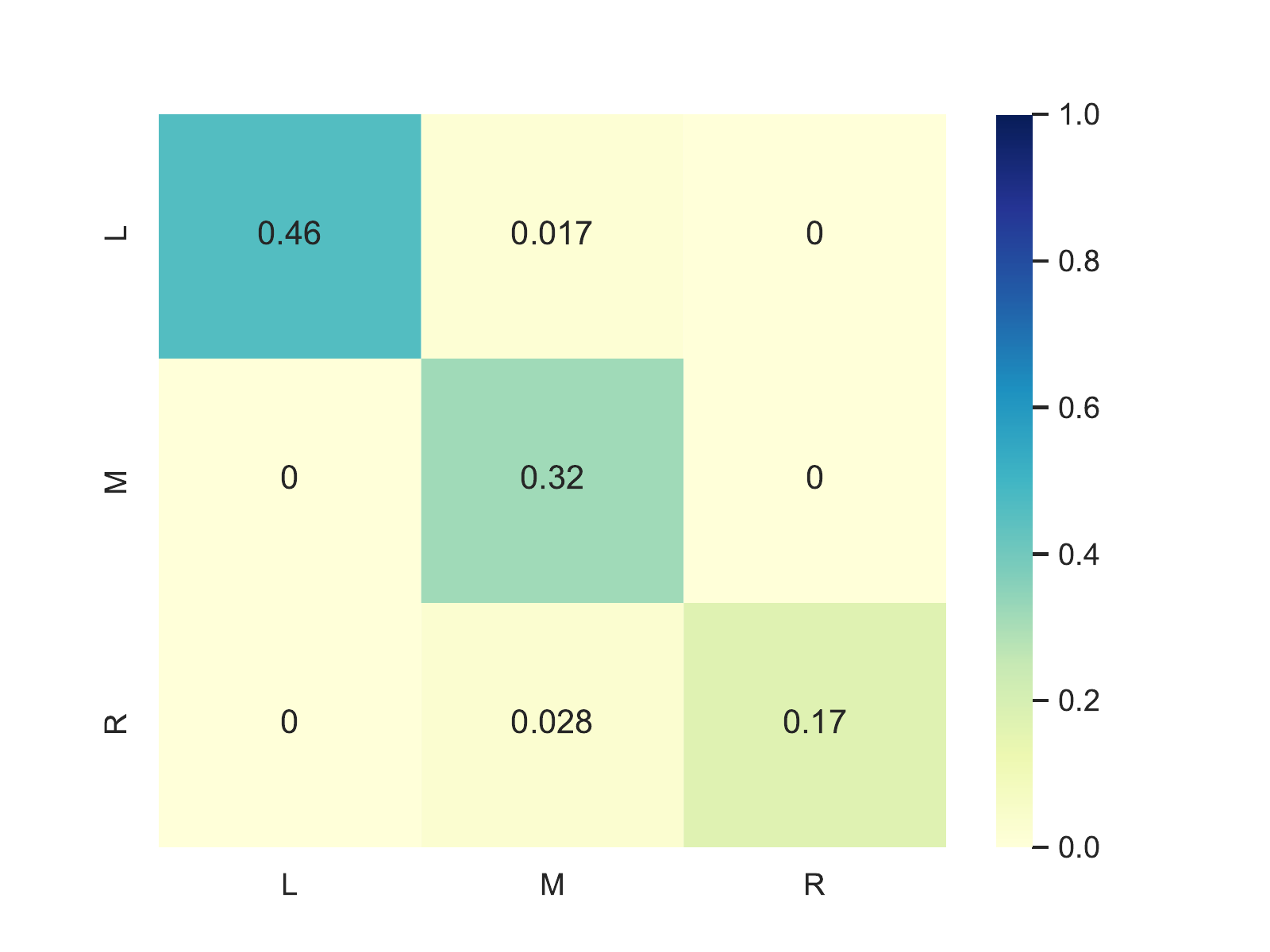}
    \caption{No action recommendations}
    \label{fig:game4_NE_states}
    \end{subfigure}
    \begin{subfigure}[b]{0.32\textwidth}
       \includegraphics[width=\textwidth]{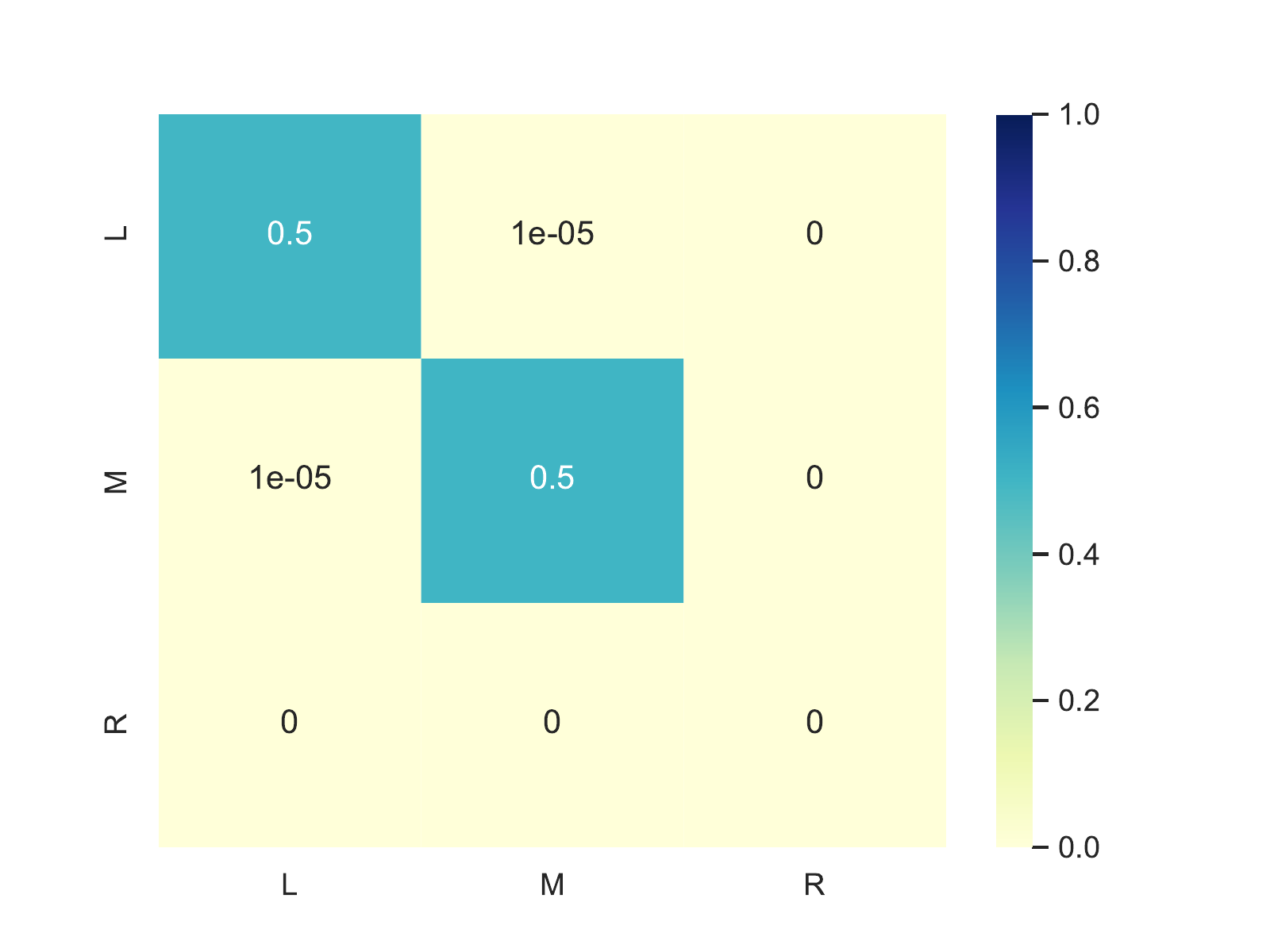}
    \caption{Single-signal CE}
    \label{fig:game4_sCE_states}
    \end{subfigure}
    \begin{subfigure}[b]{0.32\textwidth}
        \includegraphics[width=\textwidth]{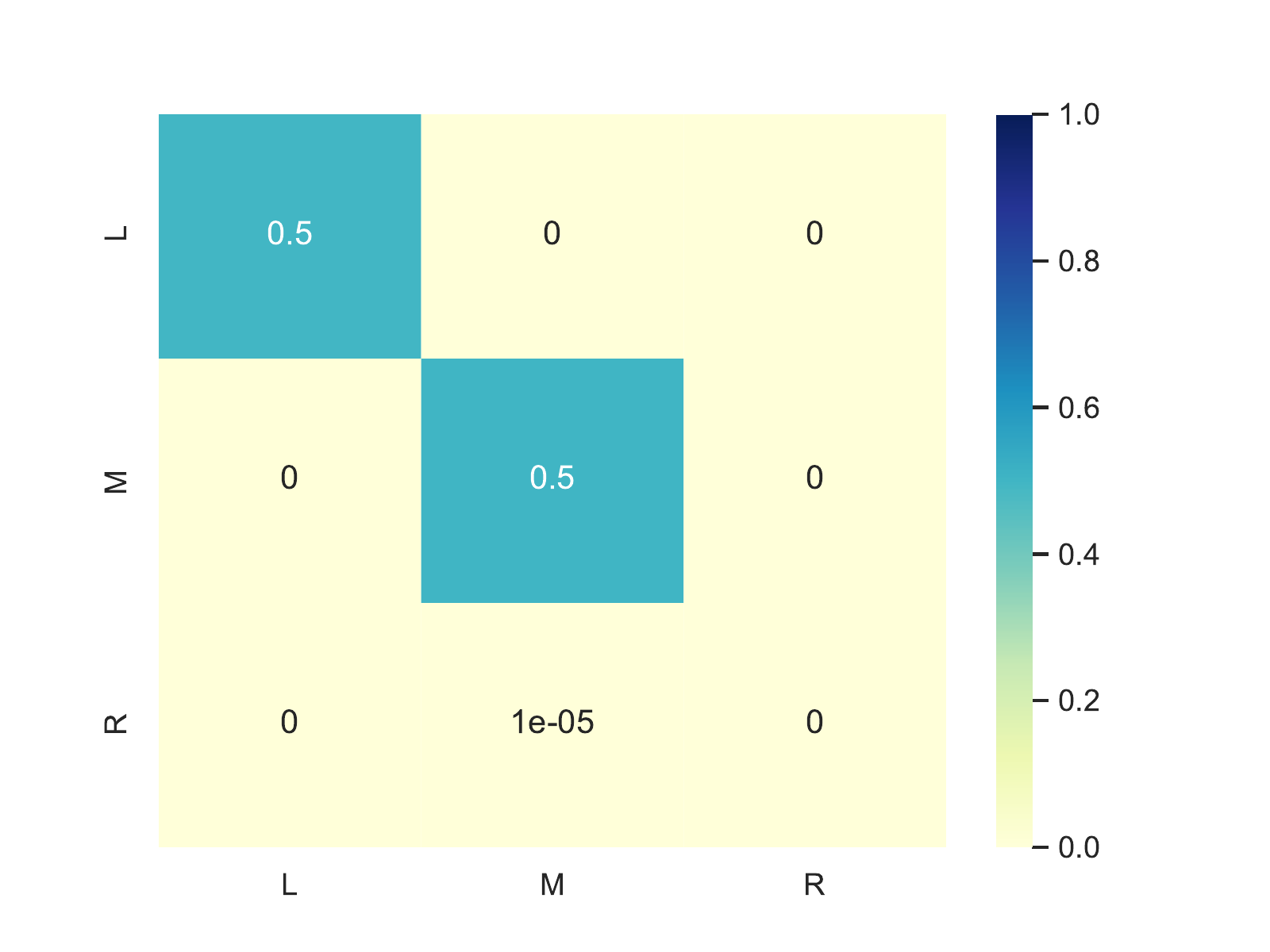}
    \caption{Multi-signal CE}
    \label{fig:game4_mCE_states}
    \end{subfigure}
    \caption{Game 3: Joint action probabilities over the last 1000 episodes under SER.}
    \label{fig:game4_states}
\end{figure}

\section{Conclusion and Future Work}
\label{sec:conclusion}
In this work, we explored the differences between two optimisation criteria for MOMAS: expected scalarised returns and scalarised expected returns. Using the framework of MONFGs, we constructed sets of conditions for the existence of Nash and correlated equilibria, two of the most commonly-used solution concepts in the single-objective MAS literature. Our analysis demonstrated that fundamental differences exist between the ESR and SER criteria in multi-agent settings.

While we have provided some theoretical results concerning the existence of equilibria in utility-based MONFGs, a number of deep and interesting open questions remain unanswered. Even though we provide examples of games where Nash equilibrium and multi-signal correlated equilibrium both exist (Table~\ref{table:MONFG_with_NE}) or do not exist (proof of Theorems \ref{th:ser_nash} and \ref{th:ser_ce}) under SER when considering non-linear utility functions, we have no concrete conclusion on how or if the relation between NE and CE modifies under SER, that is, when we can expect equilibria to exist, and when they do not;  therefore, further detailed theoretical analysis is required.

In the proof of Theorem \ref{th:ser_ce_single} we provide an example where a single-signal correlated equilibrium does exist under SER, although it not known whether single-signal correlated equilibria always exist in this setting. The existence of correlated equilibria in single-objective NFGs has been proven by \cite{hart1989existence} based on linear duality, an argument which does not rely on the existence of Nash equilibria (or by extension, the use of a fixed point theorem as per \cite{Nash1951Non}) as the original proof by \cite{aumann1974subjectivity} did. Extending the work of \cite{hart1989existence} for utility-based MONFGs under SER is a promising direction for future work. 
As we saw in the example Chicken game in Table \ref{table:chicken}, correlated equilibria allow for better compromises to be achieved between conflicting payoff functions in single-objective NFGs, when compared with Nash equilibria. In utility-based MONFGs, we demonstrated that this property translates well, allowing compromises to be achieved between conflicting utility functions (and allowing a stable compromise solution to be reached in MONFGs where no stable compromise may be reached using Nash equilibria, when conditioning on the received signal).

The analysis in this paper has a number of important limitations which should be addressed in future work. Our worked examples considered MONFGs with two agents only, so the interaction between equilibria and optimisation criteria should be further explored in larger MOMAS. It would also be worthwhile to conduct larger and more rigorous empirical studies to further expand upon our findings, and to develop new learning algorithms for MOMAS. One promising direction to allow agents to avoid solving non-linear optimisation problems when computing their optimal mixed strategy is to adopt an actor-critic approach, as per a new learning algorithm recently proposed for MONFGs by \cite{Zhang2020Opponent}. By adopting the MONFG model, we considered stateless decision making problems only; our analysis should be extended to stateful MOMAS models such as multi-objective stochastic games (MOSGs) \citep{Mannion2017Theoretical}, or even multi-objective versions of partially observable stochastic games \citep{wiggers2016structure}. We note that a similar equilibrium concept to the correlated equilibrium exists for single-objective stochastic games; the cyclic equilibrium (or cyclic correlated equilibrium) \citep{zinkevich2006cyclic}. Little is currently known about equilibria in multi-objective multi-agent sequential decision making settings. If the existence of Nash equilibria cannot be proven or demonstrated for MOSGs with non-linear utility functions under SER in the future, the cyclic equilibrium is one alternate solution concept which is worthy of exploration.

Another interesting line of future research concerns the interaction between MOMAS, optimisation criteria (ESR vs. SER) and reward shaping. Although reward shaping in MOMAS has received some attention to date (see e.g. \cite{Yliniemi2016Multi,Mannion2017Policy,Mannion2018Reward}), it has been primarily from the ESR perspective, and using linear and hypervolume scalarisation functions only. Principled reward shaping techniques such as potential-based reward shaping and difference rewards come with convenient theoretical guarantees (e.g. preserving the relative value of policies and/or actions, and therefore Nash and Pareto relations between policies and/or actions in MAS/MOMAS \citep{Devlin2011Theoretical,Mannion2017Policy,Colby2015Evolutionary,Mannion2017Theoretical}); how well these techniques will work under SER with non-linear utility functions is currently unknown.

How to best model users' utility functions for MOMAS remains a significant open question. Recent work on preference elicitation strategies for multi-objective decision support settings \citep{zintgraf2018ordered} has delivered promising results in single agent settings with non-linear utility; this approach could feasibly be extended to generate utility functions for decision making in MOMAS. It may also be beneficial for agents in MOMAS to learn opponent models; opponent modelling could potentially help to improve the utility of agents that implement it, as well as improving the probability of convergence to desirable equilibria. Initial work on opponent modelling for MONFGs under SER with non-linear utility functions \citep{Zhang2020Opponent} modelled opponent behaviour via policy reconstruction using conditional action frequencies. Modelling opponent utility functions directly, e.g. using Gaussian processes, is another promising avenue that could be explored in future research.

Finally, as we mentioned in Section \ref{sec:optimisation}, users may prefer either the SER or ESR criterion depending on their needs (e.g., whether they care more about average performance over a number of policy executions, or just the performance of a policy single execution \citep{roijers2018multi}). In larger MOMAS, it is possible that not all users would choose the same optimisation criterion, or that their preference for a specific optimisation criterion may change over time, potentially adding further complexity to the process of computing equilibria.

\bibliographystyle{plainnatCustom}

\end{document}